\documentclass[11pt,a4paper]{article}
\usepackage{fullpage}

\usepackage{amsmath,amssymb}
\usepackage{authblk}
\usepackage{color, xcolor}
\usepackage[linesnumbered,vlined,ruled]{algorithm2e}
\usepackage{graphicx}

\newtheorem{theorem}{Theorem}

\newtheorem{lemma}{Lemma}

\newtheorem{definition}{Definition}

\newcommand{\qed}{\hfill $\Box$ \medbreak}
\newenvironment{proof}{\noindent {\bf Proof.}}{\qed}

\newcommand{\dMAM}{\mathsf{dMAM}}

\newcommand{\T}{\mathbb{T}}
\newcommand{\R}{\mathbb{R}}
\newcommand{\N}{\mathbb{N}}

\newcommand{\id}{\mbox{\rm id}}

\renewcommand{\P}{\mathbb{P}}

\newcommand{\Sk}{\mathrm{Sk}}
\newcommand{\G}{\mathcal{G}}
\newcommand{\NT}{\mathbb{P}}
\newcommand{\hist}{\mathsf{h}}
\newcommand{\p}{\mathsf{pred}}
\newcommand{\s}{\mathsf{succ}}
\newcommand{\type}{\mathsf{T}}

\newcommand{\polylog}{\mathrm{polylog\,}}

\usepackage{physics}
\usepackage{tikz}
\usepackage{mathdots}
\usepackage{yhmath}
\usepackage{cancel}
\usepackage{color}
\usepackage{siunitx}
\usepackage{array}
\usepackage{multirow}
\usepackage{amssymb}
\usepackage{gensymb}
\usepackage{tabularx}
\usepackage{booktabs}
\usetikzlibrary{fadings}
\usetikzlibrary{patterns}
\usetikzlibrary{shadows.blur}
\usetikzlibrary{shapes}

\begin{document}

\title{Local Certification of Graphs with Bounded Genus}

\author[1]{Laurent Feuilloley\thanks{Additional support from MIPP and Jos\'e Correa's Amazon Research Award.}}
\author[2]{Pierre Fraigniaud\thanks{Additional support from ANR projects DESCARTES, FREDDA, and QuData. }}
\author[3]{Pedro Montealegre\thanks{Additional support from ANID via  PAI + Convocatoria Nacional Subvenci\'on a la Incorporaci\'on en la Academia A\~no 2017 + PAI77170068 and FONDECYT  11190482.}}
\author[4]{\\ Ivan Rapaport\thanks{Additional support from ANID via PIA/Apoyo a Centros Cient\'{\i}ficos y Tecnol\'ogicos de Excelencia AFB
170001 and Fondecyt 1170021.}}
\author[5]{\'Eric R\'emila\thanks{Additional support from IDEXLYON (project INDEPTH) within the Programme Investissements  d'Avenir (ANR-16-IDEX-0005) and  ``Math\'ematiques de la D\'ecision pour l'Ing\'enierie Physique et Sociale'' (MODMAD) }}
\author[6]{Ioan Todinca}

\affil[1]{\small Departamento de Ingenier\'ia Industrial, Universidad de Chile, \texttt{feuilloley@dii.uchile.cl}}
\affil[2]{IRIF, Universit\'e de Paris and CNRS, France, \texttt{pierre.fraigniaud@irif.fr} }
\affil[3]{Facultad de Ingenier\'ia y Ciencias, Universidad Adolfo Ib\'a\~nez, Santiago, Chile, \texttt{p.montealegre@uai.cl}}
\affil[4]{DIM-CMM (UMI 2807 CNRS), Universidad de Chile, \texttt{rapaport@dim.uchile.cl}}
\affil[5]{GATE Lyon St-Etienne (UMR 5824), Universit\'e  de Lyon -- UJM St-Etienne,  France, \texttt{eric.remila@univ-st-etienne.fr}}
\affil[6]{LIFO, Universit\'e d'Orl\'eans and INSA Centre-Val de Loire, France, \texttt{ioan.todinca@univ-orleans.fr}}

\date{}

\maketitle

\begin{abstract}


Naor, Parter, and Yogev [SODA 2020]~recently designed a compiler for automatically translating standard centralized interactive protocols to \emph{distributed} interactive protocols, as introduced by Kol, Oshman, and Saxena [PODC 2018]. In particular, by using this compiler, every linear-time algorithm for deciding the membership to some fixed graph class can be translated into a $\dMAM(O(\log n))$ protocol for this class, that is, a distributed interactive protocol with $O(\log n)$-bit proof size in $n$-node graphs, and three interactions between the (centralizer) computationally-unbounded but non-trustable \emph{prover} Merlin, and the (decentralized) randomized  computational\-ly-limited \emph{verifier} Arthur. As a corollary, there is a $\dMAM(O(\log n))$ protocol for the class of planar graphs, as well as for the class of graphs with bounded genus.  

We show that there exists a distributed interactive protocol for the class  of graphs with bounded genus performing just a \emph{single} interaction, from the prover to the verifier, yet preserving proof size of $O(\log n)$ bits.  This result also holds for the class of graphs with bounded \emph{demi-genus}, that is, graphs that can be embedded on a \emph{non-orientable} surface of bounded genus. The interactive protocols described in this paper are actually \emph{proof-labeling schemes}, i.e., a subclass of interactive protocols, previously introduced by Korman, Kutten, and Peleg [PODC 2005]. In particular, these schemes do \emph{not} require any randomization from the verifier, and the proofs may often be computed a priori, at low cost, by the nodes themselves. Our results thus extend the recent proof-labeling scheme for planar graphs by Feuilloley et al.~[PODC 2020], to graphs of bounded genus, and to graphs of bounded demigenus.

\end{abstract}

\vfill


\thispagestyle{empty}
\pagebreak
\setcounter{page}{1}


\section{Introduction}

\subsection{Context and Objective}

The paper considers the standard setting of distributed network computing, in which processing elements are nodes of a  network modeled as a simple connected graph $G=(V,E)$, and the nodes exchange information along the links of that network (see, e.g.,~\cite{Peleg00}). As for centralized computing, distributed algorithms often assume promises on their inputs, and many algorithms are designed for specific families of graphs, including regular graphs, planar graphs, graphs with bounded arboricity, bipartite graphs, graphs with bounded treewidth, etc. Distributed \emph{decision} refers to the problem of checking that the actual input graph (i.e., the network itself) satisfies a given predicate. One major objective of the check up is avoiding erroneous behaviors such as livelocks or deadlocks resulting from running an algorithm dedicated to a specific graph family on a graph that does not belong to this family. The decision rule typically specifies that, if the predicate is satisfied, then all nodes must accept, and otherwise at least one node must reject. A single rejecting node can indeed trigger an alarm (in, e.g., hardwired networks), or launch a recovery procedure (in, e.g., virtual networks such as overlay networks).  The main goal of distributed decision is to design efficient checking protocols, that is, protocols where every node exchange information  with nodes in its vicinity only, and where the nodes exchange a small volume of information between neighbors. 

\paragraph{Proof-Labeling Schemes.}

Some graph predicate are trivial to check locally (e.g., regular graphs), but others do not admit local decision algorithms. For instance, deciding whether the network is bipartite may require long-distance communication for detecting the presence of an odd cycle. \emph{Proof-labeling schemes}~\cite{KormanKP10} provide a remedy to this issue. These mechanisms have a flavor of NP-computation, but in the distributed setting. That is, a non-trustable oracle provides each node with a \emph{certificate}, and the collection of certificates is supposed to be a \emph{distributed proof} that the graph satisfies the given predicate. The nodes check locally the correctness of the proof. The specification of a proof-labeling scheme for a given predicate is that, if the predicate is satisfied, then there must exist a certificate assignment leading all nodes to accept, and, otherwise, for every certificate assignment, at least one node rejects.  As an example, for the case of the bipartiteness predicate, if the graph is bipartite, then an oracle can color blue the nodes of one of the partition, and color red the nodes of the other partition. It is then sufficient for each node to  locally check that all its neighbors have the same color, different from its own color, and to accept or reject accordingly. Indeed, if the graph is not bipartite, then there is no way that a dishonest oracle can fool the nodes, and make them all accept the graph.

Interestingly, the certificates provided to the nodes by the oracle can often be computed by the nodes themselves, at low cost, during some pre-computation. For instance, a spanning tree construction algorithm is usually simply asked to encode the tree $T$ locally at each node~$v$, say by a pointer $p(v)$ to the parent of $v$ in the tree. However, it is possible to ask the algorithm to also provide a distributed proof that $T$ is a spanning tree. Such a proof may be encoded distributedly by providing each node with a certificate containing, e.g., the ID of the root of~$T$, and the distance $d(v)$ from $v$ to the root (see, e.g., \cite{AfekKY97,AwerbuchPV91,ItkisL94}). Indeed, every node~$v$ but the root can  simply check that $d(p(v))=d(v)-1$ (to guarantee the absence of cycles), and that it was given the same root-ID as  all its neighbors in the network (for guaranteeing the unicity of the tree). 

\paragraph{Distributed Interactive Protocols.}

The good news is that all (Turing-decidable) predicates on graphs admit a proof-labeling scheme~\cite{KormanKP10}. The bad news is that there are simple graphs properties (e.g., existence of a non-trivial automorphism~\cite{KormanKP10}, non 3-colorability~\cite{GoosS16}, bounded diameter~\cite{Censor-HillelPP20}, etc.) which require certificates on  $\widetilde{\Omega}(n^2)$ bits  in $n$-node graphs. Such huge certificates do not fit with the requirement that checking algorithms must not only be local,  but they must also consume little bandwidth. Randomized proof-labeling schemes~\cite{FraigniaudPP19} enable to limit the bandwidth consumption, but this is often to the cost of increasing the space-complexity of the nodes. However, motivated by cloud computing, which may provide large-scale distributed systems with the ability to  interact with an external party, Kol, Oshman, and Saxena~\cite{KolOS18} introduced the notion of \emph{distributed interactive protocols}. In such protocols, a centralized non-trustable oracle with unlimited computation power (a.k.a.~Merlin) exchanges messages with a randomized distributed algorithm (a.k.a.~Arthur). Specifically, Arthur and Merlin perform a sequence of exchanges during which every node queries the oracle by sending a random bit-string, and the oracle replies to each node by sending a bit-string called \emph{proof}. Neither the random strings nor the proofs need to be the same for each node. After a certain number of rounds, every node exchange information with its neighbors in the network, and decides (i.e., it outputs accept or reject). It was proved that many predicate requiring large certificate whenever using proof-labeling schemes, including the existence of a non-trivial automorphism, have distributed interactive protocols with proofs on $O(\log n)$ bits~\cite{KolOS18}. 

Naor, Parter, and Yogev~\cite{NaorPY20} recently designed a compiler for automatically translating standard centralized interactive protocols to distributed interactive protocols. In particular, by using this compiler, every linear-time algorithm for deciding the membership to some fixed graph class can be translated into a $\dMAM(O(\log n))$ protocol, that is, a distributed interactive protocol with $O(\log n)$-bit proof size in $n$-node graphs, and three interactions between Merlin and Arthur: Merlin provides every node with a first part of the proof, on $O(\log n)$ bits, then every node challenges Merlin with a random bit-string on $O(\log n)$ bits,  and finally Merlin replies to every node with the second part of the proof, again on $O(\log n)$ bits. As a corollary, there is a $\dMAM(O(\log n))$ protocol for many graph classes, including planar graphs, graphs with bounded genus, graphs with bounded treewidth, etc. 

\paragraph{The Limits of Distributed Interactive Protocols.}

Although the compiler in~\cite{NaorPY20} is quite generic and powerful, it remains that the practical implementation of distributed interactive protocols may be complex, in particular for the ones based on several interactions between Merlin and Arthur. This raises the question of whether there exist protocols based on fewer interactions for the aforementioned classes of graphs, while keeping the proof size small (e.g., on $O(\polylog n)$ bits). Note that, with this objective in mind, proof-labeling schemes are particularly desirable as they do not require actual interactions. Indeed, as mentioned before, the certificates may often be constructed a priori by the nodes themselves. Unfortunately, under the current knowledge, establishing lower bounds on the number of interactions between the prover Merlin and the distributed verifier Arthur, as well as lower bounds on the proof size, not to speak about tradeoffs between these two complexity measures, remains challenging.  Therefore, it is not known whether $\dMAM(O(\log n))$ protocols are the best that can be achieved for graph classes such as graphs with bounded genus, or graphs with bounded treewidth.

\paragraph{Graphs with Bounded Genus.}

In this paper, we focus on the class of graphs with bounded genus, for several reasons. First, this class is among the prominent representative of sparse graphs~\cite{NesetrilO12}, and the design of fast algorithms for sparse graphs is not only of the utmost interest for centralized, but also for distributed computing  (see, e.g.,
 \cite{AmiriMRS18,BonamyGP20,CH06-esa,CHW08-disc,GhaffariH16a,GhaffariH16b,GhaffariP17,LenzenOW08,LenzenPW13, Wawrzyniak14}), 
as many real-world physical or logical networks are sparse. Second, graphs of bounded genus, including planar graphs, have attracted lots of attention recently in the distributed computing framework (see, e.g., 
\cite{AmiriSS16,AmiriSS19,GavoilleH99}), 
and it was shown that  this large class of graphs enjoys distributed exact or approximation algorithms that overcome several known lower bounds for general graphs~\cite{KuhnMW04,PelegR00,SarmaHKKNPPW12}. Last but not least, it appears that the graph classes for which proof-labeling schemes require certificates of large size are not closed under node-deletion, which yields the question of whether every graph family closed under node-deletion (in particular graph families closed under taking minors) have proof-labeling schemes with certificates of small size. This was recently shown  to be true for planar graphs~\cite{FeuilloleyFRRMT}, but the question is open beyond this class, putting aside simple classes such as bipartite graphs, forests, etc. 

As for the class of planar graphs, and for the class of graphs with bounded genus, every graph class~$\G$  that is closed under taking minors has a finite set of forbidden minors. As a consequence, as established in~\cite{FeuilloleyFRRMT}, there is a simple proof-labeling scheme with $O(\log n)$-bit certificates for $\overline{\G}$, i.e., for \emph{not} being in~$\G$. The scheme simply encodes a forbidden minor present in $G$ in a distributed manner for certifying that $G\notin\G$. Therefore, for every $k\geq 0$, there exists a simple proof-labeling scheme with $O(\log n)$-bit certificates for genus or demi-genus \emph{at least}~$k$. The difficulty is to design a proof-labeling scheme with $O(\log n)$-bit certificates for genus or demi-genus \emph{at most}~$k$. 

\subsection{Our Results}

\subsubsection{Compact Proof-Labeling Schemes for Graphs of Bounded Genus}

Recall that planar graphs are graphs embeddable on the 2-dimensional sphere $S^2$ (without edge-crossings). Graphs with genus~1 are embeddable on the torus~$\T_1$, and, more generally, graphs with genus $k\geq 0$ are embeddable on the  closed surface $\T_k$ obtained from $S^2$ by adding $k$ \emph{handles}. We show that, for every $k\geq 0$, there exists a proof-labeling scheme for the class of graphs with genus at most~$k$, using certificates on $O(\log n)$ bits. This extends a recent  proof-labeling scheme for planar graphs~\cite{FeuilloleyFRRMT} to graphs with arbitrary genus~$k\geq 0$. Note that the certificate-size of our proof-labeling schemes is optimal, in the sense there are no proof-labeling schemes using certificates on $o(\log n)$ bits, even for planarity~\cite{FeuilloleyFRRMT}.

For every~$k\geq 1$, our proof-labeling schemes also apply to the class of graphs with \emph{demi-genus}  (a.k.a., Euler genus) at most~$k$, that is, they also hold for graphs embeddable on a \emph{non-orientable} surface with genus~$k$. Graphs with demi-genus~$k$ are indeed graphs embeddable on the closed surface~$\NT_k$ obtained from $S^2$ by adding $k$ \emph{cross-caps}. For instance, graphs with demi-genus~1 are embeddable on the projective plane $\NT_1$, while graphs with demi-genus~2 are  embeddable on the Klein bottle~$\NT_2$. 

\medbreak

This paper therefore demonstrates that the ability of designing proof-labeling schemes with small certificates for planar graphs is not a coincidental byproduct of planarity, but this ability extends to much wider classes of sparse graphs closed under taking minors. This provides hints that proof-labeling schemes with small certificates can also be designed for very many (if not all) natural  classes of sparse graphs closed under vertex-deletion. 

\begin{figure}[tb]
\centering
\input{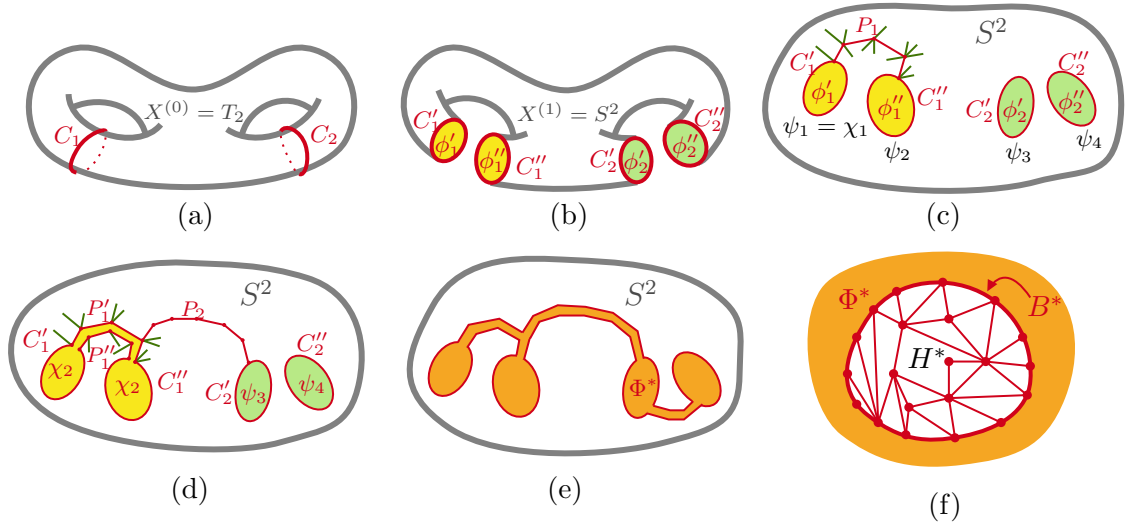}
\caption{An idealistic scenario where a graph $G$ embedded on $\T_2$ has disjoint non-separating cycles}
\label{fig:pierreT2} 
\end{figure} 

\subsubsection{Our Techniques}

Our proof-labeling schemes are obtained thanks to a local encoding of a mechanism enabling to ``unfold'' a graph~$G$ of genus or demi-genus~$k$ into a planar graph~$\widehat{G}$, by a series of vertex-duplications. Specifically, for graphs of genus~$k$, i.e., embeddable on an orientable surface~$\T_k$, we construct a sequences $G^{(0)},\dots,G^{(k)}$ where $G^{(0)}=G$, $G^{(k)}=\widehat{G}$, and, for every $i=0,\dots,k$, $G^{(i)}$ has genus~$k-i$. For $i\geq 1$, the graph $G^{(i)}$ is obtained from $G^{(i-1)}$ by identifying a non-separating cycle $C_i$ in $G^{(i-1)}$, and duplicating the vertices and cycles of $C_i$ (see Figure~\ref{fig:pierreT2}(a-b)). This enables to ``cut'' a handle of the surface $\T_{k-i+1}$, resulting in a closed surface $\T_{k-i}$ with genus one less than $\T_{k-i+1}$, while the embedding of $G^{(i-1)}$ on $\T_{k-i+1}$ induces an embedding of $G^{(i)}$ on $\T_{k-i}$. The graph $\widehat{G}$ is planar, and has $2k$ special faces $\phi'_1,\phi''_1,\dots,\phi'_k,\phi''_k$, where, for $i=0,\dots,k$, the faces  $\phi'_i$ and $\phi''_i$ results from the duplication of the face $C_i$ (see Figure~\ref{fig:pierreT2}(c)). 

The proof-labeling scheme needs  to certify not only the planarity of~$\widehat{G}$, but also the existence of the faces $\phi'_1,\phi''_1,\dots,\phi'_k,\phi''_k$, and a proof that they are indeed faces, which is non-trivial. Therefore, instead of keeping the $2k$ faces as such, we connect them by a sequence of paths $P_1,\dots,P_{2k-1}$. By duplicating each path $P_i$ into $P'_i$ and $P''_i$, the two faces $\chi$ and $\psi$ connected by a path $P_i$ is transformed into a single face, while planarity is preserved. Intuitively, the new face is the ``union'' of $\chi$, $\psi$, and the  ``piece in between'' $P'_i$ and $P''_i$ (see Figure~\ref{fig:pierreT2}(d)). The whole process eventually results in a planar graph~$H$ with a single special face $\phi$ (see Figure~\ref{fig:pierreT2}(e-f)). In fact, the paths $P_i$, $i=1,\dots,2k-1$ do not only serve the objective of merging the $2k$ faces $\phi'_1,\phi''_1,\dots,\phi'_k,\phi''_k$ into a single face~$\phi$, but also serve the objective of keeping track of consistent orientations of the boundaries of these faces. The purpose of these orientations  is to provide the nodes with the ability to locally check that the $2k$ faces can indeed be paired for forming $k$ handles. 

The planarity of $H$ and the existence of the special face $\phi$ can be certified by a slight adaptation of the proof-labeling scheme for planarity in~\cite{FeuilloleyFRRMT}. It then remains to encode the sequence of cycle and path-duplications locally, so that every node can roll back the entire process, for identify the cycles $C_i$, $i=1,\dots,k$, and the paths $P_j$, $j=1,\dots,2k-1$, and for checking their correctness. 

\begin{figure}[tb]
\centering
\input{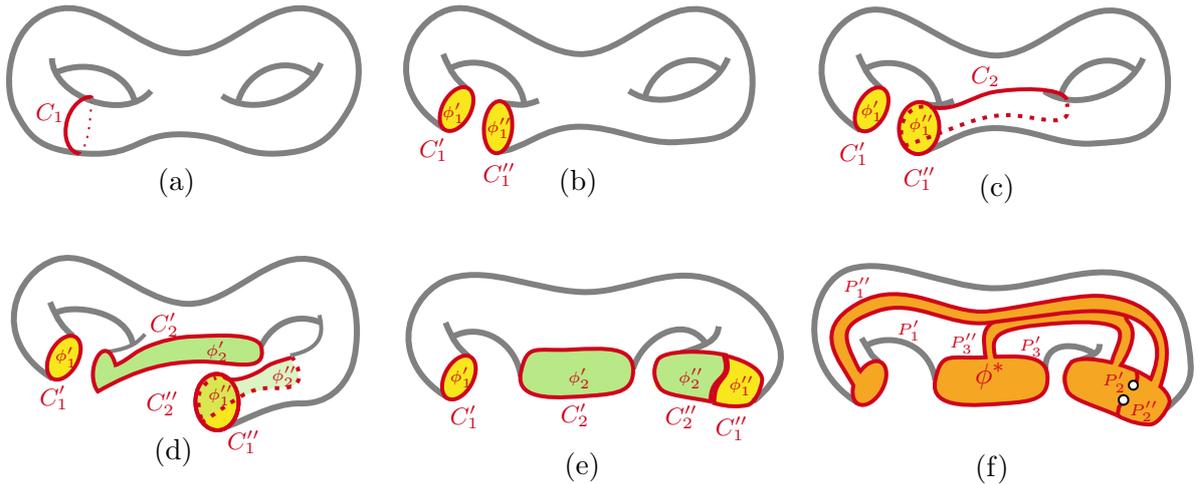}
\caption{A more complex unfolding a graph $G$ embedded on $\T_2$.}
\label{fig:pierreT2-complex} 
\end{figure} 

Among many issues, a very delicate problem is that, as opposed to cycles and paths drawn of a surface, which can be chosen to intersect at a single points, these cycles and paths are in graphs embedded on  surfaces, and thus they may intersect a lot, by sharing vertices or even edges.  Figure~\ref{fig:pierreT2} displays an idealistic scenario in which the cycles $C_i$'s are disjoint, the paths $P_j$'s are  disjoint, and these cycles and paths are also disjoint. However, this does not need to be the case, and the considered cycles and paths may mutually intersect in a very intricate manner.   For instance, Figure~\ref{fig:pierreT2-complex} displays a case in which $C_2$ intersects with $C''_1$, $P'_2$ and $P''_2$ are reduced to single vertices, and $P''_3$ intersects with $P'_1$.  It follows that the sequence of duplications may actually be quite cumbersome in general, with some nodes duplicated many times. As a consequence, keeping track of the boundaries of the faces is challenging, especially under the constraint that all information must be distributed, and stored at each node using $O(\log n)$ bits only. Also, one needs to preserve specific orientations of the boundaries of the faces, for making sure that not only the two faces $\phi'_i$ and $\phi''_i$ corresponding to a same cycle $C_i$ can be identified, but also  that they can be glued together appropriately in a way resulting to  a handle, and not be glued like, e.g., a Klein bottle. 

The case of graphs embedded on a non-orientable closed surface causes other problems, including the local encoding of the cross-caps, and the fact that decreasing the genus of a non-orientable closed surface by removing a cross-cap  may actually result in a closed surface that is orientable. Indeed, eliminating cross-caps is based on \emph{doubling} a non-orientable cycle of the graph, and this operation may result in a graph embedded on a surface that is actually orientable. (This phenomenon did not pop up in the case of orientable surfaces, as removing a handle from an orientable closed surface by cycle-duplication results in a graph embedded on an orientable closed surface.) Thus, the proof-labeling scheme for bounded demi-genus has to encode not only the identification the cross-caps, but also of faces to be identified for forming handles. 

For guaranteeing certificates on $O(\log n)$ bits, our proof-labeling schemes distribute the   information evenly to the certificates provided to the nodes, using the fact that graphs of bounded (demi-)genus have bounded \emph{degeneracy}. This property enables to store certificates on $O(\log n)$ bits at each node, even for nodes that have arbitrarily large degrees.

\subsection{Related Work}


Bounded-degree graphs form one of the most popular class of sparse graphs studied in the context of design and analysis of distributed algorithms, as witnessed by the large literature (see, e.g., \cite{Peleg00}) dedicated to construct \emph{locally checkable labelings} (e.g., vertex colorings, maximal independent sets, etc.) initiated a quarter of a century ago by the seminal work in~\cite{NaorS95}. Since then, other classes of sparse graphs  have received a lot of attention, including planar graphs, and graphs of bounded genus. In particular, there is a long history of designing distributed approximation algorithms for these classes, exemplified by the case of the minimum dominating set problem. One of the earliest  result for this latter problem is the design of a constant-factor approximation algorithm  for planar graphs, performing  in a constant number of rounds~\cite{LenzenOW08}. This result is in striking contrast with the fact that even a poly-logarithmic approximation requires at least $\Omega(\sqrt{\log n/ \log\log n})$ rounds in arbitrary $n$-node networks~\cite{KuhnMW04}. 
The paper~\cite{LenzenOW08} has paved the way for a series of works, either improving on the complexity and the approximation ratio~\cite{CHW08-disc, LenzenPW13, Wawrzyniak14}, or using weaker models~\cite{Wawrzyniak15}, or tackling more general problems~\cite{CHSWW14, CHSWW17}, or proving lower bounds~\cite{HilkeLS13,CHW08-disc}. The minimum dominating set  problem has then been studied in more general classes such as  graphs with bounded arboricity \cite{LenzenPW13}, minor-closed graphs \cite{CH06-esa}, and graphs with bounded expansion~\cite{AmiriMRS18}. Specifically, for graphs with bounded genus, it has been shown that a constant approximation can be obtained in time~$O(k)$ for graphs of genus~$k$~\cite{AmiriSS16}, and a $(1+\epsilon)$-approximation algorithm has  recently been designed, performing  in time $O(\log^*\!n)$~\cite{AmiriSS19}.

Several other problems, such as maximal independent set,  maximal matching, etc.,  have been studied for the aforementioned graph classes, and we refer to \cite{Feuilloley20} for an extended bibliography. 
In addition to the aforementioned results, mostly dealing with local algorithms, there are recent results in computational models taking into account limited link bandwidth, for graphs that can be embedded on surfaces. For instance, it was  shown that a combinatorial planar embedding can be computed efficiently in the \textsf{CONGEST} model~\cite{GhaffariH16a}. Such an embedding can then be used to derive more efficient algorithms for minimum-weight spanning tree, min-cut, and depth-first search tree constructions~\cite{GhaffariH16b,GhaffariP17}. Finally, it is worth mentioning that, in addition to algorithms, distributed data structures have been designed for graphs embedded on surfaces, including a recent  optimal adjacency-labeling for planar graphs~\cite{BonamyGP20,DujmovicEJGMM20}, and routing tables for graphs of bounded genus~\cite{GavoilleH99} as well as for graphs excluding a fixed minor~\cite{AbrahamGM05}.

Proof-labeling schemes (PLS) were introduced in~\cite{KormanKP10}, and different variants were later introduced. Stronger forms of PLS include \emph{locally checkable proofs} (LCP)~\cite{GoosS16} in which nodes forge their decisions on the certificates and on the whole states of their neighbors, and $t$-PLS~\cite{FeuilloleyFHPP18} in which nodes perform communication at distance~$t\geq 1$ before deciding. Weaker forms of PLS include  \emph{non-deterministic local decision} (NLD)~\cite{FraigniaudKP13} in which the certificates must be independent from the identity-assignment to the nodes. PLS were also extended by allowing the verifier to be randomized (see~\cite{FraigniaudPP19}). Such protocols were originally referred to as \emph{randomized PLS} (RPLS), but are nowadays referred to as distributed Merlin-Arthur (\textsf{dMA}) protocols. 

The same way NP is extended to the complexity classes forming the Polynomial Hierarchy, by alternating quantifiers, PLS were extended to a hierarchy of distributed decision classes~\cite{BalliuDFO17,FeuilloleyFH16}, which can be viewed as resulting from a game between a prover and a \emph{disprover}. Recently, \emph{distributed interactive proofs} were formalized~\cite{KolOS18}, and the classes $\mathsf{dAM}[k](f(n))$ and $\mathsf{dMA}[k](f(n))$ were defined, where $k\geq 1$ denotes the number of alternations between the centralized Merlin and the decentralized Arthur, and $f(n)$ denotes the size of the proof --- $\mathsf{dAM}[3](f(n))$ is also referred to as $\dMAM(f(n))$.  Distributed interactive protocols for problems like the existence of a non-trivial automorphism (\textsf{AUT}), and non-isomorphism ($\overline{\mathsf{ISO}}$) were designed and analyzed in~\cite{KolOS18}. The follow up paper~\cite{NaorPY20} improved the complexity of some of the protocols in~\cite{KolOS18}, either in terms of the number of interactions between the prover and the verifier, and/or in terms of the size of the certificates. A sophisticated generic way for constructing distributed IP protocols based on sequential IP protocols is presented in~\cite{NaorPY20}. One of the main outcome of this latter construction is a $\mathsf{dMAM}$ protocols using certificates on $O(\log n)$ bits for all graph classes whose membership can be decided in linear time. For other recent results on distributed interactive proof, see~\cite{CrescenziFP19,FraigniaudMORT19}.

\subsection{Organization of the Paper}

The next section provides the reader with  basic notions regarding graphs embedded on closed surfaces, and formally defines our problem. Section~\ref{sec:unfolding-a-surface} describes how to ``unfold'' a graph $G$ of genus~$k$, for producing a planar graph~$H$ with a special face~$\phi$. The section also describes how, given a planar graph~$H$ with a special face~$\phi$, one can check that $(H,\phi)$ results from the unfolding of a graph $G$ with genus~$k$.  Then, Section~\ref{sec:tools} presents our first main result, that is, a proof-labeling scheme for the class of graphs with bounded genus. In particular, it describes how to encode the description of the pair $(H,\phi)$ from Section~\ref{sec:unfolding-a-surface}, and, more importantly, how to locally encode the whole unfolding process in a distributed manner, using certificates on $O(\log n)$ bits, which allow the nodes to collectively check that their certificates form  a proof that $G$ has genus~$k$. Section~\ref{sec:PLS-demi-genus} presents our second main result, by showing how to extend the proof-labeling scheme of Section~\ref{sec:tools} to the class of graphs with bounded demi-genus. Finally, Section~\ref{sec:conclusion} concludes the paper with a discussion about the obstacles to be overcame for the design of a proof-labeling scheme for the class of graphs excluding a fixed minor.


\section{Definitions, and Formal Statement of the Problem}

This section contains a brief introduction to graphs embedded on surfaces, and provides the formal statement of our problem. 

\subsection{Closed Surfaces}

Most of the notions mentioned in this section are standard, and we refer to, e.g.,  Massey et al.~\cite{massey1991basic} for more details. 

\subsubsection{Definition} 

Recall that a \emph{topological space} is a pair $(X, T)$ where $X$ is a set, and $T$ is a topology on $X$ (e.g.,  $T$ is a collection of subsets of $X$,  whose elements are called \emph{open} sets). A topological space may  be denoted by $X$ if there is no ambiguity about the topology on~$X$. Also recall that a function $f:X\to Y$ is \emph{continuous} if the inverse image of every open set in $Y$ is open in~$X$. A \emph{homeomorphism} is a bijection that is continuous, and whose inverse is also continuous. A  \emph{topological path} in $X$ is  a  continuous  function $P: [0, 1] \to X$.   The space $X$  is \emph{ path-connected} if for any pairs $x,  y$ of points of $X$, there exists a topological path $P$ such that  $P(0) = x$ and $P(1) = y$. 

\begin{definition}
A \emph{closed surface} is   a path-connected\footnote{ Path-connected  can actually be replaced by connected (i.e., cannot be partitioned in two open sets) here,   because, under  the hypothesis of local homeomorphy to a disk, the notions of path-connectivity and connectivity are equivalent.},  compact  space
that is locally homeomorphic to a disk of~$\R^2$.
\end{definition}

\begin{figure}[tb]
\begin{center}
\input{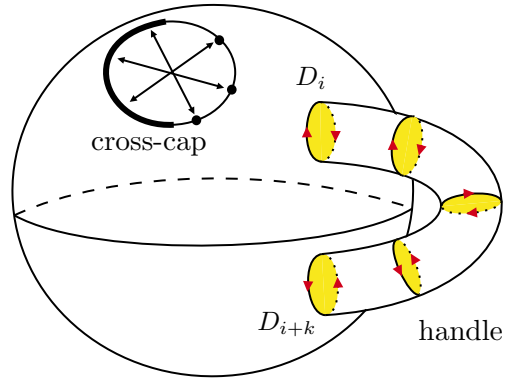}
\end{center}
\caption{Handles and cross-cap.}
\label{fig:hcc} 
\end{figure} 

\subsubsection{Construction} 
\label{sec:surface} 

Some closed surfaces  can merely be obtained by the following construction. Let $S^2$  be the 2-dimensional  sphere. For $k\geq 0$, given $2k$  disks $D_1, D_2, ...D_{2k}$ on the surface of~$S^2$, with pairwise disjoint interiors, let us direct clockwise the boundaries of $D_1, \dots, D_{k}$, and let us direct counterclockwise the boundaries of $D_{k+1},\dots,D_{2k}$. Next, let us remove the interior of each disk, and, for $1\leq i \leq k$, let us identify (i.e., glue) the boundary of $D_i$ with the boundary $D_{i+k}$ in such a way that directions coincide (see Figure~\ref{fig:hcc}). The resulting topological space is denoted by~$\T_k$. In particular, $\T_1$ is the torus, and $\T_0 = S^2$. For every~$i$, identifying $D_i$ and $D_{i+k}$  results in a \emph{handle}. It follows that $\T_k$ contains $k$ handles.

Another family of closed surfaces is constructed as follows. Let $D_1, \dots, D_{k}$ be $k\geq 1$ disks with pairwise disjoint interiors. Let us again remove the  interior of each disk. For every $1\leq i \leq k$, and for every antipodal point $v$ and $v'$ of the boundary of $D_i$, let us identify (i.e., glue) the points $v$ and $v'$ (see Figure~\ref{fig:hcc}).  The resulting topological space is  denoted by $\NT_k$. In particular, $\NT_1$ is the projective plane, and $\NT_2$ is the Klein bottle ($\NT_0$ is not defined). For every~$i$, the operation performed on $D_i$ results in a \emph{cross-cap}. It follows that $\NT_k$ contains $k$ cross-caps. 

 The surfaces resulting from the above constructions can thus be \emph{orientable} (e.g., the sphere $\T_0$ or the torus $\T_1$) or not (e.g., the projective plane $\NT_1$ or the Klein Bottle $\NT_2$), as displayed on Figure~\ref{fig:orientation}.  

\begin{figure}[tb]
\centerline{
\includegraphics[scale=0.125]{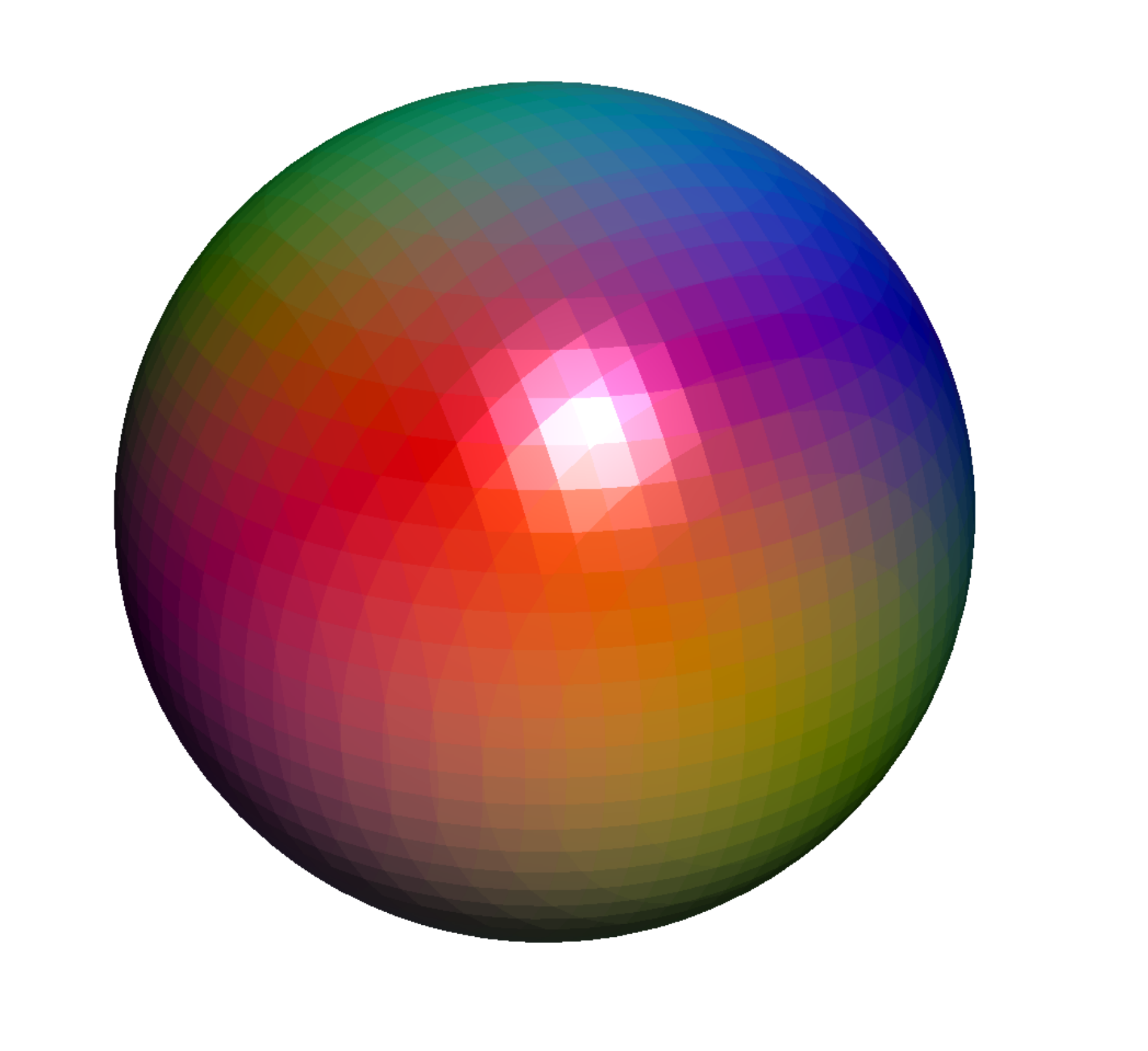}\quad
\includegraphics[scale=0.1]{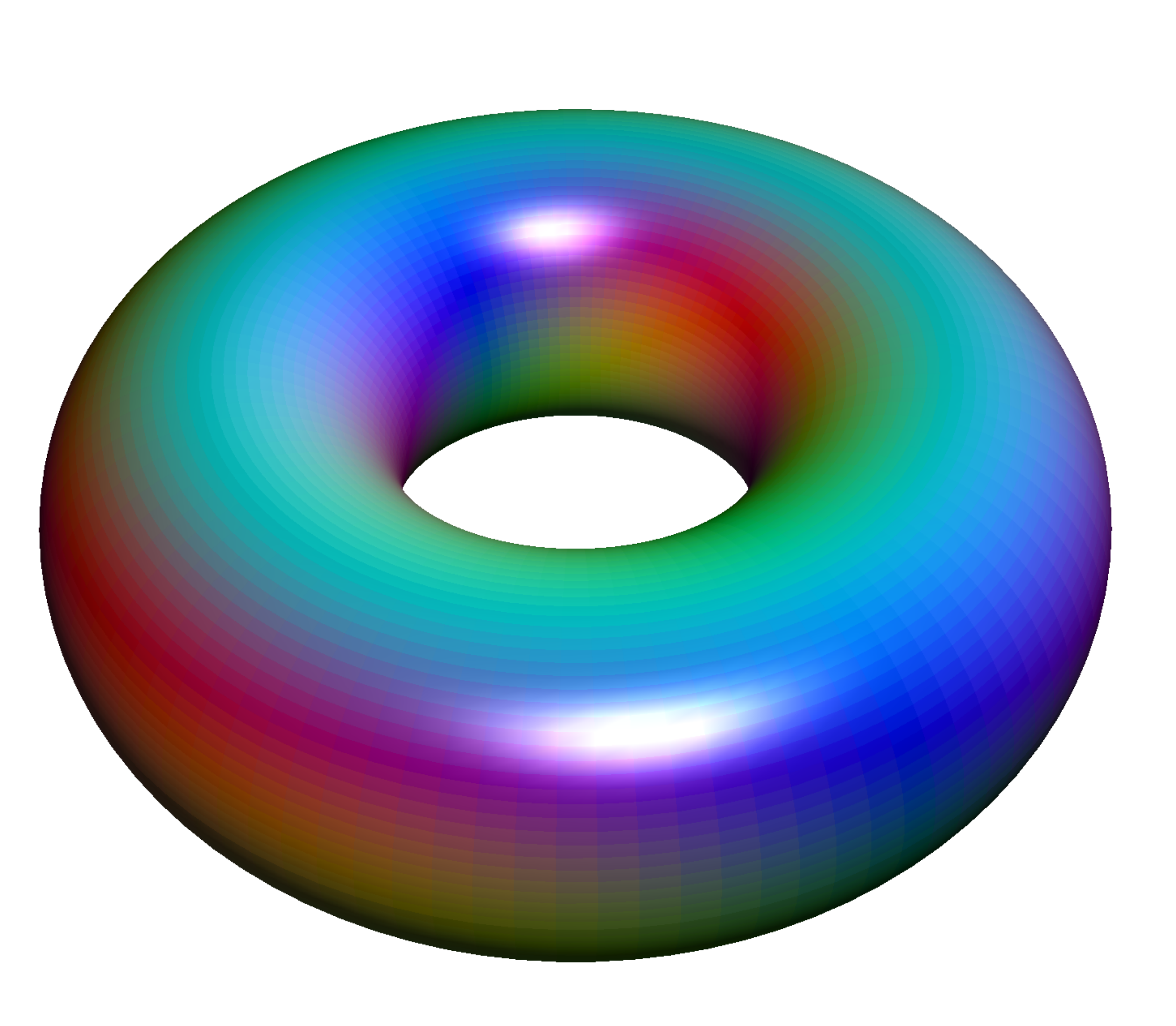}\quad
\includegraphics[scale=0.075]{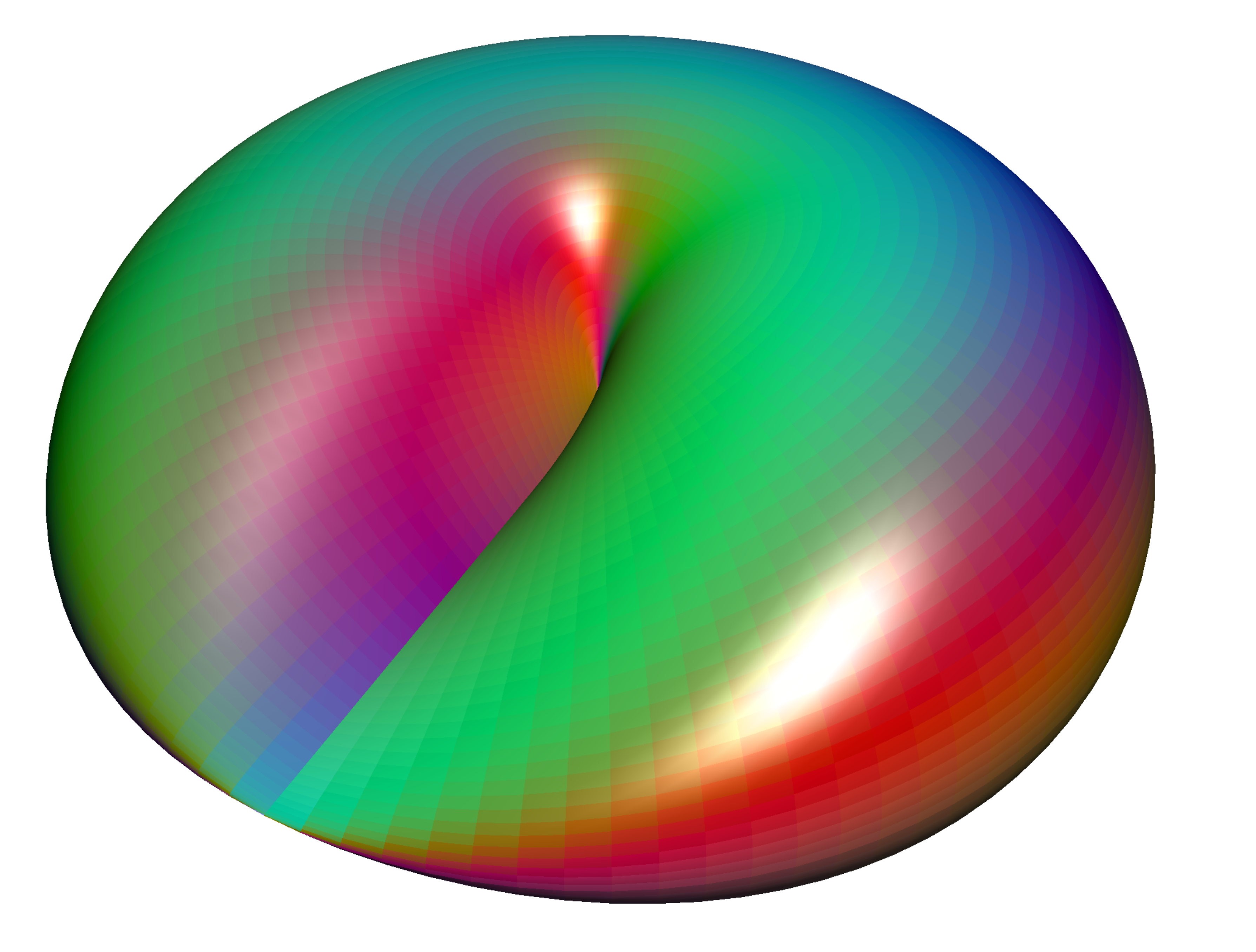}\quad
\includegraphics[scale=0.1]{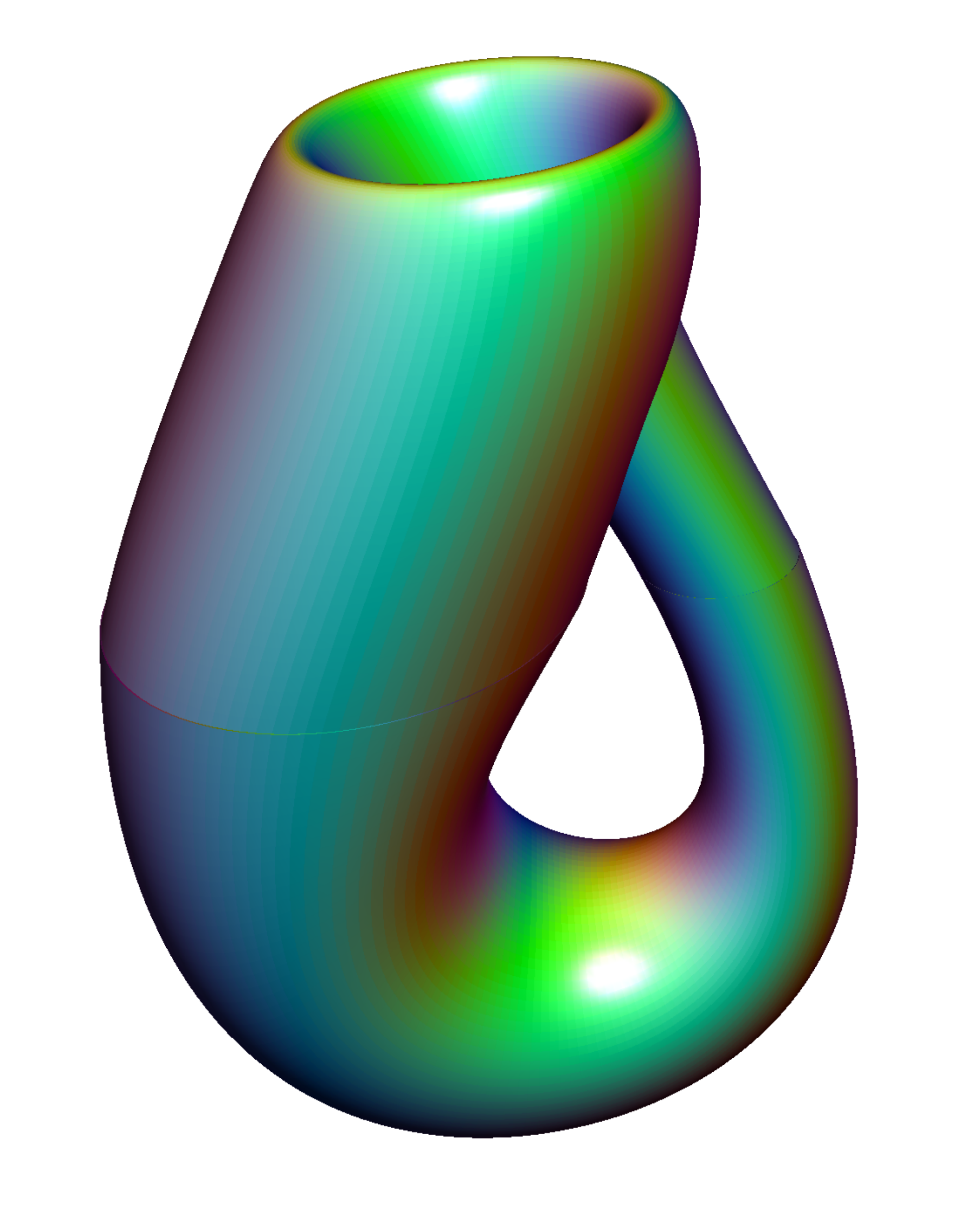}
}
\caption{The sphere, torus, projective plane, and Klein Bottle.}
\label{fig:orientation}
\end{figure} 

\subsubsection{Orientability} 

For defining orientability of a closed surface~$X$, we use the notion of  \emph{curve}, defined as a  continuous  function $C:S^1 \to X$, where $S^1$ denotes the unidimensional sphere (homeomorphic to, e.g., the trigonometric circle). A curve is \emph{simple} if it is injective.  A simple curve~$C$ is \emph{orientable} if one can define the left side and the right side of the curve at every point of the curve in a consistent manner.  Specifically, a curve~$C$ is orientable if, for every $x \in C$, there exists a neighborhood $N_x$ of $x$ such that $N_x \smallsetminus C$ has two connected components, one called the left side $L(N_x)$ of $N_x$, and the other the right side $R(N_x)$ of $N_x$, such that, for every $x,x' \in C$ and every $y\in X$, 
\[
(y \in N_x \cap N_{x'})    \land   (y \in L(N_x)) \; \Longrightarrow \; y \in L(N_{x'}).
\]
A closed surface $X$ is orientable if every simple curve of  $X$ is orientable. It is easy to check that orientability is a topological invariant. That is, if  $X$ and $Y$ are two homeomorphic topological spaces,  then $X$ is orientable if and only if $Y$ is orientable.

\subsubsection{Genus of a Surface} 

An orientable closed  surface  $X$ is of \emph{genus}~$k$ if it is homeomorphic to a closed surface~$\T_k$ constructed as in Section~\ref{sec:surface}. The Classification Theorem of orientable closed surfaces (see, e.g., \cite{Brah22}) states that every orientable closed surface has a genus. That is, for every orientable surface~$X$, there exists a unique~$k\geq 0$ such that $X$ is of genus~$k$. The fact that every pair of orientable closed surfaces with the same genus~$k$  are homeomorphic, justifies that a unique notation can be adopted for these surfaces, and any  orientable closed  surface of genus~$k$ is denoted by~$\T_k$. Observe however that two closed surfaces that are homeomorphic are not necessarily homotopic, i.e., they may not be continuously deformable one into the other other (for instance, the torus is not homotopic to the trefoil knot, although both are homeomorphic). 

The genus can also be defined for non-orientable closed surfaces. For $k\geq 1$, a non-orientable closed surface  is said to be of \emph{genus}~$k$ if it is homeomorphic to  a closed surface $\NT_k$ constructed as in Section~\ref{sec:surface}. Again, the Classification Theorem of  non-orientable closed surfaces (see, e.g., \cite{Brah22}) states that every non-orientable closed surface has a genus. That is, for every  non-orientable closed surface~$X$, there exists a unique~$k\geq 0$ such that $X$ is of genus~$k$.  As for orientable surfaces, every pair of non-orientable closed  surfaces of genus~$k$  are homeomorphic, and a non-orientable closed  surface of genus $k$ is denoted by~$\NT_k$. 

\subsection{Graphs Embedded on Surfaces} 

In this section, we recall standard notions related to graph embeddings on surfaces, and we refer to Mohar and Thomassen  \cite{mohar2001graphs} for more details. Throughout the paper, all considered graphs are supposed to be simple (no multiple edges, and no self-loops), and connected. 

\subsubsection{Topological Embeddings}

Given a graph $ G = (V, E) $, and a closed surface  $X$,  a \emph{topological embedding} of  $G$ on $X$ is given by (1)~an injective  mapping $ f : V \rightarrow X$, and, (2)~a topological path $f_e : [0, 1] \rightarrow X$ defined for every edge $e$ such that:
\begin{itemize}
\item if  $e = \{v, v'\}  \in E$,  then $f_e(\{0, 1\}) =   \{f(v), f(v')\}$, and 
\item if  $e.e'  \in E$ and  $e \neq e'$, then   $f_e(]0, 1[) \cap f _{e'}(]0, 1[)  = \varnothing$.    
\end{itemize}

The second condition is often referred to as the \emph{non-crossing} condition.  See Figure  \ref{fig:2cell} for two embeddings of the complete graph~$K_4$ on~$\T_1$. 
Throughout the paper, we may confuse a vertex~$v$ with its representation~$f(v)$, and an edge~$e$  with its representation $f_e$ (i.e., the image $f_e([0,1])$ of $[0,1]$ by $f_e$), even referred to as~$f(e)$ in the following. The set $ \cup_{e \in E} f(e) $ is called the \emph{skeleton} of the embedding, and is denoted by~$\Sk(G)$.  Each connected component of $X \smallsetminus  \Sk(G)$ is an open set of $X$ (as complement of a closed set), called a \emph{face} of the embedding. In fact, in this paper, we will abuse notation, and often refer to~$G$ instead of~$\Sk(G)$ when referring to the embedding of $G$ on~$X$. 

\subsubsection{2-Cell Embeddings}

We now recall a slightly more sophisticated, but significantly richer form of topological embedding, called \emph{2-cell embedding}. A 2-cell embedding is a topological embedding such that every face is homeomorphic to an open disk of~$\R^2$. 

In a 2-cell embedding of a graph~$G$, the border of a  face can  be described by giving a so-called \emph{boundary (closed) walk}, that is, an ordered list $(v_0, \dots, v_r)$ of non-necessarily distinct vertices of~$G$, where, for $i=0,\dots,r-1$, $\{v_i,v_{i+1}\}\in E(G)$, and $\{v_r,v_0\}\in E(G)$.  The vertices and   edges of a face are the images by the embedding of the vertices and edges of the boundary walk. The  boundary  walk is however not necessarily a simple cycle, as  an edge may appear twice in the walk, once for each direction, and a  vertex may even appear many times. 

\begin{figure}
    \centering
\tikzset{every picture/.style={line width=0.75pt}} 

\begin{tikzpicture}[x=0.75pt,y=0.75pt,yscale=-1,xscale=1]

\draw  [color={rgb, 255:red, 128; green, 128; blue, 128 }  ,draw opacity=1 ][line width=1.5]  (134.25,134.38) .. controls (134.25,98.82) and (180.31,70) .. (237.13,70) .. controls (293.94,70) and (340,98.82) .. (340,134.38) .. controls (340,169.93) and (293.94,198.75) .. (237.13,198.75) .. controls (180.31,198.75) and (134.25,169.93) .. (134.25,134.38) -- cycle ;
\draw [color={rgb, 255:red, 128; green, 128; blue, 128 }  ,draw opacity=1 ][line width=1.5]    (192.75,128.67) .. controls (219,152.92) and (254.75,150.17) .. (291,126.42) ;
\draw [color={rgb, 255:red, 128; green, 128; blue, 128 }  ,draw opacity=1 ][line width=1.5]    (204.33,137.92) .. controls (227.92,117.67) and (250.58,117.33) .. (273.33,135.92) ;
\draw [line width=0.75]    (218.58,175.67) -- (255.17,180.42) ;
\draw [shift={(255.17,180.42)}, rotate = 7.4] [color={rgb, 255:red, 0; green, 0; blue, 0 }  ][fill={rgb, 255:red, 0; green, 0; blue, 0 }  ][line width=0.75]      (0, 0) circle [x radius= 2.34, y radius= 2.34]   ;
\draw [shift={(218.58,175.67)}, rotate = 7.4] [color={rgb, 255:red, 0; green, 0; blue, 0 }  ][fill={rgb, 255:red, 0; green, 0; blue, 0 }  ][line width=0.75]      (0, 0) circle [x radius= 2.34, y radius= 2.34]   ;
\draw [line width=0.75]    (255.17,180.42) -- (281.58,169.33) ;
\draw [shift={(281.58,169.33)}, rotate = 337.24] [color={rgb, 255:red, 0; green, 0; blue, 0 }  ][fill={rgb, 255:red, 0; green, 0; blue, 0 }  ][line width=0.75]      (0, 0) circle [x radius= 2.34, y radius= 2.34]   ;
\draw [line width=0.75]    (218.58,175.67) -- (247.92,158) ;
\draw [line width=0.75]    (281.58,169.33) -- (247.92,158) ;
\draw [line width=0.75]    (255.17,180.42) -- (247.92,158) ;
\draw [shift={(247.92,158)}, rotate = 252.08] [color={rgb, 255:red, 0; green, 0; blue, 0 }  ][fill={rgb, 255:red, 0; green, 0; blue, 0 }  ][line width=0.75]      (0, 0) circle [x radius= 2.34, y radius= 2.34]   ;
\draw [line width=0.75]    (218.58,175.67) .. controls (177.58,165.33) and (150.25,143) .. (162.58,117) .. controls (174.92,91) and (214.25,82.67) .. (238.25,83.33) .. controls (262.25,84) and (304.92,92) .. (314.25,114.33) .. controls (323.58,136.67) and (308.58,156.33) .. (281.58,169.33) ;
\draw  [color={rgb, 255:red, 128; green, 128; blue, 128 }  ,draw opacity=1 ][line width=1.5]  (394.25,134.38) .. controls (394.25,98.82) and (440.31,70) .. (497.13,70) .. controls (553.94,70) and (600,98.82) .. (600,134.38) .. controls (600,169.93) and (553.94,198.75) .. (497.13,198.75) .. controls (440.31,198.75) and (394.25,169.93) .. (394.25,134.38) -- cycle ;
\draw [color={rgb, 255:red, 128; green, 128; blue, 128 }  ,draw opacity=1 ][line width=1.5]    (452.75,128.67) .. controls (479,152.92) and (514.75,150.17) .. (551,126.42) ;
\draw [color={rgb, 255:red, 128; green, 128; blue, 128 }  ,draw opacity=1 ][line width=1.5]    (464.33,137.92) .. controls (487.92,117.67) and (510.58,117.33) .. (533.33,135.92) ;
\draw [line width=0.75]    (478.58,175.67) -- (515.17,180.42) ;
\draw [shift={(515.17,180.42)}, rotate = 7.4] [color={rgb, 255:red, 0; green, 0; blue, 0 }  ][fill={rgb, 255:red, 0; green, 0; blue, 0 }  ][line width=0.75]      (0, 0) circle [x radius= 2.34, y radius= 2.34]   ;
\draw [shift={(478.58,175.67)}, rotate = 7.4] [color={rgb, 255:red, 0; green, 0; blue, 0 }  ][fill={rgb, 255:red, 0; green, 0; blue, 0 }  ][line width=0.75]      (0, 0) circle [x radius= 2.34, y radius= 2.34]   ;
\draw [line width=0.75]    (515.17,180.42) -- (541.58,169.33) ;
\draw [shift={(541.58,169.33)}, rotate = 337.24] [color={rgb, 255:red, 0; green, 0; blue, 0 }  ][fill={rgb, 255:red, 0; green, 0; blue, 0 }  ][line width=0.75]      (0, 0) circle [x radius= 2.34, y radius= 2.34]   ;
\draw [line width=0.75]    (541.58,169.33) -- (507.92,158) ;
\draw [line width=0.75]    (515.17,180.42) -- (507.92,158) ;
\draw [shift={(507.92,158)}, rotate = 252.08] [color={rgb, 255:red, 0; green, 0; blue, 0 }  ][fill={rgb, 255:red, 0; green, 0; blue, 0 }  ][line width=0.75]      (0, 0) circle [x radius= 2.34, y radius= 2.34]   ;
\draw [line width=0.75]    (478.58,175.67) .. controls (437.58,165.33) and (410.25,143) .. (422.58,117) .. controls (434.92,91) and (474.25,82.67) .. (498.25,83.33) .. controls (522.25,84) and (564.92,92) .. (574.25,114.33) .. controls (583.58,136.67) and (568.58,156.33) .. (541.58,169.33) ;
\draw    (507.92,158) .. controls (508.17,145.5) and (502.17,141.5) .. (500,145) ;
\draw    (490,199) .. controls (478.83,194.5) and (473.17,185.83) .. (478.58,175.67) ;
\draw  [dash pattern={on 0.84pt off 2.51pt}]  (500,145) .. controls (484.83,162.17) and (504.83,191.83) .. (490,199) ;

\draw (254.75,181.9) node [anchor=north west][inner sep=0.75pt]  [font=\footnotesize]  {$a$};
\draw (251.75,145.57) node [anchor=north west][inner sep=0.75pt]  [font=\footnotesize]  {$b$};
\draw (283.58,172.73) node [anchor=north west][inner sep=0.75pt]  [font=\footnotesize]  {$c$};
\draw (211,158.4) node [anchor=north west][inner sep=0.75pt]  [font=\footnotesize]  {$d$};
\draw (514.75,181.9) node [anchor=north west][inner sep=0.75pt]  [font=\footnotesize]  {$a$};
\draw (511.75,145.57) node [anchor=north west][inner sep=0.75pt]  [font=\footnotesize]  {$b$};
\draw (543.58,172.73) node [anchor=north west][inner sep=0.75pt]  [font=\footnotesize]  {$c$};
\draw (477,156.4) node [anchor=north west][inner sep=0.75pt]  [font=\footnotesize]  {$d$};

\end{tikzpicture}

    \caption{Two embeddings of  $K_4$ on the torus $\T_1$. }
    \label{fig:2cell}
 \end{figure}
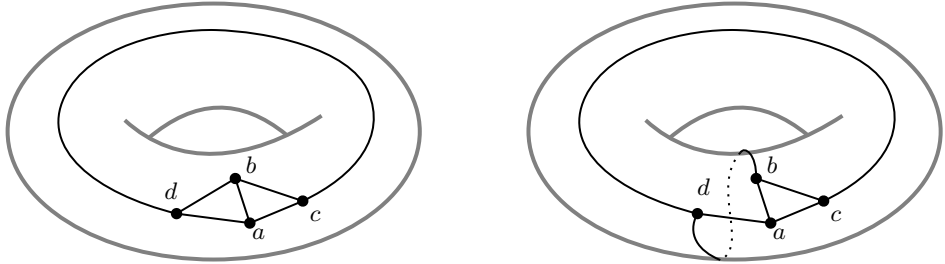
 
For instance, Figure~\ref{fig:2cell} displays two embeddings of the complete graph $K_4$ on the torus~$\T_1$. The embedding on the left is not a 2-cell embedding. Indeed, this embedding results in three faces, including the two faces with boundary walk $(a,b,c)$ and $(a,b,d)$. The third face is however not homeomorphic to an open disk (there is a hole in it, resulting from the hole in the torus).  On the other hand, the embedding on the right in Figure  \ref{fig:2cell} is a 2-cell embedding. Indeed, there are two faces, including the face with boundary walk $(a,b,c)$. The other face is also homeomorphic to an open disk. A boundary walk of this latter face is $(d, a, b, d, c, a, d,  b, c)$. This can be seen by starting from $d$, traversing the edge $\{d,a\}$, and adopting the ``left-hand rule'' when entering a vertex, leading from $a$ to $b$, then back to $d$, next to~$c$, etc.  Notice that this  boundary walk uses some edges twice. It follows that the closure of a face is not necessarily homeomorphic to a closed disk, even in a 2-cell embedding.  

We complete the section with an observation, which  allows us to restrict our attention to cycles in graphs instead of arbitrary curves in topological spaces. It also illustrates the interest of 2-cell embeddings (the result does not necessarily hold for arbitrary embeddings, as illustrated by the embedding on the left of Figure~\ref{fig:2cell}).

\begin{lemma}  \label{lem:simul}
For every graph $G$, and every closed surface~$X$, any 2-cell embedding  of $G$ on $X$ satisfies that 
every closed curve in $X$ is either contractible, or homotopic to a closed cycle of~$\Sk(G)$. 
\end{lemma}
 
The rough reason why the result holds is that, in a 2-cell embedding, any sub-path of a path traversing a face can be replaced by a sub-path following the border of the face. (This is not necessarily true for a general embedding). 

\subsubsection{Genus and Demigenus of a Graph}

For any graph $G$, there exists $k\geq 0$ such that $G$ can be embedded on $\T_k$, as any embedding of $G$ in the plane with $x$ pairs of crossing edges induces an embedding of $G$ on $\T_x$ without crossings, by replacing each crossing with a handle. Also, if $G$ can be embedded on $\T_k$, then $G$ can be embedded on $\T_{k'}$ for every $k' \geq k$. The \emph{genus} of a graph $G$ is the smallest $k$ such that there exists an embedding of $G$ on $\T_k$. Similarly, the \emph{non-orientable} genus,  a.k.a.~\emph{demigenus}, or \emph{Euler genus} of~$G$, is defined as the smallest~$k$ such that there exists an embedding of $G$ on $\NT_k$. 

The embeddings of graphs of genus~$k$ on~$\T_k$  have a remarkable property (see, e.g.,~\cite{Yo63}).

\begin{lemma} \label{lem:if-genusk-then-2cell}
Every embedding of a graph  $G$ of genus $k$  on $\T_k$ is a 2-cell embedding. 
\end{lemma}

The same property does not necessarily hold fo graphs with bounded demi-genus. However, some weaker form of Lemma~\ref{lem:if-genusk-then-2cell} can be established (see, e.g., \cite{PPPV87}). 

\begin{lemma} \label{lem:if-demigenusk-then-2cell} 
For every  graph  $G$ of demigenus $k$, there exists a 2-cell embedding of $G$ on $\NT_k$.  
\end{lemma}

The next result  is extremely helpful for computing the genus of a graph, and is often referred to as the Euler-Poincar\'e formula~\cite{poincare}. 

\begin{lemma}
\label{euler}
Let $G = (V, E)$, and let $X$ be a closed surface of genus~$k$. Let us consider any 2-cell embedding of $G$ on $X$,  and let $F$ be the set of faces of this embedding. If $X$ is orientable then $ \vert V \vert -  \vert E \vert  +  \vert F  \vert =  2 - 2k$. If $X$ is non orientable then $ \vert V \vert -  \vert E \vert  +  \vert F  \vert =  2 - k$.
\end{lemma}

Recall that, for $d\geq 0$, a graph $G$ is \emph{$d$-degenerate} if every subgraph of $G$ has a node of degree at most~$d$. Degeneracy will play a crucial role later in the paper, for evenly distributing the information to be stored in the certificates according to our proof-labeling schemes. Graphs with bounded genus have bounded degeneracy (see, e.g., \cite{NesetrilO12}), as recalled below for further references. 

\begin{lemma}\label{genus-implies-degeneracy}
For every $k\geq 0$, there exists $d\geq 0$ such that every graph of genus at most $k$ is $d$-degenerate. Similarly, For every $k\geq 1$, there exists $d\geq 0$ such that every graph of demigenus at most $k$ is $d$-degenerate.
\end{lemma}


\subsection{Formal Statement of the Problem}

Proof-Labeling Schemes (PLS) are distributed mechanisms for verifying graph properties. More precisely, let $\G$ be a graph family. A PLS for $\G$ is defined as a prover-verifier pair $(\mathbf{p},\mathbf{v})$, bounded to satisfy the following. Given any graph $G=(V,E)$ whose $n$ vertices are arbitrarily labeled by $n$ distinct identifiers (ID) picked from a set  $\{1,\dots,n^k\}$, $k\geq 1$, of polynomial range, the prover~$\mathbf{p}$ is a non-trustable oracle that provides every vertex $v\in V$ with a \emph{certificate}~$c(v)$. The verifier~$\mathbf{v}$ is a distributed protocol performing a single round in parallel at all vertices, as follows. Every vertex collects the certificates of all its neighbors, and must output ``accept'' or ``reject'', on the basis of its ID, its certificate, and the certificates of its neighbors. The pair $(\mathbf{p},\mathbf{v})$ is a correct PLS for~$\G$ if the following two conditions hold.

\begin{description}
\item[Completeness:] For every $G\in\G$, and for every ID-assignment to the vertices of~$G$, the (non-trustable) prover~$\mathbf{p}$ can assign certificates to the vertices such that the verifier~$\mathbf{v}$ \emph{accepts} at all vertices;
\item[Soundness:] For every $G\notin\cal{G}$, for every ID-assignment to the vertices of~$G$, and for every certificate-assignment to the vertices by the non-trustable prover~$\mathbf{p}$, the verifier~$\mathbf{v}$  \emph{rejects} in at least one vertex. 
\end{description}

The main \emph{complexity measure} for a PLS is the size of the certificates assigned to the vertices by the prover. The objective of the paper is to design schemes with logarithmic-size certificates, for two classes of graphs: the class~$\G^+_k$, $k\geq 0$, of graphs embeddable on an orientable closed surface of genus at most~$k$ (i.e., the graphs of genus~$\leq k$), and the class $\G^-_k$, $k\geq 1$, of graphs embeddable on a non-orientable closed surface of genus at most~$k$ (i.e., the graphs of demi-genus~$\leq k$). 

\paragraph{Remark.}
 
 Throughout the rest of the paper, for $G\in \G^+_k$ (resp., $G\in \G^-_k$) with genus~$k'<k$ (resp., demigenus~$k'<k$), our proof-labeling scheme certifies an embedding of $G$ on $\T_{k'}$ (resp., on $\NT_{k'}$). Therefore, in the following, $k$ is supposed to denote the exact genus of~$G$.


\section{Unfolding a Surface}
\label{sec:unfolding-a-surface}

In this section, we describe how to ``flat down'' a surface, by reducing it to a disk whose boundary has a specific form. This operation is central for constructing the distributed certificates in our proof-labeling scheme. In fact, it provides a centralized certificate for bounded genus. The section is  dedicated to orientable surfaces, and the case of non-orientable surfaces will be treated further in the text. 
 
\subsection{Separation and Duplication}
\label{G_C}

Given a 2-cell embedding of a graph $G$ on a  closed  surface  $X$, a \emph{non-separating cycle}  of the embedding is a simple cycle $C$ in $G$ such that     $X \smallsetminus  C$ is connected.   Figure~\ref{fig:separating} illustrates this notion: the cycle displayed on~(a) is non-separating, as shown on~(b); instead, the cycle displayed on~(c) is separating, as shown on~(d). The result hereafter is a classical result, whose proof can be found in, e.g., \cite{mohar2001graphs,Ortner08}.     
 
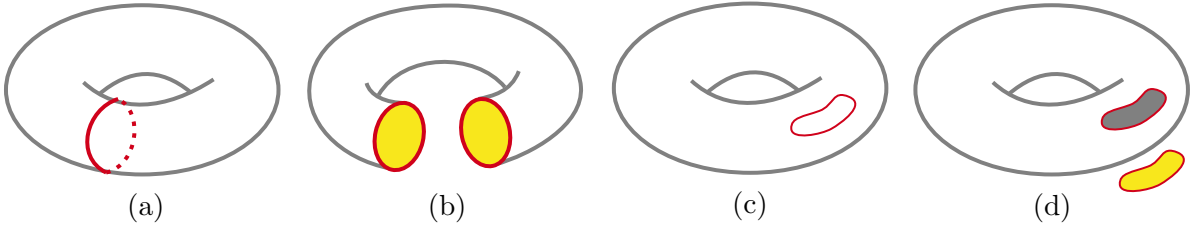
\begin{figure}[tb]
\centering
\tikzset{every picture/.style={line width=0.75pt}} 

\begin{tikzpicture}[x=0.75pt,y=0.75pt,yscale=-1,xscale=1]

\draw  [color={rgb, 255:red, 128; green, 128; blue, 128 }  ,draw opacity=1 ][line width=1.5]  (20,111.41) .. controls (20,87.9) and (50.45,68.84) .. (88.02,68.84) .. controls (125.59,68.84) and (156.04,87.9) .. (156.04,111.41) .. controls (156.04,134.91) and (125.59,153.97) .. (88.02,153.97) .. controls (50.45,153.97) and (20,134.91) .. (20,111.41) -- cycle ;
\draw [color={rgb, 255:red, 128; green, 128; blue, 128 }  ,draw opacity=1 ][line width=1.5]    (58.68,107.63) .. controls (76.04,123.67) and (99.67,121.85) .. (123.64,106.15) ;
\draw [color={rgb, 255:red, 128; green, 128; blue, 128 }  ,draw opacity=1 ][line width=1.5]    (66.34,113.75) .. controls (81.93,100.36) and (96.92,100.14) .. (111.96,112.43) ;
\draw [color={rgb, 255:red, 128; green, 128; blue, 128 }  ,draw opacity=1 ][line width=1.5]    (255.83,116.01) .. controls (262.83,114.51) and (272,112.5) .. (276,102) ;
\draw [color={rgb, 255:red, 128; green, 128; blue, 128 }  ,draw opacity=1 ][line width=1.5]    (206,114) .. controls (218.67,92.5) and (254.33,92.5) .. (269,110) ;
\draw  [color={rgb, 255:red, 128; green, 128; blue, 128 }  ,draw opacity=1 ][line width=1.5]  (323.43,110.56) .. controls (323.43,87.06) and (353.88,68) .. (391.45,68) .. controls (429.02,68) and (459.47,87.06) .. (459.47,110.56) .. controls (459.47,134.07) and (429.02,153.13) .. (391.45,153.13) .. controls (353.88,153.13) and (323.43,134.07) .. (323.43,110.56) -- cycle ;
\draw [color={rgb, 255:red, 128; green, 128; blue, 128 }  ,draw opacity=1 ][line width=1.5]    (362.11,106.79) .. controls (379.46,122.82) and (403.1,121.01) .. (427.07,105.3) ;
\draw [color={rgb, 255:red, 128; green, 128; blue, 128 }  ,draw opacity=1 ][line width=1.5]    (369.77,112.91) .. controls (385.36,99.52) and (400.35,99.3) .. (415.39,111.58) ;
\draw  [color={rgb, 255:red, 128; green, 128; blue, 128 }  ,draw opacity=1 ][line width=1.5]  (473.96,111.41) .. controls (473.96,87.9) and (504.41,68.84) .. (541.98,68.84) .. controls (579.55,68.84) and (610,87.9) .. (610,111.41) .. controls (610,134.91) and (579.55,153.97) .. (541.98,153.97) .. controls (504.41,153.97) and (473.96,134.91) .. (473.96,111.41) -- cycle ;
\draw [color={rgb, 255:red, 128; green, 128; blue, 128 }  ,draw opacity=1 ][line width=1.5]    (512.64,107.63) .. controls (530,123.67) and (553.63,121.85) .. (577.6,106.15) ;
\draw [color={rgb, 255:red, 128; green, 128; blue, 128 }  ,draw opacity=1 ][line width=1.5]    (520.3,113.75) .. controls (535.89,100.36) and (550.88,100.14) .. (565.92,112.43) ;
\draw [color={rgb, 255:red, 208; green, 2; blue, 27 }  ,draw opacity=1 ][line width=1.5]    (70,153) .. controls (54,145) and (60,121) .. (76,116) ;
\draw [color={rgb, 255:red, 208; green, 2; blue, 27 }  ,draw opacity=1 ][line width=1.5]  [dash pattern={on 1.69pt off 2.76pt}]  (70,153) .. controls (86,149.17) and (89,122.5) .. (76,116) ;
\draw [color={rgb, 255:red, 128; green, 128; blue, 128 }  ,draw opacity=1 ][line width=1.5]    (212,151) .. controls (186.67,146.17) and (167.67,127.17) .. (172,106) .. controls (176.33,84.83) and (201.67,68.83) .. (242,69) .. controls (282.33,69.17) and (305.67,90.83) .. (307,110) .. controls (308.33,129.17) and (282.33,143.83) .. (264,149) ;
\draw [color={rgb, 255:red, 128; green, 128; blue, 128 }  ,draw opacity=1 ][line width=1.5]    (200,108) .. controls (201.67,114.17) and (209.33,118.17) .. (220,118) ;
\draw  [color={rgb, 255:red, 208; green, 2; blue, 27 }  ,draw opacity=1 ][fill={rgb, 255:red, 248; green, 231; blue, 28 }  ,fill opacity=1 ][line width=1.5]  (204.59,132.22) .. controls (206.67,123.01) and (213.59,116.72) .. (220.07,118.18) .. controls (226.54,119.63) and (230.1,128.29) .. (228.02,137.5) .. controls (225.95,146.72) and (219.02,153.01) .. (212.55,151.55) .. controls (206.08,150.09) and (202.52,141.44) .. (204.59,132.22) -- cycle ;
\draw  [color={rgb, 255:red, 208; green, 2; blue, 27 }  ,draw opacity=1 ][fill={rgb, 255:red, 248; green, 231; blue, 28 }  ,fill opacity=1 ][line width=1.5]  (247.7,135.42) .. controls (245.74,126.18) and (249.4,117.57) .. (255.89,116.19) .. controls (262.38,114.81) and (269.23,121.18) .. (271.2,130.42) .. controls (273.16,139.66) and (269.49,148.27) .. (263.01,149.65) .. controls (256.52,151.03) and (249.67,144.66) .. (247.7,135.42) -- cycle ;
\draw  [color={rgb, 255:red, 208; green, 2; blue, 27 }  ,draw opacity=1 ] (426,122) .. controls (433.67,119.17) and (432.33,111.17) .. (441,115) .. controls (449.67,118.83) and (437,128.17) .. (432,131) .. controls (427,133.83) and (412.33,136.83) .. (412,131) .. controls (411.67,125.17) and (418.33,124.83) .. (426,122) -- cycle ;
\draw  [color={rgb, 255:red, 208; green, 2; blue, 27 }  ,draw opacity=1 ][fill={rgb, 255:red, 128; green, 128; blue, 128 }  ,fill opacity=1 ] (581,119) .. controls (588.67,116.17) and (587.33,108.17) .. (596,112) .. controls (604.67,115.83) and (592,125.17) .. (587,128) .. controls (582,130.83) and (567.33,133.83) .. (567,128) .. controls (566.67,122.17) and (573.33,121.83) .. (581,119) -- cycle ;
\draw  [color={rgb, 255:red, 208; green, 2; blue, 27 }  ,draw opacity=1 ][fill={rgb, 255:red, 248; green, 231; blue, 28 }  ,fill opacity=1 ] (590,150) .. controls (597.67,147.17) and (596.33,139.17) .. (605,143) .. controls (613.67,146.83) and (601,156.17) .. (596,159) .. controls (591,161.83) and (576.33,164.83) .. (576,159) .. controls (575.67,153.17) and (582.33,152.83) .. (590,150) -- cycle ;

\draw (79,162) node [anchor=north west][inner sep=0.75pt]   [align=left] {(a)};
\draw (229,162) node [anchor=north west][inner sep=0.75pt]   [align=left] {(b)};
\draw (381,161) node [anchor=north west][inner sep=0.75pt]   [align=left] {(c)};
\draw (531,162) node [anchor=north west][inner sep=0.75pt]   [align=left] {(d)};

\end{tikzpicture}
\caption{Separating and non-separating cycles.}
\label{fig:separating}
\end{figure} 
 
 \begin{lemma} \label{first_cycle}
Let $G$ be a graph embeddable on a closed orientable surface $X$ with genus $k \geq 1$. For any 2-cell embedding of $G$ on $X$, there exists a non-separating cycle $C$ in $G$.  
 \end{lemma}

\subsubsection{Cycle-Duplication} 

Let $G$ be a graph embeddable on a closed orientable surface $X$. An orientable cycle is a cycle of $G$ whose embedding on~$X$ yields an orientable curve. Given a 2-cell embedding~$f$ of  $G$ on~$X$,  let $C$ be a non-separating orientable cycle of $G$ whose existence is guaranteed by Lemma~\ref{first_cycle}. By definition, the left and right sides of $C$ can be defined on the neighborhood of $C$. We denote by $G_C$ the graph obtained by the \emph{duplication} of~$C$ in~$G$. Specifically, let us assume that $C = (v_0, \dots, v_r)$. Every vertex $w\notin C$ remains in $G_C$, as well as every edge non incident to a vertex of~$C$. Every vertex $v_i $ of $C$ is replaced by a \emph{left} vertex~$v'_i$ and a \emph{right} vertex~$v''_i $. For every $i=0,\dots,r-1$,  $\{v'_i, v'_{i+1} \} $  and $ \{v''_i, v''_{i+1} \} $ are  edges of $G_C$, as well as $\{v'_r,v'_0\}$ and $\{v''_r,v''_0\}$. Finally, for every $i=0,\dots,r$, and every neighbor $w\notin C$ of $v_i$ in $G$,  if $f( \{v_i, w \}) $  meets the left of $C$, then $\{v'_i,w\}$ is an edge of $G_C$, otherwise  $\{v''_i,w\}$ is an edge of $G_C$.  The embedding $f$ of $G$ on $X$ directly induces an embedding of $G_C $ on~$X$. Figure~\ref{fig:duplication-simple}(a-b) illustrates the operation of duplication, and the resulting embedding on~$X$. 

\begin{figure}[tb]
\centering
\input{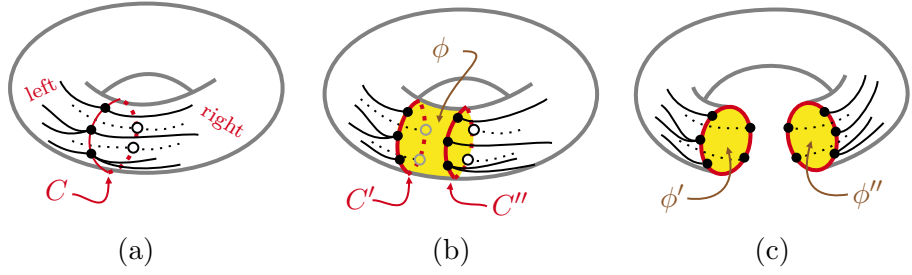}
\caption{Cycle-duplication and the associated surface.}
\label{fig:duplication-simple}
\end{figure}

The embedding of $G_C $ on~$X$ is however not a 2-cell embedding, as it contains the face $\phi$ between $C'$ and $C''$ on~$X$, where  $C'  = (v'_0, \dots, v'_r)$ and $C'' = (v''_0, \dots,v''_r)$ (see Figure~\ref{fig:duplication-simple}(b)). Formally, $\phi$ is the face with boundaries $C'$ and $C''$, and, as such, it is not homeomorphic to a disc. Let $X_C$ be the closed surface\footnote{Notice that 
$ X \setminus \phi$,  $\overline{ \phi'}  =  \phi' \cup C'$ (where $\overline{ \phi'}$ denotes the adherence of  $\phi'$), and $\overline{\phi'' } =   \phi'' \cup C''$ are compact sets. Thus $X_C$ is compact as  the union of these  three sets.}  obtained from $X$ by removing~$\phi$, and by replacing $\phi$ with two faces $\phi'$ and $\phi''$ with  boundary walks $C'$ and $C''$, respectively (see Figure~\ref{fig:duplication-simple}(c)). The embedding $f$ of $G$ on $X$ induces a 2-cell embedding $f_C$ of $G_C$ on $X_C $.  
Also, since $C$ is a non-separating cycle of  $G$ in $X$, the surface $X_C$ is path-connected, which ensures that $G_C$ is connected using Lemma \ref{lem:simul}.  

 Moreover, as $X$ is orientable, $X_C$  is  also orientable. Indeed, every simple cycle of $X_C$ not intersecting $\phi'$ nor $\phi''$ is a cycle of $X$, and  is therefore orientable. Furthermore, any simple cycle of $X_C$ intersecting $\phi'$ and/or $\phi''$ is homotopic to a cycle separated from both boundaries of $\phi'$ and $\phi''$ by an open set, and thus is homotopic to a cycle of~$X$. It follows that  $X_C$ is a closed orientable surface, and thus, thanks to Lemma~\ref{euler}, the genus of $X_C $  is~$k-1$. 

\subsubsection{Path-Duplication} 

Again, let us consider a graph $G$, an orientable closed surface~$X$, and a 2-cell embedding $f$ of   $G$ on  $X$. Let $\chi, \psi$ be two distinct faces of the embedding, and let $ P =   (w_0, \dots, w_s)$  be simple path (possibly reduced to a single vertex belonging to the two cycles) between $\chi$ and $\psi$ (see Figure~\ref{fig:path-duplication-simple}). That is, $P$ is such that  $w_0 $ is on the boundary of  $\phi$, $w_s$ is on the boundary of  $\psi$, and no  intermediate vertex $w_i$, $0 < i < s$,  is  on the boundary of $\chi$ or~$\psi$.  The path~$P$ enables to define a graph  $G_{ P}$ obtained by duplicating the path $P$ in a  way similar to the way the cycle $C$ was duplicated in the previous section. There is only one subtle difference, as the left and right side of the path cannot be defined at its endpoints. Nevertheless, the left and right sides of $P$ can still be properly defined all along $P$, including its extremities, by virtually ``extending'' $P$ so that it ends up in the interiors  of $\chi$ and $\psi$.  Thanks to this path-duplication, the two faces $\chi$ and $\psi$ of $G$ are replaced by a unique face of $G_P$ as illustrated on Figure~\ref{fig:path-duplication-simple}, reducing the number of faces by one.  

\begin{figure}[tb]
\begin{center}
\input{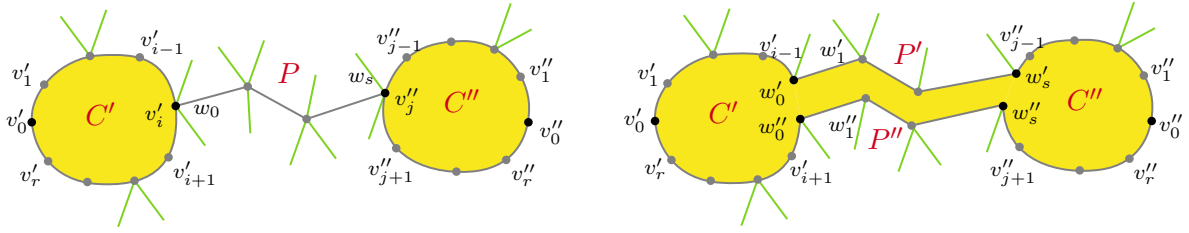}
\caption{Path-duplication. }
\label{fig:path-duplication-simple} 
\end{center}
\end{figure} 

\paragraph{Remark.} 

Cycle-duplication and path-duplication are typically used conjointly. A basic example, used for the torus~$\T_1$ in the next section, consists of, first, duplicating a cycle~$C$, then connecting the two faces resulting from this duplication by a path~$P$, and, finally, duplicating~$P$ for merging these two faces into one single face. Further, for the general case $\T_k$, $k\geq 1$,  $k$~cycles $C_1,\dots,C_k$ are duplicated, and  $2k-1$  paths $P_1,\dots,P_{2k-1}$ are duplicated for connecting the $2k$ faces $\phi'_1, \phi''_1,\dots,\phi'_k,\phi''_k$ resulting from the $k$ cycle-duplications, ending up in a unique face~$\phi^*$. 

\subsection{Unfolding the Torus}
\label{subsec:cas-du-tore}

As a warm up, we consider the case of a graph embedded on the torus~$\T_1$, and show how to ``unfold'' this embedding. 

\begin{figure}[tb]
\centering
\input{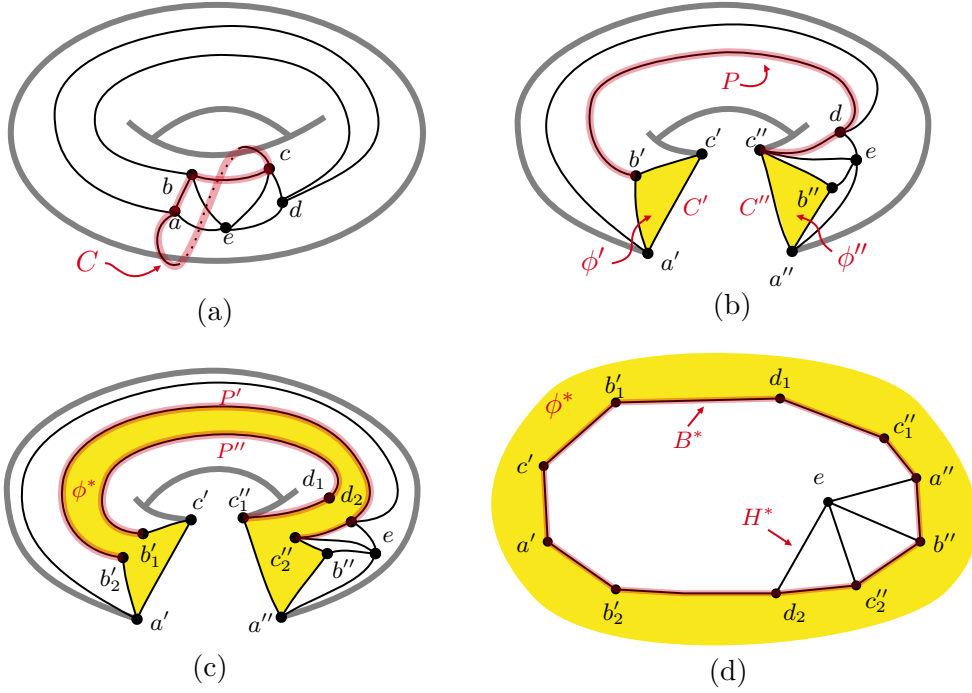}
\caption{Unfolding $K_5$ embedded on the torus $\T_1$. }
\label{fig:cut-pierre} 
\end{figure} 

\subsubsection{Making  a Graph of Genus 1 Planar}
\label{sec:making-planar}

Let $G$ be a graph, and let $f$ be a 2-cell embedding of $G$ on~$X=\T_1$ --- see Figure~\ref{fig:cut-pierre}(a) for an embedding of $K_5$ on $\T_1$, as an illustrative example. Let $C=(v_0,\dots,v_r)$ be a non-separating orientable cycle of $G$, e.g., the cycle $(a,b,c)$ on Figure~\ref{fig:cut-pierre}(a). Let 
$
C'=(v'_0,\dots,v'_r) \; \mbox{and} \; C''=(v''_0,\dots,v''_r)
$
be the two cycles resulting from the duplication of~$C$, e.g., the cycles $(a',b',c')$ and $(a'',b'',c'')$ on Figure~\ref{fig:cut-pierre}(b). The graph $G_C$ with two new faces $\phi'$ and $\phi''$ is connected. In particular,  there exists  a simple path  $ P =   (w_0,  \dots, w_s)$  in $G_C$ from a vertex $v'_i\in C'$  to a vertex $v''_j\in C''$, such that every  intermediate vertex $w_i$, $0 < i < s$,  is not  in $C'\cup C''$, e.g., the path $(c'',d,b')$ on Figure~\ref{fig:cut-pierre}(b). Note that it may be the case that $i\neq j$. On Figure~\ref{fig:cut-pierre}(b), the path $(b'',e,d,b')$ satisfies $i=j$, but Figure~\ref{fig:K33-torus-paths1} illustrates an embedding of $K_{3,3}$ on $\T_1$ for which $i=j$ cannot occur (simply because every vertex of $K_{3,3}$ has degree~3, and thus it has a single edge not in the cycle). Duplicating $P$ enables to obtain a graph $G_{C, P}$ with a special face $\phi^*$,  whose boundary contains all duplicated vertices and only them (see Figure~\ref{fig:cut-pierre}(c)). The details of the vertex-duplications, and of the edge-connections are detailed hereafter.

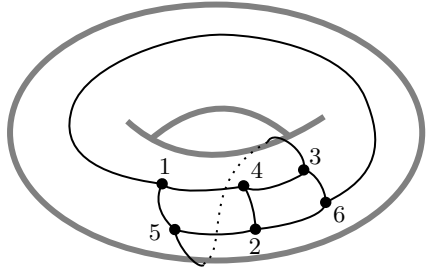
\begin{figure}[tb]
\begin{center}
\tikzset{every picture/.style={line width=0.75pt}} 

\begin{tikzpicture}[x=0.75pt,y=0.75pt,yscale=-1,xscale=1]

\draw  [color={rgb, 255:red, 128; green, 128; blue, 128 }  ,draw opacity=1 ][line width=2.25]  (302.25,151.38) .. controls (302.25,115.82) and (348.31,87) .. (405.13,87) .. controls (461.94,87) and (508,115.82) .. (508,151.38) .. controls (508,186.93) and (461.94,215.75) .. (405.13,215.75) .. controls (348.31,215.75) and (302.25,186.93) .. (302.25,151.38) -- cycle ;
\draw [color={rgb, 255:red, 128; green, 128; blue, 128 }  ,draw opacity=1 ][line width=2.25]    (360.75,145.67) .. controls (387,169.92) and (422.75,167.17) .. (459,143.42) ;
\draw [color={rgb, 255:red, 128; green, 128; blue, 128 }  ,draw opacity=1 ][line width=2.25]    (372.33,154.92) .. controls (395.92,134.67) and (418.58,134.33) .. (441.33,152.92) ;
\draw [line width=0.75]    (378.17,177.17) .. controls (344.17,172.83) and (326.17,160.5) .. (333.5,138.5) .. controls (340.83,116.5) and (385.17,101.5) .. (409.17,102.17) .. controls (433.17,102.83) and (470.17,111.5) .. (479.5,133.83) .. controls (488.83,156.17) and (486.83,173.5) .. (459.83,186.5) ;
\draw    (448.83,170.17) .. controls (451.17,160.5) and (433.83,148.83) .. (430.83,156.17) ;
\draw    (399,218) .. controls (395.83,220.5) and (385.17,208.17) .. (384.5,200.17) ;
\draw  [dash pattern={on 0.84pt off 2.51pt}]  (430.83,156.17) .. controls (395.5,175.5) and (412.83,210.83) .. (398,218) ;
\draw    (424.75,200) .. controls (415.17,203.5) and (395.17,203.5) .. (384.5,200.17) ;
\draw [shift={(384.5,200.17)}, rotate = 197.35] [color={rgb, 255:red, 0; green, 0; blue, 0 }  ][fill={rgb, 255:red, 0; green, 0; blue, 0 }  ][line width=0.75]      (0, 0) circle [x radius= 2.34, y radius= 2.34]   ;
\draw    (459.83,186.5) .. controls (459.5,195.17) and (432.17,199.5) .. (424.75,200) ;
\draw [shift={(424.75,200)}, rotate = 176.14] [color={rgb, 255:red, 0; green, 0; blue, 0 }  ][fill={rgb, 255:red, 0; green, 0; blue, 0 }  ][line width=0.75]      (0, 0) circle [x radius= 2.34, y radius= 2.34]   ;
\draw [shift={(459.83,186.5)}, rotate = 92.2] [color={rgb, 255:red, 0; green, 0; blue, 0 }  ][fill={rgb, 255:red, 0; green, 0; blue, 0 }  ][line width=0.75]      (0, 0) circle [x radius= 2.34, y radius= 2.34]   ;
\draw    (378.17,177.17) .. controls (385.17,184.5) and (426.25,177.67) .. (418.83,178.17) ;
\draw [shift={(418.83,178.17)}, rotate = 176.14] [color={rgb, 255:red, 0; green, 0; blue, 0 }  ][fill={rgb, 255:red, 0; green, 0; blue, 0 }  ][line width=0.75]      (0, 0) circle [x radius= 2.34, y radius= 2.34]   ;
\draw    (384.5,200.17) .. controls (379.5,195.83) and (373.17,187.5) .. (378.17,177.17) ;
\draw [shift={(378.17,177.17)}, rotate = 295.82] [color={rgb, 255:red, 0; green, 0; blue, 0 }  ][fill={rgb, 255:red, 0; green, 0; blue, 0 }  ][line width=0.75]      (0, 0) circle [x radius= 2.34, y radius= 2.34]   ;
\draw    (418.83,178.17) .. controls (425.83,185.5) and (456.25,169.67) .. (448.83,170.17) ;
\draw [shift={(448.83,170.17)}, rotate = 176.14] [color={rgb, 255:red, 0; green, 0; blue, 0 }  ][fill={rgb, 255:red, 0; green, 0; blue, 0 }  ][line width=0.75]      (0, 0) circle [x radius= 2.34, y radius= 2.34]   ;
\draw    (459.83,186.5) .. controls (459.83,179.83) and (454.5,171.17) .. (448.83,170.17) ;
\draw    (424.75,200) .. controls (424.75,193.33) and (422.5,183.83) .. (418.83,178.17) ;

\draw (375,163) node [anchor=north west][inner sep=0.75pt]  [font=\footnotesize]  {$1$};
\draw (370,196) node [anchor=north west][inner sep=0.75pt]  [font=\footnotesize]  {$5$};
\draw (420,203) node [anchor=north west][inner sep=0.75pt]  [font=\footnotesize]  {$2$};
\draw (421,165.23) node [anchor=north west][inner sep=0.75pt]  [font=\footnotesize]  {$4$};
\draw (450.08,158) node [anchor=north west][inner sep=0.75pt]  [font=\footnotesize]  {$3$};
\draw (462.08,186.9) node [anchor=north west][inner sep=0.75pt]  [font=\footnotesize]  {$6$};

\end{tikzpicture}
\caption{$K_{3,3}$ embedded on the torus  $\T_1$.}
\label{fig:K33-torus-paths1} 
\end{center}
\end{figure}

\paragraph{Connections in path-duplication.}

Let $ P' =   (w'_0,  \dots, w'_s)$ and $ P'' =   (w''_0,  \dots, w''_s)$ be the two paths obtained by duplicating~$P$. In particular, the vertices $w_0=v'_i$ and $w_s = v''_j$ are both duplicated in $w'_0, w''_0$,  and  $w'_s, w''_s$, respectively. The edges 
\[
\{v'_{i-1},v'_i\}, \; \{v'_i,v'_{i+1}\}, \; \{v''_{j-1},v''_j\}, \; \mbox{and} \; \{v''_j,v''_{j+1}\}
\] 
are  replaced by the edges connecting $v'_{i-1},v'_{i+1},v''_{j-1},v''_{j+1}$ to $w'_0, w''_0, w'_s,w''_s$.  For defining these edges, observe that the path $P$ in $\T_1$ induces a path $Q =   (v_i,  w_1,\dots, w_{s-1}, v_j)$ in $G$  connecting the vertices $v_i$ and $v_j$ of~$C$,  such that, in the embedding on~$\T_1$, the edge  $\{v_i,  w_1\}$ meets $C$ on one side while the edge  $\{w_{s-1}, v_j\}$  meets $C$  on the other side (see Figure~\ref{fig:cutpathQ}(a-b)). 

\begin{figure}[tb]
\centering
\input{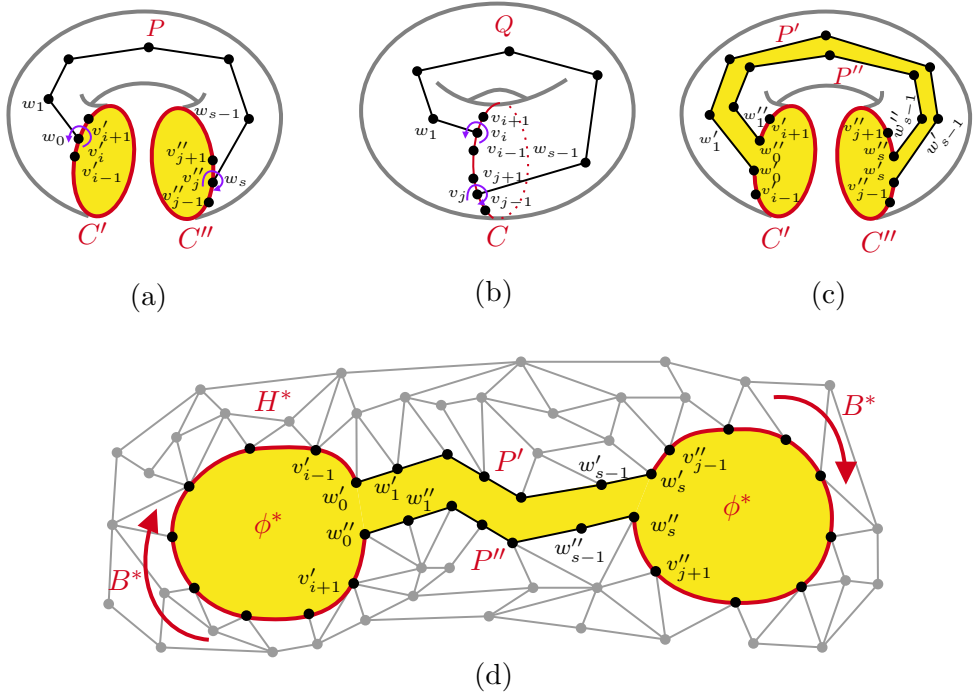}
\caption{Setting up the connections in path-duplication.}
\label{fig:cutpathQ} 
\end{figure} 

\noindent Let us assume, w.l.o.g., that  the edges of $C \cup Q$ around $v_i$ are in the order 
\[
\{v_i, v_{i-1} \}, \; \{v_i, v_{i+1}\}, \; \{v_i, w_1 \}
\]
when visited counter-clockwise in~$\T_1$. It follows that the edges of $C \cup Q$ around $v_j$ are in the order 
\[
 \{v_j, v_{j-1} \},  \{v_j, v_{j+1}\},  \{v_j, w_{s-1} \} 
 \]
 when visited clockwise in~$\T_1$ (see Figure~\ref{fig:cutpathQ}(b)). These orders are transferred in~$G_C$, that is, the edges of $C' \cup P$ around $v'_i$ are in counter-clockwise order 
 \[
  \{v'_i, v'_{i-1} \},   \{v'_i, v'_{i+1} \}, \{v'_i, w_{1} \}, 
  \]
  while the edges of $C'' \cup P$ around $v''_j$ are in clockwise order 
  \[
  \{v''_j, v''_{j-1} \},  \{v''_j, v''_{j+1}\}, \{v''_j, w_{s-1} \},
  \]
  as illustrated on  Figure~\ref{fig:cutpathQ}(a).   This guarantees that $v'_{i-1}$ and $v''_{j-1}$ are in the same side of the path $P$. More generally,  the relative positions of $v'_{i-1}$, $v'_{i+1}$,  $v''_{j-1}$, and  $v''_{j+1}$ w.r.t.~$P$ are as follows.  Vertices $v'_{i-1}$ and $v''_{j-1}$  are on the same side of~$P$, while vertices $v'_{i+1}$ and $v''_{j+1}$ are on the other  side of~$P$ (see again Figure~\ref{fig:cutpathQ}(a)).  As a consequence, it can be assumed that, in the graph $G_{C, P}$ resulting from the duplications of both~$C$ and~$P$, the vertices $v'_{i-1}$ and $v''_{j-1}$ are connected to the end points of~$P'$,  while $v'_{i+1}$ and $v''_{j+1}$ are connected to the end points of~$P''$.  It follows that  
  \[
  \{w'_0, v'_{i-1} \}, \{w'_s, v''_{j-1} \}, \{w''_0, v'_{i+1} \}, \mbox{and} \; \{w''_s, v''_{j+1} \}
  \]
  are edges of $G_{C, P}$ (see Figure~\ref{fig:cutpathQ}(c)).  

\paragraph{Unfolding.}

The embedding $f$ of $G$ on $X=\T_1$ directly induces an embedding of $H^*=G_{C, P}$ on~$X_C$, as illustrated on Figure~\ref{fig:cutpathQ}(d). As observed before, the genus of $X_C$ is one less than the genus of~$X$. Since $X=\T_1$, it follows that the embedding $f$ of $G$ on $\T_1$  actually induces a planar embedding $f^*$ of $H^*$.  The faces of this embedding are merely the faces of $G$,  plus another, special face $\phi^*$ whose boundary walk is 
\begin{align}
B^*= (w'_0, w'_1, \dots, w'_s, & v''_{j-1}, v''_{j-2}, \dots, v''_0,v''_r,\dots,v''_{j+1},  \label{eq:boundwalk}\\
& w''_s, w''_{s-1} , \dots, w''_0, v'_{i+1}, v'_{i+2},\dots,v'_r,v'_0,\dots, v'_{i-1}), \nonumber
\end{align}
as displayed on Figure~\ref{fig:cutpathQ}(d)). For instance, on Figure~\ref{fig:cut-pierre}(d), $B^*=(b'_1,d_1,c''_1,a'',b'',c''_2,d_2,b'_2,a',c')$. The face $\phi^*$ can be pointed out as special, as on Figure~\ref{fig:cutpathQ}(d), or can be made  the external face of the embedding of~$H^*$, as on Figure~\ref{fig:cut-pierre}(d). Our interest for $H^*,f^*,\phi^*$, and $B^*$ as far as the design of a proof-labeling scheme is concerned, resides in the fact that, as shown hereafter, they form a (centralized) certificate for genus~1. 

\subsubsection{Certifying Genus~1}
\label{subsec:certif-genus-1}

Let us first define the notion of  \emph{splitting}. 

\begin{definition}\label{def:splitting}
A \emph{splitting}  of a graph $G$ into a graph $H$ is a pair $\sigma=(\alpha,\beta)$ of functions, where $\alpha:V(G)\to 2^{V(H)}$, and $\beta:E(G)\to 2^{E(H)}$, such that:
\begin{enumerate}
\item the set $\{\alpha(v):v\in V(G)\}$ forms a partition of $V(H)$; 
\item for every $e=\{u,v\}\in E(G)$, $\beta(e)$ is a non-empty matching between $\alpha(u)$ and $\alpha(v)$. 
\end{enumerate} 
\end{definition}

For every $v\in V(G)$, the vertices $\alpha(v)$ in $H$ are the \emph{avatars} of $v$ in $H$. The \emph{degree} of a splitting $\sigma=(\alpha,\beta)$ of $G$ into $H$ is $\max_{v\in V(G)}|\alpha(v)|$, and $H$ is said to be a $d$-splitting of~$G$ whenever $d=\max_{v\in V(G)}|\alpha(v)|$. A vertex $v\in V(G)$ is split in $H$ if $|\alpha(v)|\geq 2$, otherwise it is not split in~$H$.  If a vertex $v$ is not split, we abuse notation by writing $\alpha(v)=v$, i.e., by referring to $v$ as a vertex of $G$ and as a vertex of~$H$. For any subgraph $G'$ of~$G$, we denote by $\sigma(G')$ the subgraph $H'$ of $H$ with edge-set $E(H')=\{\beta(e):e\in E(G')\}$, and vertex-set $V(H')$ equal to the set of end-points of the edges in~$E(H')$. Note that $\sigma(G')$ may not be connected, even if $G'$ is connected.  On the other hand, $\sigma(G')$ is necessarily connected whenever no vertices in $G'$ are split. With a slight abuse of notation, for  a splitting $\sigma=(\alpha,\beta)$ of $G$ into $H$, we often refer to $\sigma(v)$ instead of $\alpha(v)$ for $v\in V(G)$, and to $\sigma(e)$ instead of $\beta(e)$ for  $e\in E(G)$.

Let $H$ be a splitting of a graph~$G$ for which there exists a 2-splittting $U$ of $G$ such that $H$ is a 2-splitting of $U$. Let $f$ be a planar embedding of~$H$, and let~$\phi$ be a face of~$H$ embedded on~$\T_0$. Let 
$
B=(u_0,\dots,u_N)
$
be a boundary walk of $\phi$. Let $\sigma_{G,U}$ and $\sigma_{U,H}$ be the splitting of $G$ into $U$, and the splitting of $U$ into $H$, respectively. Let $\sigma_{G,H}=\sigma_{U,H}\circ \sigma_{G,U}$. 
We say that $(G,H,B,U)$ is \emph{globally consistent} if there exist vertices $v'_0,\dots,v'_r$, $v''_0,\dots,v''_r$, $w'_0, \dots, w'_s$, $w''_0, \dots, w''_s$ of~$H$ such that 
\[
B= (w'_0, \dots, w'_s,  v''_{j-1},  \dots, v''_0,v''_r,\dots, v''_{j+1}, 
w''_s,  \dots, w''_0, v'_{i+1},\dots,v'_r,v'_0,\dots, v'_{i-1}) 
\]
where 
\begin{itemize}
\item for every vertex $u\notin \{v'_k,v''_k: 0\leq k \leq r\}\cup \{w'_k,w''_k: 0\leq k \leq s\}$ of $H$,  $\sigma_{G,H}(u)=u$;

\item for every $k\in\{1,\dots,s-1\}$, $\sigma_{U,H}^{-1}(\{w'_k,w''_k\})=w_k \in V(U)$, and $\sigma_{G,U}(w_k)=w_k$;

\item for every~$k\in\{0,\dots,r\}\smallsetminus \{i,j\}$, $\sigma_{G,U}^{-1}(\{v'_k,v''_k\})=v_k\in V(U)$, and $\sigma_{U,H}(v_k)=v_k$;

\item $\sigma_{U,H}^{-1}(\{w'_0,w''_0\})=v'_i \in V(U)$, $\sigma_{U,H}^{-1}(\{w'_s,w''_s\})=v''_j \in V(U)$, $\sigma_{G,U}^{-1}(\{v'_i,v''_i\})=v_i\in V(G)$, and $\sigma_{G,U}^{-1}(\{v'_j,v''_j\})=v_j\in V(G)$ (note that this applies to both cases $i=j$ and $i\neq j$).
\end{itemize} 

\paragraph{Remark.}

The way the vertices of $B$ are listed provides $B$ with a reference direction, say clockwise. This reference direction is crucial for checking that the two faces of $U$ with respective boundary walks $v'_i,v'_{i+1},\dots,v'_r,v'_0,\dots, v'_{i-1}$ and $v''_j,v''_{j-1},  \dots, v''_0,v''_r,\dots, v''_{j+1}$ can be merged for forming a handle. Global consistency specifies that, for these two faces to be merged, their directions inherited from the reference direction of $B$ must both be clockwise (cf., Figure~\ref{fig:cutpathQ}(d)). Indeed, while one face is traversed clockwise with increasing indices, the other is traversed clockwise with decreasing indices. This matches the specification of handles (cf. Figure~\ref{fig:hcc}). 

\medskip 

By the construction in Section~\ref{sec:making-planar}, for every graph $G$ of genus~1, $(G,H^*,B^*,U^*)$ is globally  consistent, where $H^*=G_{C, P}$, $U^*=G_C$, and  $B^*$ is  the boundary walk of $\phi^*$ displayed in Eq.~\eqref{eq:boundwalk}. The following result is specific to the torus, but it illustrates the basis for the design of our proof-labeling schemes. 
 
\begin{lemma}\label{lem:certif-T1-centralized}
Let $H$ be a splitting of a graph $G$, and assume that there exists a planar embedding~$f$ of $H$ with a face $\phi$  and a boundary walk $B$ of $\phi$.  Let $U$ be a 2-splittting of $G$ such that $H$ is a 2-splitting of $U$. If $(G,H,B,U)$ is globally consistent, then $G$ can be embedded on the torus~$\T_1$.
\end{lemma}

\begin{proof}
Using the specifications of the splits, the two sub-paths $(w'_0, \dots, w'_s)$ and  $(w''_0, \dots, w''_s)$ of~$B$ can be identified by merging each pair of vertices $w'_k$ and $w''_k$, $k\in\{1,\dots,s-1\}$, into a single vertex~$w_k=\sigma_{U,H}^{-1}(\{w'_k,w''_k\})$ of~$U$, by merging the vertices $w'_0$ and $w''_0$ into a single vertex $v'_i$ of~$U$, and by merging the vertices $w'_s$ and $w''_s$ into a single vertex $v''_j$ of~$U$. The resulting sequence $v'_i,w_1,\dots,w_{s-1},v''_j$ forms a path in~$U$ connecting two faces $\phi'$ and $\phi''$, replacing the face~$\phi$ of the planar embedding~$f$ of $H$, with respective boundary walks $(v'_0,v'_1,\dots, v'_r)$ and  $(v''_r, v''_{r-1},\dots, v''_0)$, where the vertices are ordered clockwise. These transformations preserve the planarity of the embedding, that is, $U$ is planar. Next, the two cycles $(v'_0,\dots, v'_r)$ and  $(v''_0, \dots, v''_r)$ can be identified, by merging each pair of nodes $v'_k$ and $v''_k$ into a single node~$v_k=\sigma_{G,U}^{-1}(\{v'_k,v''_k\})$ of~$G$. As a result, the two faces $\phi'$ and $\phi''$ are replaced by a handle, providing an embedding of $G$ on~$\T_1$.
\end{proof}

The outcome of Lemma~\ref{lem:certif-T1-centralized} is that $(H^*,f^*,\phi^*,B^*)$ is essentially a certificate that $G$ can be embedded on~$\T_1$ (up to also providing the ``intermediate'' splitting~$U^*$ resulting from cycle-duplication). In the next section, we show how to generalize this construction for deriving a certificate that a graph $G$ can be embedded on~$\T_k$, $k>1$. 

\subsection{Unfolding $\T_k$ for $k\geq 1$}
\label{sec:unfolding-any-Tk}

The process described in the previous section for genus~1 can be generalized to larger genus $k\geq 1$, as follows. Again, let $G$ be a graph, and let $f$ be a 2-cell embedding of $G$ on~$\T_k$. 

\subsubsection{The Face-Duplication Phase}
\label{subsec:faceduplication}

Let $X^{(0)}=\T_k$. As for the torus, let $C_1$ be a non-separating orientable cycle of $G^{(0)}=G$, and let us consider the embedding of $G^{(1)}=G^{(0)}_{C_1}$  induced by~$f$, on  the  surface $X^{(1)}=X^{(0)}_{C_1}$ of genus~$k-1$. This operation can be repeated. Indeed, by Lemma~\ref{first_cycle}, there exists a non separating cycle $C_2$ of  $G^{(1)}$. The graph $G^{(2)}=G^{(1)}_{C_2}$ can be embedded on the surface $X^{(2)}=X^{(1)}_{C_2} $ with one face more than the number of faces of the embedding of $G^{(1)}$ on  $X^{(1)}$, and thus two more  faces than the number of faces of the embedding of $G$ on  $\T_k$. By Lemma~\ref{euler}, $X^{(2)}$ has thus genus~$k-2$. See Figure~\ref{fig:pierreT2}(a-b). 

This process can actually be iterated  $k$ times, resulting in a sequence of $k+1$ graphs $G^{(0)}, \dots, G^{(k)}$ where $G^{(0)}=G$, and a sequence of $k+1$ closed surfaces $X^{(0)}, \dots, X^{(k)}$ where $X^{(0)}=\T_k$. Each graph $G^{(i)}$ is embedded on the closed surface $X^{(i)}$ of genus~$k-i$, as follows. The embedding  of $G^{(0)}$ on $X^{(0)}$ is the embedding of $G$ on $X$, and, for every $i=0,\dots,k-1$,  the embedding of $G^{(i+1)}$ on $X^{(i+1)}$ is induced by the embedding of $G^{(i)}$ on $X^{(i)}$, after duplication of a non-separating cycle~$C_{i+1}$ of $G^{(i)}$ into two cycles $C'_{i+1}$ and $C''_{i+1}$. 

The closed surface $X^{(k)}$ is of genus~0, i.e.  $X^{(k)}$ is homeomorphic to the sphere $\T_0 = S^2$ (see  Figure~\ref{fig:pierreT2}(b)). The graph  $G^{(k)}$ is therefore planar, for it contains $k$ more faces  than the number of faces in~$G$, as two new faces $\phi'_i$ and $\phi''_i$ are created at each iteration~$i$, in replacement to one face $\phi_i$, for every $i=1,\dots,k$.

\subsubsection{The Face-Reduction Phase}
\label{subsec:facereduction}

The objective is now to replace the $2k$ faces $\phi'_i,\phi''_i$, $i=0,\dots,k-1$, by a single face. For this purpose, let us relabel these faces as $\psi_1,\dots,\psi_{2k}$ (see Figure~\ref{fig:pierreT2}(c)) so that, for $i=1,\dots,k$,
\[
\phi'_i=\psi_{2i-1}, \;\mbox{and}\; \phi''_i=\psi_{2i}.
\]
Let $\chi_1=\psi_1$. There exists a simple path~$P_1$ between the two  faces  $\chi_1$ and $\psi_2$.  Duplicating $P_1$ preserves the fact that the  graph $G^{(k+1)}=G^{(k)}_{P_1}$ can be embedded on the sphere~$\T_0$.  By this duplication, the two faces $\chi_1$ and $\psi_2$ are merged into a single face $\chi_2$. Now, there is a simple path~$P_2$ between the two  faces  $\chi_2$ and $\psi_3$ (see Figure~\ref{fig:pierreT2}(d)).  Again, duplicating $P_2$ preserves the fact that the  graph $G^{(k+2)}=G^{(k+1)}_{P_2}$ can be embedded on the sphere~$\T_0$, in which the two faces $\chi_2$ and $\psi_3$ are now merged into a single face $\chi_3$. By iterating this process, a finite sequence of graphs $G^{(k)},\dots, G^{(3k-1)}$ is constructed, where, for $i=0,\dots,2k-1$, the graph $G^{(k+i)}$ is coming with its embedding on~$\T_0$, and with a set of special faces $\chi_{i+1},\psi_{i+2},\dots,\psi_{2k}$. A path $P_{i+1}$ between $\chi_{i+1}$ and $\psi_{i+2}$ is duplicated for merging these two faces into a single face~$\chi_{i+2}$, while preserving the fact that $G^{(k+i+1)}=G^{(k+i)}_{P_{i+1}}$ can be embedded on the sphere~$\T_0$. 
 
Eventually, the process results in a single face $\phi^*=\chi_{2k}$ of $H^*=G^{(3k-1)}$ (see Figure~\ref{fig:pierreT2}(e)). This face contains all duplicated vertices. The embedding $f$ of $G$ on $\T_k$ induces a planar embedding of $H^*$ whose external face is~$\phi^*$ (see Figure~\ref{fig:pierreT2}(f)). 

\subsubsection{Certifying Genus at Most~$k$}
\label{subsec:global_certificate}

Conversely, for a graph $G$ of genus~$k$, an embedding of  $G$ on $\T_k$ can be  induced from the embedding $f^*$ of $H^*$ on $\T_0$, and from the boundary walk~$B^*$  of~$\phi^*$. The latter is indeed entirely determined  by the successive cycle- and path-duplications performed during the whole process. It contains all duplicated vertices, resulting from the cycles $C'_1,\dots ,C'_k$ and $C''_1,\dots ,C''_k$, and from the paths $P'_1, \dots, P'_{2k-1}$ and $P''_1, \dots, P''_{2k-1}$.  Note that the duplication process for a vertex may be complex. A vertex may indeed be duplicated once, and then one of its copies may be duplicated again, and so on, depending on which cycle or path is duplicated at every step of the process.  This phenomenon actually already occurred in the basic case of the torus~$\T_1$ where the duplications of $v_i$ and $v_j$ were more complex that those of the other vertices, and were also differing depending on whether $i=j$ or not (see Section~\ref{subsec:cas-du-tore}). Figure~\ref{fig:pierreT2-complex} illustrates a case in which two cycles $C_i$ and $C_j$ share vertices and edges in $\T_2$, causing a series of duplication more complex than the basic case illustrated on Figure~\ref{fig:pierreT2}. In particular, a same vertex of $H^*$  may appear several times the boundary walk $B^*$, and a same edge of $H^*$ may be traversed twice, once in each direction.  

Let $H$ be a splitting of a graph~$G$, let $f$ be a planar embedding of~$H$, and let~$\phi$ be a face of~$H$ embedded on~$\T_0$. Let 
$
B=(u_0,\dots,u_N)
$
be a boundary walk of $\phi$, and let $\vec{B}$ be an arbitrary reference direction given to~$B$, say clockwise. Let $\mathcal{U}=(U_0,\dots,U_{3k-1})$ be a sequence of graphs such that $U_0=G$, $U_{3k-1}=H$, and, for every $i\in\{0,\dots,3k-2\}$, $U_{i+1}$ is a 2-splitting of $U_i$. The splitting of $U_i$ into $U_{i+1}$ is denoted by $\sigma_i=(\alpha_i,\beta_i)$. The following extend the notion of global consistency defined in the case of the torus~$\T_1$.  We say that $(G,H,\vec{B},\mathcal{U})$, is \emph{globally consistent} if the following two conditions hold. 
\begin{enumerate}
\item  \textbf{Path-duplication checking.}  Let $\chi_{2k}=\phi$, with directed boundary walk $\vec{B}(\chi_{2k})=\vec{B}$. For every $i=0,\dots,2k-1$, there exist faces $\chi_{i+1},\psi^{(i)}_{i+2},\dots,\psi^{(i)}_{2k}$ of $U_{k+i}$, with respective directed boundary walks $\vec{B}(\chi_{i+1}),\vec{B}(\psi^{(i)}_{i+2}),\dots,\vec{B}(\psi^{(i)}_{2k})$, and there exist vertices $u^{(i)}_1,\dots,u^{(i)}_t$, $v^{(i)}_1,\dots,v^{(i)}_r$,  $w'^{(i)}_0, \dots, w'^{(i)}_s$, and $w''^{(i)}_0, \dots, w''^{(i)}_s$ of~$U_{k+i}$ such that 
\begin{itemize}
\item $ \vec{B}(\chi_{i+1})= (w'^{(i)}_0, \dots, w'^{(i)}_s,  v^{(i)}_1,  \dots, v^{(i)}_r,w''^{(i)}_s,  \dots, w''^{(i)}_0, u^{(i)}_1,\dots, u^{(i)}_t)$; 

\item for every vertex $x \in V(U_{k+i})\smallsetminus ( \{w'^{(i)}_0, \dots, w'^{(i)}_s\}\cup \{w''^{(i)}_0, \dots, w''^{(i)}_s\})$,  $\sigma_{k+i-1}(x)=x$;

\item for every $j \in\{0,\dots,s\}$, $|\sigma_{k+i-1}^{-1}(\{w'^{(i)}_j,w''^{(i)}_j\})|=1$;

\item $\vec{B}(\chi_i)=(x,u^{(i)}_1,\dots,u^{(i)}_t,x)$ where $x=\sigma_{k+i-1}^{-1}(\{w'^{(i)}_0,w''^{(i)}_0\})$; 

\item $\vec{B}(\psi^{(i-1)}_{i+1})=(y,v^{(i)}_1,\dots,v^{(i)}_r,y)$  where $y=\sigma_{k+i-1}^{-1}(\{w'^{(i)}_s,w''^{(i)}_s\})$; 

\item for $j=i+2,\dots,2k$, $\sigma_{k+i-1}(\vec{B}(\psi^{(i-1)}_j))=\vec{B}(\psi^{(i)}_j)$.
\end{itemize} 

\item \textbf{Cycle  duplication checking.} Let $\phi'^{(k)}_1=\chi_1$, and, for $i=2,\dots,k$, let $\phi'^{(k)}_i=\psi^{(0)}_{2i-1}$. For $i=1,\dots,k$, let $\phi''^{(k)}_i=\psi^{(0)}_{2i}$.  For every $i=1,\dots,k$, there exists faces $\phi'^{(i)}_1,\phi''^{(i)}_1,\dots, \phi'^{(i)}_i,\phi''^{(i)}_i$ of $U_i$ with respective directed boundary walks $\vec{B}(\phi'^{(i)}_1),\vec{B}(\phi''^{(i)}_1),\dots, \vec{B}(\phi'^{(i)}_i),\vec{B}(\phi''^{(i)}_i)$ such that 
\begin{itemize}
\item $\vec{B}(\phi'^{(i)}_{i})=(v'_0,v'_1,\dots,v'_r,v'_0)$ and  $\vec{B}(\phi''^{(i)}_{i})=(v''_0,v''_r,v''_{r-1},\dots,v''_1,v''_0)$ for some $r\geq 2$, with $|\sigma^{-1}_{i-1}(\{v'_j,v''_j\})|=1$ for every $j=0,\dots,r$; 
\item for $j=1,\dots,i-1$, $\sigma_{i-1}(\vec{B}(\phi'^{(i-1)}_j))=\vec{B}(\phi'^{(i)}_j)$, and $\sigma_{i-1}(\vec{B}(\phi''^{(i-1)}_j))=\vec{B}(\phi''^{(i)}_j)$.
\end{itemize}
\end{enumerate}
By the construction performed in Sections~\ref{subsec:faceduplication} and~\ref{subsec:facereduction}, for every graph $G$ of genus~$k$, $(G,H^*,\vec{B}^*,\mathcal{U}^*)$ is globally consistent, where $\mathcal{U}^*=(G^{(0)},\dots,G^{(3k-1)})$. The following result generalizes Lemma~\ref{lem:certif-T1-centralized} to graphs of genus larger than~1. 

\begin{lemma}\label{lem:certif-Tk-centralized}
Let $H$ be a splitting of a graph $G$, and assume that there exists a planar embedding~$f$ of $H$ with a face $\phi$  and a boundary walk $B$ of $\phi$. Let $\mathcal{U}=(U_0,\dots,U_{3k-1})$ be a series of graphs such that $U_0=G$, $U_{3k-1}=H$, and, for every $i\in\{0,\dots,3k-2\}$, $U_{i+1}$ is a 2-splitting of $U_i$. 
If  $(G,H,\vec{B},\mathcal{U})$ is globally consistent, then $G$ can be embedded on the torus~$\T_k$.
\end{lemma}

\begin{proof}
Condition~1 in the definition of global consistency enables to  recover a collection $\psi_1,\dots,\psi_{2k}$ of faces of $U_{k}$. These faces are inductively constructed, starting from the face $\phi$ of the planar embedding~$f$ of $U_{3k-1}=H$. At each iteration~$i$ of the induction, $U_{k+i-1}$ has faces $\chi_{i},\psi^{(i-1)}_{i+1},\dots,\psi^{(i-1)}_{2k}$ obtained from the faces $\chi_{i+1},\psi^{(i)}_{i+2},\dots,\psi^{(i)}_{2k}$  of $U_{k+i}$ by separating the face  $\chi_{i+1}$ into two faces $\chi_{i}$ and $\psi^{(i-1)}_{i+1}$ connected by a path, while preserving the other faces $\psi^{(i)}_{i+2},\dots,\psi^{(i)}_{2k}$. This operation preserves planarity, and thus, in particular, $U_k$ is planar. 

The directions of the boundary walks of the faces  $\psi_1,\dots,\psi_{2k}$  are inherited from the original direction given to the boundary walk~$B$. Condition~2 enables to iteratively merge face $\psi_{2i}$ with face $\psi_{2i-1}$, $i=1,\dots,k$, by identifying the vertices of their boundary walks while respecting the direction of these walks, which guarantees that handles are created (and not a Klein-bottle-like construction). The process eventually results in the graph $U_0$ with $k$ handles, providing an embedding of $U_0=G$ on~$\T_k$. 
\end{proof}

Thanks to Lemma~\ref{lem:certif-Tk-centralized}, the overall outcome of this section is that the tuple
\[
c=(H^*,f^*,\phi^*,B^*,\mathcal{U}^*)
\]
constructed in Sections~\ref{subsec:faceduplication} and~\ref{subsec:facereduction} is a certificate that $G$ can be embedded on~$\T_k$. This certificate~$c$ and its corresponding verification algorithm are  however centralized. In the next section, we show how to distribute both the certificate~$c$, and the verification protocol.

\section{Proof-Labeling Scheme for Bounded Genus Graphs}
\label{sec:tools}

In this section, we establish our first main result. 

\begin{theorem}\label{theo:main1}
Let $k\geq 0$, and let $\G^+_k$ be the class of graphs embeddable on an orientable closed surface of genus at most~$k$. There is a proof-labeling scheme for $\G^+_k$ using certificates on $O(\log n)$ bits in $n$-node graphs. 
\end{theorem}

The proof essentially consists of showing how to distribute the centralized certificate 
\[
(H,f,\phi,B,\mathcal{U})
\]
used in Lemma~\ref{lem:certif-Tk-centralized} for a graph $G$, by storing $O(\log n)$ bits at each vertex of~$G$, while allowing the vertices to locally verify the correctness of the distributed certificates, that is, in particular, verifying that $(G,H,B,\mathcal{U})$ is globally consistent. The rest of the section is entirely dedicated to the proof of Theorem~\ref{theo:main1}.  We start by defining the core of the certificates assigned to the nodes, called \emph{histories}. Then, we  show how to distribute the histories so that every node stores at most $O(\log n)$ bits, and we describe the additional information to be stored in the certificates for enabling the liveness and completeness properties of the verification scheme to hold. Recall that the nodes of $G$ are given arbitrary distinct IDs picked from a set of polynomial range. The ID of node~$v\in V(G)$ is denoted by $\id(v)$. Note that $\id(v)$ can be stored on $O(\log n)$ bits. 

\subsection{Histories}
\label{sec:histories}

The description of the certificates is for positive instances, that is, for graphs $G\in \G^+_k$. For such an instance~$G$, the  prover performs the construction of Section~\ref{sec:unfolding-any-Tk}, resulting in the series of 2-splitting graphs $G^{(0)}=G,G^{(1)},\dots,G^{(2k-2)}, G^{(2k-2)}=H^*$, a planar embedding $f$ of $H^*$, and the identification of a special face $\phi^*$ in this embedding, with boundary walk~$B^*$. The successive duplications experienced by a vertex~$v$ of the actual graph $G$ during the face-duplication and face-reduction phases resulting in $H^*$ can be encoded as a rooted binary tree unfolding these duplications, called \emph{history}. 

For every vertex~$v$ of~$G$, the history of~$v$ is denoted by~$\hist(v)$. The history of $v$ is a rooted binary tree of depth $3k-1$ (all leaves are at distance $3k-1$ from the root). For $\ell=0,\dots,3k-1$, the level~$\ell$ of $\hist(v)$ consists of the at most $2^\ell$ nodes at distance~$\ell$ from the root. The internal nodes of $\hist(v)$ with two children are call \emph{binary} nodes, and the internal nodes with one child are called \emph{unary}. 
\begin{itemize}
\item  For $\ell=0,\dots,k-1$, the edges connecting nodes of level~$\ell$ to nodes of level~$\ell+1$ are corresponding to the duplication of the cycle $C_{\ell+1}$ in $G^{(\ell)}$ (cf. Section~\ref{subsec:faceduplication}), and, 
\item for $\ell=0,\dots,2k-1$, the edges connecting nodes of level~$k+\ell$ to nodes of level~$k+\ell+1$ are corresponding to the duplication of the path $P_{\ell+1}$ in $G^{(k+\ell)}$  (cf. Section~\ref{subsec:facereduction}).
\end{itemize}
The nodes of~$\hist(v)$ are provided with additional information, as follows. 

\subsubsection{Vertices and Adjacencies in the Splitting Graphs} 

For every $\ell=1,\dots,3k-1$, every node $x$ at level~$\ell$ in $\hist(v)$ is provided with the vertex $u$ of $G^{(\ell)}$ it corresponds to, after the duplications of~$v$ corresponding to the path from the root to~$x$. In particular,  each leaf of $\hist(v)$ is provided with the single vertex of $H^*=G^{(3k-1)}$ it corresponds to. Specifically, each internal node~$x$ of $\hist(v)$ is provided with the set $S_x$ of vertices of $H^*$ marked at the leaves of the subtree of $\hist(v)$ rooted at~$x$. For a leaf $x$, $S_x=\{u\}$, where $u$ is the avatar of $v$ in $H^*$ corresponding to the path from the root to the leaf~$x$. Note that, for two distinct nodes at level $\ell$ in $\hist(v)$, we have  $S_x \cap S_y=\varnothing$. 

The $3k-1$ splittings successively performed starting from~$G$ are 2-splittings, from which it follows that every vertex of $G$ is split a constant number of times for a fixed~$k$. The $\nu\geq 1$ avatars of $v\in V(G)$ is $H^*$ are labeled $(\id(v),1),\dots,(\id(v),\nu)$. It follows that the $\nu$ leaves of $\hist(v)$ are respectively labeled $(\id(v),1),\dots,(\id(v),\nu)$. For every node $x$ of $\hist(v)$, each set $S_x$ is a subset of $\{(\id(v),1),\dots,(\id(v),\nu)\}$, and thus these sets $S_x$ can be stored on $O(\log n)$ bits.  

Every node $x$ of $\hist(v)$ at level $\ell\in\{0,\dots,3k-1\}$, which, as explained above, corresponds to a vertex of $G^{(\ell)}$, is also provided with the set $N_x$ of the neighbors of $S_x$ in~$G^{(\ell)}$. The set $N_x$ has the form $N_x=\{X_1,\dots,X_d\}$ for some $d\geq 1$, where, for $i=1,\dots,d$, $X_i$ is a vertex of $G^{(\ell)}$ corresponding to a set of avatars in $H^*$ of some neighbor $w$ of $v$ in~$G$. 

Since some vertices $v\in V(G)$ may have arbitrarily large degree (up to $n-1$), the sets $N_x$ may not be storable using $O(\log n)$ bits. As a consequence, some histories may not be on $O(\log n)$ bits, and may actually be much bigger. Nevertheless, a simple trick using the fact that graphs with bounded genus have bounded degeneracy (cf. Lemma~\ref{genus-implies-degeneracy})  allows us to reassign locally the set $N_x$ in the histories so that every node of $G$ stores $O(\log n)$ bits only. 

\subsubsection{Footprints }

Every node $x$ of $\hist(v)$ at level $\ell\in\{0,\dots,3k-1\}$ is provided with a (possibly empty) set $F_x$ of ordered triples of the form $(X,Y,Z)$ where $X\in N_x$, $Y=S_x$, and $Z\in N_x$, called \emph{footprints}. Intuitively, each footprint encodes edges $\{X,Y\}$ and $\{Y,Z\}$ of $G^{(\ell)}$ occurring in:
\begin{itemize}
\item  a boundary walk of one of the faces $\phi'_i$ or $\phi''_i$, $i=1,\dots,\ell$, if $\ell\leq k$, or
\item  a boundary walk of one of the faces $\chi_{\ell-k},\psi_{\ell-k+1},\dots,\psi_{2k}$, otherwise.
\end{itemize}
Note that these two edges are actually directed, from  $X$ to $Y$, and from $Y$ to $Z$, reflecting that the boundary walk is traveled in a specific direction, inherited from some a priori direction, say clockwise, given to the boundary walk $B^*$ of the face $\phi^*=\chi_{2k}$ (hence the terminology ``footprints''). 

Note that a same vertex of $G^{(\ell)}$ may appear several times in the boundary walk of a face, and a same edge may appear twice, once in every direction. Therefore, a same node $x$ of $\hist(v)$ may be provided with several footprints, whose collection form the set~$F_x$, which may be of non-constant size. On the other hand, for a fixed~$k$, a constant number of boundary walks are under concern in total, from which it follows that even if a node $x$ at level $\ell$ of $\hist(v)$ must store a non-constant number of footprints in~$F_x$, each of $x$'s incident edges in~$G^{(\ell)}$ appears in at most  two footprints of $F_x$. We use this fact, together with the bounded degeneracy of the graphs of bounded genus, for reassigning locally the sets $F_x$ in the histories so that every node of $G$ stores $O(\log n)$ bits only. 

\subsubsection{Types } 

Last, but not least, for every node $x$ of $\hist(v)$, each of the two (directed) edges $(X,Y)$ and $(Y,Z)$ in every footprint $(X,Y,Z)$ in $F_x$ also comes with a \emph{type} in 
\[
\type_k=\{C'_1,\dots,C'_k,C''_1,\dots,C''_k,P'_1,\dots,P'_{2k-1},P''_1,\dots,P''_{2k-1}\},
\]
which reflects when this edge was created during the cycle- and path-duplications. 

\paragraph{Example.} 

Figure~\ref{fig:exemple-historie} provides examples of histories for some vertices of $G$ in the case displayed on Figure~\ref{fig:cutpathQ}. Figure~\ref{fig:exemple-historie}(a-b) display the histories of $v_i$ and $v_j$ whenever $i\neq j$, while Figure~\ref{fig:exemple-historie}(c)  displays the histories of $v_i=v_j$ whenever $i=j$. In this latter case, the leaves $w'_0,w''_0,w'_s,w''_s$ may be labeled as $(\id(v),1),(\id(v),2),(\id(v),3),(\id(v),4)$, respectively. Then $v'_i$ is labeled $S_{v'_i}=\{(\id(v),1),(\id(v),2)\}$, while $v''_i$ is labeled $S_{v''_i}=\{(\id(v),3),(\id(v),4)\}$, and root is labeled $S_{v_i}=\{(\id(v),1),(\id(v),2),(\id(v),3),(\id(v),4)\}$. The neighborhoods $N_x$ of these nodes $x$ of $\hist(v_i)$ are depending on the graphs $G^{(0)}=G,G^{(1)}$, and $G^{(2)}=H^*$. Assuming that $B^*$ is directed clockwise, as displayed on Figure~\ref{fig:cutpathQ}(d), the leaf $w'_0$ is provided with footprint $(v'_{i-1},w'_0,w'_1)$ while the  leaf $w''_0$ is provided with footprint $(w''_1,w''_0,v'_{i+1})$. Similarly, $w'_s$ and $w''_s$ are respectively provided with footprint $(w'_{s-1},w'_s,v''_{i-1})$ and $(v''_{i+1}, w''_s,w''_{s-1})$, where the various nodes in these footprints are encoded depending on their labels in $H^*$, which depend on the IDs given to the neighbors of~$v_i$ in~$G$. The footprint at $v'_i$ is $(v'_{i-1},v'_i,v'_{i+1})$, while the footprint at $v''_i$ is $(v''_{i+1},v''_i,v'_{i-1})$. In both case, the directions of the edges are inherited from the initial clockwise direction of the boundary walk~$B^*$. The directed edges $(v'_{i-1},v'_i)$ and $(v'_i,v'_{i+1})$ receives type $C'_1$, while the directed edges $(v''_{i+1},v''_i)$ and $(v''_i,v''_{i-1})$ receives type $C''_2$. The four edges $(v'_{i-1},w'_0)$,  $(w''_0,v'_{i+1})$, $(v''_{i+1},w''_s)$, and $(w'_s,v''_{i-1})$ are respectively inheriting the types $C'_1,C'_1,C''_1$, and $C''_1$ of the four edges $(v'_{i-1},v'_i)$, $(v'_i,v'_{i+1})$, $(v''_{i+1},v''_i)$, and $(v''_i,v''_{i-1})$. The directed edges $(w'_0,w'_1)$ and $(w'_{s-1},w'_s)$ receive type~$P'_1$, while the directed edges $(w''_1,w''_0)$ and $(w''_s,w''_{s-1})$ receive type~$P''_1$. Observe that the footprints are constructed upward the histories, while the types are assigned downward those trees. 

\begin{figure}[tb]
\input{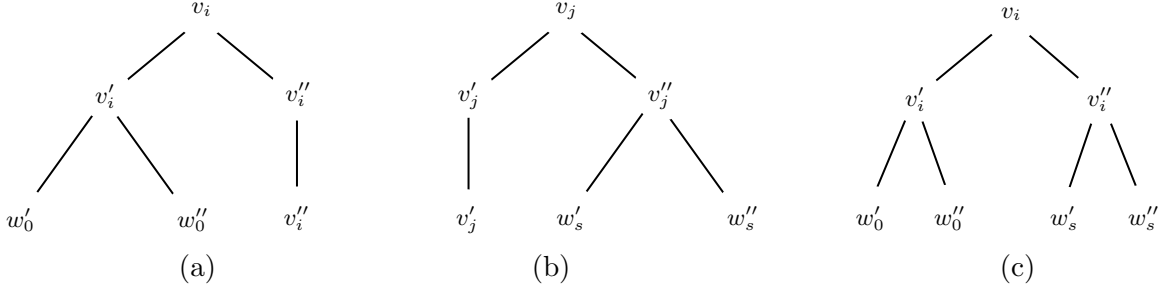}
\caption{Examples of histories.}
\label{fig:exemple-historie} 
\end{figure} 

\medskip

We now detail how  the footprints are constructed in general, and how the types are assigned to the edges of the footprints.

\subsubsection{Construction of the Footprints}

 Let us give an arbitrary orientation, say clockwise, to the boundary walk $B^*$ of the special face $\phi^*$ of $H^*$. This orientation induces footprints $(\p(u),u,\s(u))\in F_x$ given to every leaf~$x$ of every history~$\hist(v)$, $v\in V(G)$. The vertex $\p(u)\in V(H^*)$ is the predecessor of the avatar $u\in V(H^*)$ of $v$ in~$H^*$, and $\s(u) \in V(H^*)$ is its successor. Note that some leaves $x$ have $F_x=\varnothing$, whenever the corresponding node $u$ in $H^*$ does not belong to the boundary walk $B^*$. On the other hand, as a same node can be visited several times when traveling along the boundary walk $B^*$, some leaves may be given several footprints  $(\p_1(u),u,\s_1(u)),\dots,(\p_d(u),u,\s_d(u))$ in $F_x$, for some $d\geq 1$. The footprints provided to the internal nodes of the histories of the vertices of~$G$ are given in a way consistent with the orientation of~$B^*$. More specifically, the footprints are constructed upward the histories, as follows. 

Hereafter, the symbol ``$\xrightarrow{\ell}$'' stands for the operation performed when going from level $\ell-1$ to level~$\ell$, or vice-versa, from level $\ell$ to level $\ell-1$. For instance, for three sets $S,S',S''$ of vertices from~$H^*$, the relation 
\[
S \xrightarrow{\ell} S',S''
\]
states that the vertices $S'$ and $S''$ of $G^{(\ell)}$ are the results of a cycle- or path-duplication experienced by the vertex $S$ occurring from $G^{(\ell-1)}$ to $G^{(\ell)}$, i.e., the vertex $S=S'\cup S''$ of $G^{(\ell-1)}$ is split into two avatars, $S'$ and $S''$, in $G^{(\ell)}$. If $\ell\leq k$, the split was caused by a cycle-duplication, otherwise it was caused by a path-duplication. Similarly, for two footprints $F'$ and $F''$ at two nodes at level $\ell$, children of a same binary node, the relation
\[
F' , F'' \xrightarrow{\ell} F
\]
states that, when going upward an history, the two footprints $F'$ and $F''$ of level $\ell$ generate the footprint  $F$ at level $\ell-1$.  

Three rules, called \emph{Elementary, Extremity}, and \emph{Vacancy}, are applied for the construction of the footprints. Their role is to ``role back'' the boundary walk $B^*$ of the special face~$\phi^*$ in the planar embedding of~$H^*$. Each edge of the boundary walk $B^*$ is indeed resulting from some duplication, of either a cycle or a path. The footprints encode the histories of all edges of the boundary walk  $B^*$ in all graphs $G^{(\ell)}$, $0\leq \ell \leq 3k-1$, including when the edges were created (referred to as the \emph{types} of the edges), and what were their successive extremities when those extremities are duplicated. 

\begin{description}
\item[Elementary rule.]  Assuming $X \xrightarrow{\ell} X' ,X'' , \; Y \xrightarrow{\ell} Y', Y'', \; \mbox{and} \; Z \xrightarrow{\ell} Z',Z''$, the elementary rule matches two footprints of two children $Y'$ and $Y''$, and produces none at the parent~$Y$: 
\[
(X',Y',Z'), \; (Z'',Y'',X'') \xrightarrow{\ell}  \bot.
\]
The Elementary rule applies to the case of cycle duplication, as well as to the case of path-duplication, but to the internal nodes of the path only (see Figure~\ref{fig:elementary}). When two cycles are merged (as the opposite to cycle duplication), their faces are glued together, and  disappear. Similarly, when two paths are merged (as the opposite of path-duplication), the resulting path is of no use, and it can be discarded. Note that the two footprints $(X',Y',Z')$ and $(Z'',Y'',X'')$ are ordered in opposite directions. This matches the requirement for correctly glueing the borders of two faces in order to produce a handle (see Figures~\ref{fig:hcc} and~\ref{fig:duplication-simple}). This also matches the way the two copies of a path $P_i$ are traversed when traveling along the boundary walk $B^*$ in clockwise direction (cf. Eq.~\eqref{eq:boundwalk} and Figure~\ref{fig:path-duplication-simple}). 

\begin{figure}[tb]
\begin{center}
\input{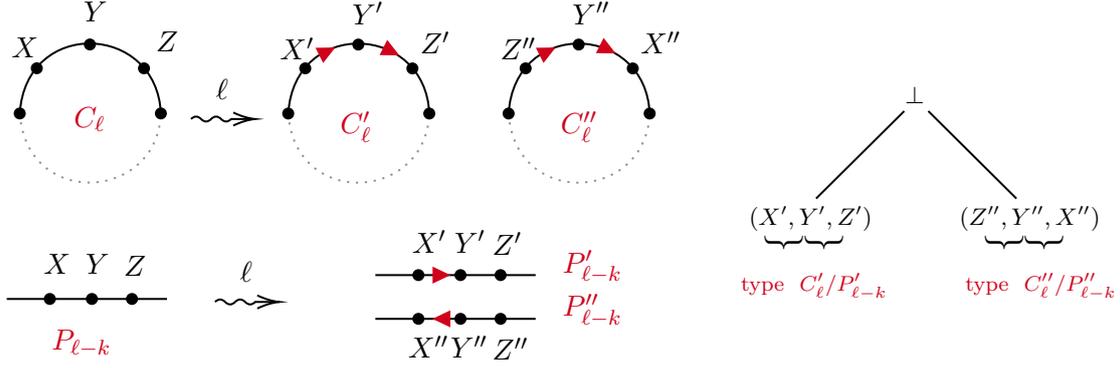}
\caption{Footprint construction, and type assignment: Elementary rule.}
\label{fig:elementary} 
\end{center}
\end{figure} 

\item[Extremity rule.]  This rule applies only for levels $\ell>k$. It has two variants, defined below. 
\begin{description}
\item[{\it Single extremity rule.}] Assuming $X'\xrightarrow{\ell} X', \;  X''\xrightarrow{\ell} X'', \; Y \xrightarrow{\ell} Y,|Y''$, and $Z \xrightarrow{\ell} Z',Z''$, the single extremity rule matches two footprints of two children $Y'$ and $Y''$, and produces one footprint at the parent $Y$: 
\[
 (X',Y',Z'), \;  (Z'',Y'',X'')  \xrightarrow{\ell}  (X',Y,X'').
\] 
\item[{\it Double extremity rule.}] Assuming $X'\xrightarrow{\ell} X', \; X''\xrightarrow{\ell} X'', \; Y \xrightarrow{\ell} Y',Y'', \; Z' \xrightarrow{\ell} Z'$, and $Z'' \xrightarrow{\ell} Z''$,  the double extremity rule matches two footprints of two children $Y'$ and $Y''$, and produces two footprints at the parent $Y$: 
\[
(X',Y',Z'), \;  (Z'',Y'',X'')  \xrightarrow{\ell}  \big \{(X',Y,X''), (Z'',Y,Z')  \big \}.
\]
The Extremity rule refers to path duplication only (i.e., to levels $\ell>k$), as displayed on Figure~\ref{fig:extremity}. It is dedicated to the extremities of the path considered at this phase (see Figure~\ref{fig:path-duplication-simple}). The Single extremity rule (cf. Figure~\ref{fig:extremity}(a)) handles the standard case in which the path is not trivial (i.e., reduced to a single vertex), whereas the Double extremity rule  (cf. Figure~\ref{fig:extremity}(b))  handles the case in which the path connecting two faces is reduced to a single vertex~$Y$ (i.e., the two corresponding cycles share at least one vertex~$Y$). Then only the vertex~$Y$ is split during the path duplication, while its four neighbors $X',X'',Z'$, and $Z'$ remain intact. 
\end{description}

\begin{figure}[tb]
\centering
\input{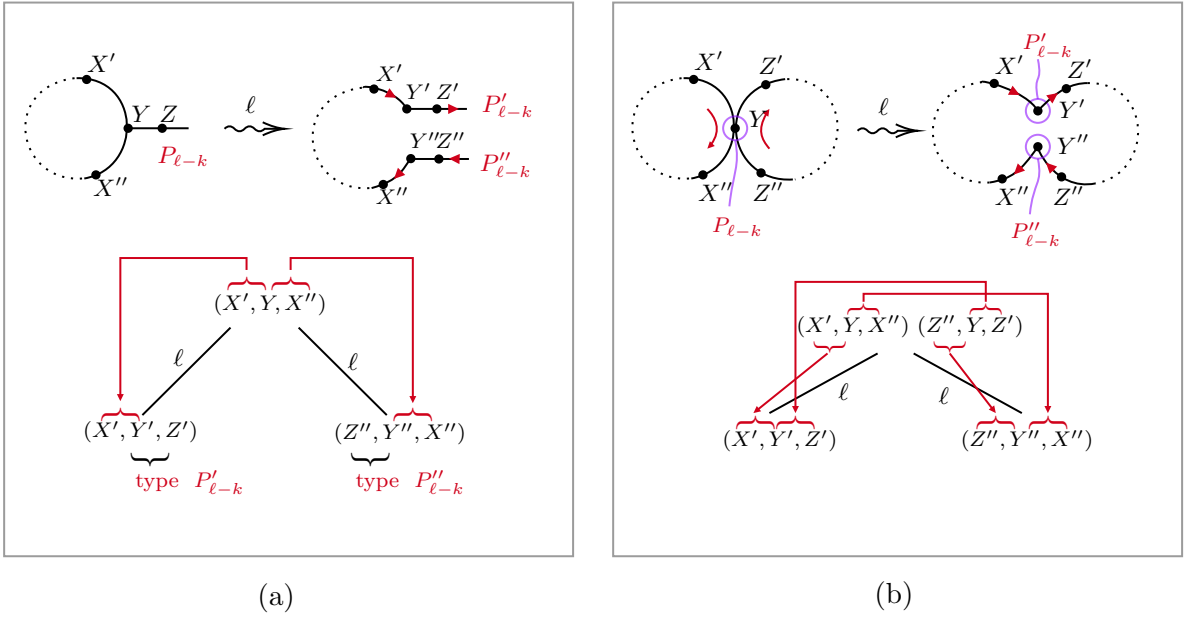}
\caption{Footprint construction, and type assignment: Extremity rule.}
\label{fig:extremity} 
\end{figure} 

\item[Vacancy rule.]  The vacancy rule simply forwards  a footprint upward: 
\[
(X',Y,'Z') \xrightarrow{\ell}   (X,Y,Z)
\]
with $X \xrightarrow{\ell} X',X''$ (resp., $Y  \xrightarrow{\ell} Y',Y''$, and $Z \xrightarrow{\ell} Z' ,Z''$), unless $X \xrightarrow{\ell} X$ (resp., $Y \xrightarrow{\ell} Y$, and $Z \xrightarrow{\ell} Z$), in which case $X=X'$ (resp., $Y=Y'$, and $Z=Z'$). 

The Vacancy rule handles the case where one of the twin nodes carries a footprint $(X',Y',Z')$ (resp., $(X'',Y'',Z'')$), which is copied to the parent node, after updating the vertices in case the latter experienced duplications (see Figure~\ref{fig:vacancy}). 

\begin{figure}[tb]
\centering
\input{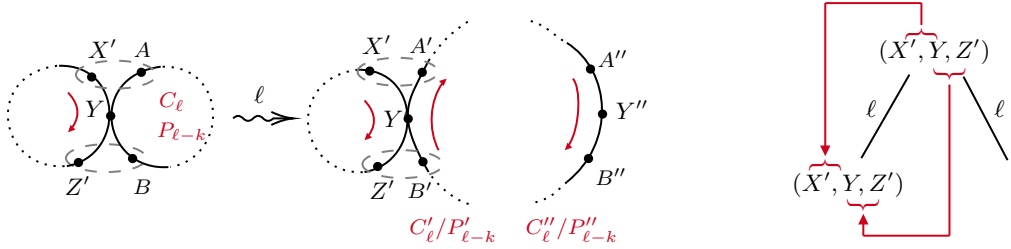}
\caption{Footprint construction, and type assignment: Vacancy rule.}
\label{fig:vacancy} 
\end{figure} 

\end{description}

\subsubsection{Assigning Types to Footprints} 

The types in $\type_k$ are assigned to the edges of the footprints, downwards the histories, as follows. 
\begin{itemize}
\item If the footprints $(X',Y',Z')$ and $(Z'',Y'',X'')$ are matched by application of the Elementary rule at level $\ell$, then the two (directed) edges $(X',Y')$ and $(Y',Z')$ (resp., $(Z'',Y'')$ and $(Y'',X'')$) of $G^{(\ell)}$  are given type $C'_\ell$ (resp., $C''_\ell$) if $\ell\leq k$, and $P'_{\ell-k}$ (resp. $P''_{\ell-k}$) otherwise. See Figure~\ref{fig:elementary}. 
\item If the footprints $(X',Y',Z')$ and $(Z'',Y'',X'')$ are matched by application of the Single extremity rule at level $\ell$, then the two edges  $(X',Y)$ and $(Y,X'')$ adopt the types of the edges  $(X',Y')$ and $(Y'',X'')$, respectively, while the two edges $(Y',Z')$ and $(Z'',Y'')$ are given type~$P'_{k-\ell}$ and $P''_{k-\ell}$, respectively. See Figure~\ref{fig:extremity}(a). 
\item If the footprints $(X',Y',Z')$ and $(Z'',Y'',X'')$ are matched by application of the Double extremity rule at level $\ell$, then the four edges $(X',Y')$, $(Y',Z')$, $(Z'',Y'')$, and $(Y'',X'')$ adopt the types of the edges $(X',Y)$, $(Y,Z')$, $(Z'',Y)$, and $(Y,X'')$, respectively. See Figure~\ref{fig:extremity}(b)
\item If the footprint $(X',Y,'Z')$ is forwarded upward as $(X,Y,Z)$ by application of the Vacancy rule, then  $(X',Y')$, and $(Y',Z')$ adopt the types of the edges $(X,Y)$, and $(Y,Z)$, respectively. See Figure~\ref{fig:vacancy}.
\end{itemize}

\bigbreak

We have now all the ingredients to state what will be proved as sufficient to certify that a graph $G$ has genus at most~$k$. 

\subsection{Assignment of the Histories to the Certificates}
\label{sec:linePLS}

As it was mentioned in Section~\ref{sec:histories}, the history $\hist(v)$ of a node $v$ of the actual graph $G$ may not be on $O(\log n)$ bits. The reason for that is that, even if $G$ has bounded genus~$k$, the node $v$ may have an arbitrarily large degree. As a consequence, the sum of  the degrees of its $v$'s avatars in each of the graphs $G^{(0)}, \dots,G^{(3k-1)}$ may be arbitrarily large. This has direct consequences not only on the memory requirement for storing the neighborhood $N_x$ of each node $x\in \hist(v)$, but also on the number of footprints to be stored in~$F_x$. In both cases, this memory requirement may exceed $O(\log n)$ bits. On the other hand, every graph $G$ of bounded genus is sparse, which implies that the average degree of~$G$, and of all its splitting graphs $G^{(0)}, \dots,G^{(3k-1)}$ is constant. Therefore, the average memory requirement per vertex~$v$ for storing all the histories~$\hist(v)$, $v\in V(G)$, is constant. Yet, it remains that some vertices~$v\in V(G)$ may have large histories, exceeding $O(\log n)$ bits. 

The simple trick under this circumstances (cf., e.g., \cite{FeuilloleyFRRMT}) is to consider the space-complexity of the histories not per node of~$G$, but per edge. Indeed, the space-complexity of the information related to each edge $e$ of $G$, as stored in the histories, is constant, for \emph{every} edge~$e$. For instance, at a node~$x$ of level $\ell$ in some historie~$\hist(v)$, instead of storing $N_x$ at~$v$, one could virtually store every edge $\{S_x,S_y\}$, $S_y\in N_x$, on the edge $\{v,w\}$ of~$G$, where~$w$ is the neighbor of~$v$ in~$G$ with avatar~$S_y$ in~$G^{(\ell)}$. 

Let us define a \emph{line} proof-labeling scheme as a proof-labeling scheme in which certificates are not only assigned to the vertices of~$G$, but also to the edges of~$G$ (i.e., to vertices of the line-graph of~$G$). In a line proof-labeling scheme, the vertices forge their decisions not only on their certificates and on the certificates assigned to their adjacent vertices, but also on the certificates assigned to their incident edges. Our interest for the concept of line proof-labeling scheme is expressed in the following result, after having recalled that, thanks to Lemma~\ref{genus-implies-degeneracy}, every graph of genus at most $k$ is $d$-degenerate for some constant~$d$ depending on~$k$. 
 
\begin{lemma}\label{lem:line-PLS}
Let $f:\N\to \N$ such that $f(n)\in \Omega(\log(n))$. Let $d\geq 1$, and let $\G$ be a graph family such that every graph in $\G$ is $d$-degenerate. If $\G$ has a line proof-labeling scheme with certificate size $O(f(n))$ bits, then $\G$ has a proof-labeling scheme with certificate size $O(f(n))$ bits.
\end{lemma}

\begin{proof}
Let $(\mathbf{p},\mathbf{v})$ be line proof-labeling scheme for $\G$. For $G\in \G$, the prover $\mathbf{p}$ assigns certificate $\mathbf{p}(v)$ to every node $v\in V(G)$, and certificate $\mathbf{p}(e)$ to every edge $e\in V(G)$. Since $G$ is $d$-degenerate, there exists a node $v$ of $G$ with degree $d_v\leq d$. Let $c(v)$ be the certificate  of $v$ defined as 
\[
c(v)=\Big(\mathbf{p}(v),\big \{(\id(u_1),\mathbf{p}(e_1)),\dots, (\id(u_{d_v}),\mathbf{p}(e_{d_v})\big\}\Big),
\]
where $u_1,\dots,u_{d_v}$ are the $d_v$ neighbors of $v$ in~$G$, and, for every $i=1,\dots,d_v$, $e_i=\{v,u_i\}$. Since the IDs can be stored on $O(\log n)$ bits, and since $f(n)\in\Omega(\log n)$, we get that $c(v)$ can be stored on $O(f(n))$ bits. This construction can then be repeated on the graph $G'=G-v$, which still has degeneracy at most~$d$. By iterating this construction, all nodes are exhausted, and assigned certificates on $O(f(n))$ bits, containing all the information originally contained in the node- and edge-certificates assigned by~$\mathbf{p}$. We complete the proof by observing that, for every edge~$e=\{u,v\}$ of~$G$, the certificate  $\mathbf{p}(e)$ assigned by $\mathbf{p}$ to~$e$ can be found either in $c(u)$ or in~$c(v)$. This suffices for simulating the behavior of $\mathbf{v}$, and thus for the design of a standard proof-labeling scheme for~$\G$. 
\end{proof}

\subsection{Certifying Planarity}\label{subsec:planarity}

In this section, we show how to certify that $H$ is a planar embedding with a special face~$\phi$ with boundary walk~$B$. For this purpose, we just need to slightly adapt a recent proof-labeling scheme for planarity~\cite{FeuilloleyFRRMT}.

\begin{lemma}\label{lem:PODC2020}
There exists a proof-labeling scheme for certifying that a given graph $H$ has a planar embedding~$f$, including a face~$\phi$ with boundary walk~$B$. 
\end{lemma}

\begin{figure}[tb]
\centering
\input{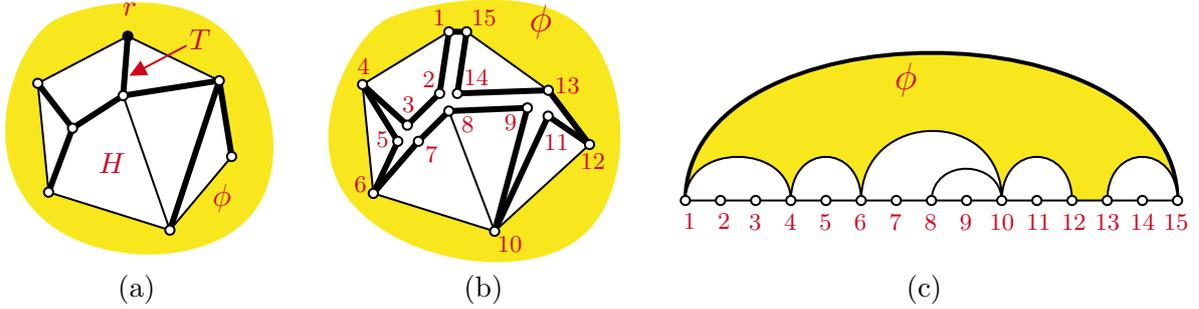}
\caption{Illustration of the PLS for planarity in~\cite{FeuilloleyFRRMT}. }
\label{fig:PLS-planar}
\end{figure}

\begin{proof}
Let $H$ be a planar graph with a planar embedding~$f$.  The scheme for planarity in~\cite{FeuilloleyFRRMT} constructs the certificates as follows (cf. Figure~\ref{fig:PLS-planar}). Let $T$ be an arbitrary spanning tree of~$H$, and let us root $T$ at a vertex $r\in V(H)$ on the outer face~$\phi$, as displayed on Figure~\ref{fig:PLS-planar}(a).  The tree $T$ is ``flattened'' into a cycle $C$ in a splitting $H'$ of $H$ by replacing every vertex $v\in V(H)$ by as many vertices as the number of times $v$ is visited by a DFS traversal of $T$ starting from~$r$ (see Figure~\ref{fig:PLS-planar}(b)). The scheme in~\cite{FeuilloleyFRRMT} certifies the cycle $C$, viewed as a path $P$ whose two extremities are avatars of~$r$, with respective DFS numbers~1 and~$2n-1$, plus an edge connecting these two avatars (see Figure~\ref{fig:PLS-planar}(c)). A property of this construction taken from~\cite{FeuilloleyFRRMT} is that the vertices of $H$ on the outer face~$\phi$ are those which have at least one avatar in $H'$ such that no co-tree edges ``jumps over it'' when the vertices are displayed as on Figure~\ref{fig:PLS-planar}(c). For instance, the avatars $1,4,6,10,12,13,15$ have no co-tree edges jumping over them, and indeed these avatars are the ones of the vertices on the boundary of the outer face $\phi$.  The scheme of \cite{FeuilloleyFRRMT} is precisely based on a local encoding of the ``lower edge'' jumping over every avatars in~$H'$. It follows that this scheme suffices for certifying not only the planarity of~$H$, but also that $\phi$ is a face of~$H$ with boundary~$B$. 
\end{proof}

\subsection{Local Consistency}
\label{subsec:consistent}

Let  $H$ be a splitting of a graph $G$, let~$f$  be a planar embedding of $H$, and let $\phi$ be a face of $H$ with  boundary walk $B$ directed, say, clockwise. The directed boundary walk~$B$ is denoted by~$\vec{B}$. Let 
$
\hist(G)=\{\hist(v),v\in V(G)\}
$
be a collection of histories for the vertices of~$G$, of depth $3k-1$, for some $k\geq 1$.  We say that $(G,H,\vec{B},\hist(G))$  is \emph{locally consistent} if the following holds. 

\begin{enumerate}
\item \label{consistent:1} There exists a sequence of graphs $U_0,\dots,U_{3k-1}$ with $U_0=G$, $U_{3k-1}=H$, and, for every $0\leq \ell < 3k-1$, $U_{\ell+1}$ is a degree-2 splitting of $U_\ell$, such that, for every $v\in V(G)$, and for every $\ell=0,\dots,3k-1$, every node $x$ at level $\ell$  of $\hist(v)$ satisfies that $S_x$ is a vertex of $U_\ell$, the neighborhood of $S_x$ defined in $N_x$ is consistent with the neighborhood of $S_x$ in $U_\ell$, and the footprints in $F_x$ contains edges of $U_\ell$. Moreover, if $x$ has two children $x'$ and $x''$ in $\hist(v)$, then there are exactly two footprints, one in $E_{x'}$ and one in $E_{x''}$, for which the Elementary rule or the Extremity rule was applied, all the other footprints in $E_{x'}$ and $E_{x''}$ being subject to the Vacancy rule. Furthermore, if $x$ has a unique child $x'$, then all footprints in  $E_{x'}$ are subject to the Vacancy rule. Finally, the typing is consistent with the specified typing rules. 

\item \label{consistent:2} The collection of footprints at the leaves of the histories in $\hist(G)$ can be ordered as $(x_0,y_0,z_0)$, $\dots,$ $(x_N,y_N,z_N)$ such that,  $y_i=z_{i-1}=x_{i+1}$ for every $i=0,\dots,N$, and $\vec{B}=(y_0,\dots,y_N)$.

\item \label{consistent:3}  For every $\ell=1,\dots,,2k-1$, the following must be satisfied: 
\begin{enumerate}
\item the collection of footprints at the nodes at level $k+\ell$ whose both edges have type $P'_{\ell}$ (resp.,  type $P''_{\ell}$) in the histories in~$\hist(G)$ can be ordered as $(X'_0,Y'_0,Z'_0),\dots,(X'_{s_\ell},Y'_{s_\ell},Z'_{s_\ell})$ (resp., $(Z''_0,Y''_0,X''_0),\dots,(Z''_{s_\ell},Y''_{s_\ell},X''_{s_\ell})$), for some $s_\ell\geq 0$, such that:
\begin{enumerate}
\item for every $i=0,\dots,s_\ell$, $Y_i \xrightarrow{k+\ell} \{Y'_i,Y''_i\}$; 
\item for every $i=1,\dots,s_\ell$, $Y'_i=Z'_{i-1}$ and $Y''_i=Z''_{i-1}$; 
\item for every $i=0,\dots,s_\ell-1$, $Y'_i=X'_{i+1}$ and $Y''_i=X''_{i+1}$; 
\end{enumerate}
\item  the collection of footprints at the nodes at level $k+\ell$ whose both edges have  type $C'_{\lceil \frac{\ell+1}{2}\rceil}$ if $\ell+1$ is odd, or type $C''_{\frac{\ell+1}{2}}$ if $\ell+1$ is even, can be ordered as $(X_0,Y_0,Z_0),\dots,$ $(X_{r_\ell},Y_{r_\ell},Z_{r_\ell})$, for some $r_\ell\geq 0$ such that, for every $i=1,\dots,r_\ell$, $Y_i=Z_{i-1}=X_{i+1}$; 
\item the collection of footprints at the nodes at level $k+\ell$ whose both edges have same type $P'_1,P''_1,\dots,P'_{\ell-1}, P''_{\ell-1},C'_1,C''_1,\dots,C'_{\lceil\ell/2\rceil},C''_{\lceil\ell/2\rceil}$, or $C'_{(\ell+1)/2}$ if $\ell+1$ is even, can be ordered as $(X_0,Y_0,Z_0),\dots,(X_{t_\ell},Y_{t_\ell},Z_{t_\ell})$, for some $t_\ell\geq 0$, such that for every $i=1,\dots,t_\ell$, $Y'_i=Z'_{i-1}$ and $Y''_i=Z''_{i-1}$; 
\end{enumerate}

\item  \label{consistent:4} For every $\ell=1,\dots,k$, the collection of footprints at the nodes at level $\ell$ whose both edges have type $C'_{\ell}$ (resp.,  type $C''_{\ell}$) in the histories in~$\hist(G)$ can be ordered as $(X'_0,Y'_0,Z'_0),\dots,$ $(X'_{r_\ell},Y'_{r_\ell},Z'_{r_\ell})$ (resp., $(Z''_0,Y''_0,X''_0),\dots,(Z''_{r_\ell},Y''_{r_\ell},X''_{r_\ell})$), for some $r_\ell\geq 0$, such that:
\begin{enumerate}
\item for every $i=0,\dots,r_\ell$, $Y_i \xrightarrow{\ell} \{Y'_i,Y''_i\}$; 
\item for every $i=1,\dots,r_\ell$, $Y_i=Z_{i-1}=X_{i+1}$;
\end{enumerate}
\end{enumerate}

By construction, $(G,H^*,\vec{B}^*,\hist^*(G))$ produced by encoding the unfolding of the embedding of $G$ on $\T_k$, described in Section~\ref{sec:unfolding-any-Tk}, is locally consistent. The following result shows that the local notion of historical consistency based on the histories fits with the global notion of historical consistency used in Section~\ref{sec:unfolding-any-Tk}. 

\begin{lemma}\label{lem:verification}
Let  $H$ be a splitting of a graph $G$, let $f$ be a planar embedding of $H$, let $\phi$ be a face of $H$ with boundary walk $\vec{B}$ directed clockwise. Let $\hist(G)$ be an history of all the vertices in $G$. If $(G,H,\vec{B},\hist(G))$ is locally consistent, then $(G,H,\vec{B},\mathcal{U})$ is globally consistent, where $\mathcal{U}=U_0,\dots,U_{3k-1}$ is  a sequence of graphs enabling Condition~\ref{consistent:1} of the historical consistency of $(G,H,\vec{B},\hist(G))$ to hold. 
\end{lemma}

\begin{proof}
Thanks to Condition~\ref{consistent:1}, for every $0\leq \ell < 3k-1$, $U_{\ell+1}$ is a degree-2 splitting of $U_\ell$. Moreover, by the consistence of the footprints and the typing in the histories,  the splitting of from $U_\ell$ to $U_{\ell+1}$ is locally consistent at each node of $U_i$ with the duplication of a cycle whenever $\ell\leq k$, and with the duplication of a path otherwise. 

Condition~\ref{consistent:2} in the definition of local consistency guarantees that the footprints at the leaves of the histories are correctly set, that is, they collectively encode the boundary walk~$B$. 

Condition~\ref{consistent:3} guarantees that, for $\ell=1,\dots,2k-1$, starting from $\chi_{2k}=B$, one can iteratively decompose the boundary walk of the face $\chi_{\ell+1}$ of $U_{k+\ell+1}$ into a boundary walk of a face $\psi_{\ell+1}$ of $U_{k+\ell}$, a boundary walk of a face $\chi_{\ell}$ of $U_{k+\ell}$, and the duplication of a path in $U_{k+\ell}$ connecting $\chi_\ell$ to $\psi_{\ell+1}$. It follows that $2k$ faces $\psi_1,\dots,\psi_{2k}$ of $U_{\ell}$ have been identified. Since, the merging of the $2k-1$ paths successively identified in the graphs $U_{k+\ell}$, $\ell=1,\dots,2k-1$ preserves planarity, the graph $U_k$ is planar. 

Moreover, each of the boundary walks of the faces $\psi_1,\dots,\psi_{2k}$ is oriented in a direction inherited from the clockwise orientation of $B$, as guaranteed by the Elementary, Extremity, and Vacancy rules satisfied by the footprints, whose validity are themselves guaranteed by Condition~\ref{consistent:1}.  Condition~\ref{consistent:4} guarantees that the $2k$ faces $\psi_1,\dots,\psi_{2k}$ of $U_k$ can be reordered as $k$ pairs  $(\phi'_i,\phi''_i)$, $i\in\{1,\dots,k\}$ that can be successively merged for creating handles. More specifically, for $i=k,k-1,\dots,1$, Condition~\ref{consistent:4} guarantees that the boundary walks of $\phi'_i$ and $\phi''_i$ are directed such that, by identifying the vertices of $U_i$ that are split of vertices in $U_{i-1}$, a handle is created, resulting in $U_i$ embedded in~$\T_{k-i}$. 
\end{proof}

\subsection{Existence and Unicity of the Paths and Cycles}
\label{subsec:exists-unique}

Our proof-labeling scheme relies on a collection of paths and cycles in the graphs $G^{(0)},\dots, G^{(3k-1)}$. The footprints and types encode these paths and cycles locally. One needs to guarantee the existence and unicity of each path and cycle, in each graph $G^{(i)}$, $i=0,\dots,3k-1$. The next lemma, which is standard, achieve this task. 

\begin{lemma}\label{lem:exists-unique}
Let $G$ be a graph, and let $P$ (resp., $C$) be a (non-necessary simple) directed path (resp., cycle) in $G$. Assume each vertex~$v$ of $P$ (resp., $C$) is given a triple $(\p(v),v,\s(v))$, where $\p(v)$ and $\s(v)$) are the predecessor and  successor of $v$ in $P$ (resp.,~$C$). If $v$ is an extremity of $P$, then $\p(v)=\bot$ or $\s(v)=\bot$, or both $\p(v)=\bot$ and $\s(v)=\bot$ in case $P$ is reduced to~$v$. There exists a proof-labeling scheme with certificates on $O(\log n)$ bits that guarantees the existence and unicity of~$P$. 
\end{lemma}

\begin{proof}
Let $P$ be a directed path in $G$. The proof-labeling scheme uses a spanning tree $T$ of $G$ rooted at the starting vertex $v_0$ of $P$. Every vertex $v$ is given the ID of its parent $p(v)$ in~$T$ ($v_0$ has $p(v_0)=\bot$). The tree $T$ is certified by providing a certificate to every node $v$ containing a pair $(\id(v_0),d(v))$, where $d(v)$ is the distance from $v$ to $v_0$ in~$T$. Every vertex~$v$ checks that it is given the same root-ID as its neighbors in~$G$, and that $d(p(v)=d(v)-1$. Every node that is given one or many triples $(\p(v),v,\s(v))$ checks that, for each of them, $\p(\s(v))=v$ and $\s(\p(v))=v$. (Of course, every such vertex~$v$ also checks consistence of the triples given to it, including the fact that $\p(v)\neq \s(v)$ unless they are both equal to~$\bot$, that it is not given the same successor in two different triples, etc.). If one of the tests is not passed at a vertex, this vertex rejects, otherwise it accepts. The case of a cycle $C$ is treated the same, where the spanning tree $T$ is rooted at any vertex of~$C$. It is easy to check that this standard proof-labeling scheme satisfies both completeness and soundness. 
\end{proof}

\subsection{Verification Procedure}
\label{subsec:verif-procedure}

We now have all ingredients for describing our proof-labeling scheme for $\G^+_k$, $k\geq 0$. 
First, we describe the certificates assigned to the vertices of a graph~$G$ of genus~$k$. The main part of the certificate of~$v$ is the history~$\hist(v)$, as constructed in Section~\ref{sec:histories}. As mentioned in Section~\ref{sec:linePLS}, a history may require more than just $O(\log n)$ bits. However, Lemma~\ref{lem:line-PLS} has shown how to resolve this issue, so that histories can be spread out among the vertices in a way guaranteeing that every vertex stores $O(\log n)$ bits, and, in a single round of communication with its neighbors, every node $v$ can recover its entire history. More importantly even, although a vertex $v$ may not be able to recover the whole history of each of its neighbors in a single round, yet it can recover from each neighbor $w$ the part of $\hist(w)$ corresponding to every edge between an avatar of $v$ and an avatar of $w$, which is sufficient to check the consistency of the neighborhoods, footprints, etc., in all graphs $G^{(0)},\dots,G^{(3k-1)}$ used in the construction.  In addition, the certificate of every vertex is provided with the information enabling to check planarity of $H=G^{(3k-1)}$ (cf. Lemma~\ref{lem:PODC2020}), and to guarantee the existence and unicity of all the directed cycles $C'_i,C''_i$, $i=1,\dots,k$, and all directed paths $P'_j,P''_j$, $j=1,\dots,2k-1$ (cf. Lemma~\ref{lem:exists-unique}). The vertices can then check local consistency, as specified in Section~\ref{subsec:consistent}. Since $G$ has genus~$k$, it follows that, whenever the prover assigns the certificates appropriately, all vertices pass all tests, and therefore all vertices accept. Completeness is therefore satisfied by the scheme. 

Soundness is guaranteed by Lemmas~\ref{lem:certif-Tk-centralized} and~\ref{lem:verification}. Indeed, the latter lemma shows that if the vertices are given certificates that are consistent, and in particular for which the histories are locally consistent, then global consistency is also guaranteed. And the former lemma says that if global consistency is satisfied then the graph can be embedded on~$\T_k$. Therefore, if a graph $G$ cannot be embedded on~$\T_k$, then global consistency cannot be satisfied, which means that the local consistency of the histories cannot be satisfied either, and therefore, at least one vertex of~$G$ fails to pass all tests, and rejects. This completes the proof of Theorem~\ref{theo:main1}.


\section{Proof-Labeling Scheme for Bounded Demigenus Graphs}
\label{sec:PLS-demi-genus}

This section is entirely dedicated to the proof of our second main result. 

\begin{theorem}\label{theo:main2}
Let $k\geq 0$, and let $\G^-_k$ be the class of graphs with demi-genus at most~$k$, i.e., embeddable on a non-orientable closed surface of genus at most~$k$. There is a proof-labeling scheme for $\G^-_k$ using certificates on $O(\log n)$ bits in $n$-node graphs. 
\end{theorem}

The proof-labeling scheme for $\G^-_k$ is based on the same ingredients as the one for $\G^+$ in Theorem~\ref{theo:main1} (e.g., Lemma~\ref{lem:if-demigenusk-then-2cell} is used in replacement of Lemma~\ref{lem:if-genusk-then-2cell}, etc.). However,  new ingredients must be introduced for handling the cross-caps from which non-orientable surfaces result. The proof will thus mainly consist in describing these new ingredients, and in explaining their interactions with the ingredients used for establishing Theorem~\ref{theo:main1}. We start by defining the notion of \emph{doubling} performed on cycles. 
 
 \subsection{Doubling of a Non-Orientable Cycle} 

 Let us assume that we are given an embedding of a graph $G$ on a non-orientable closed surface $X$ of genus $k$,   and let $D= (v_0, v_1, \dots, v_{p-1}, v_p= v_0)$  be a non-orientable cycle of $G$.  Note that a non-orientable cycle is non-separating. The graph $G_{D}$ is obtained by \emph{doubling}~$D$, i.e., by multiplying its length by~2. This doubling of~$D$, and  the canonical embedding of $G_D$ on a closed surface $X_{D} $,  are obtained as follows (see Figure \ref{fig:doubling} for an illustration).

\begin{itemize}
\item   Each vertex $v_i $, $0 \leq i < p$,   is split into two  vertices $v'_i$  and $v'_{p+i}$  in such a way that  $ {D'} = (v'_0, v'_1, v'_2, \dots, v'_{2p -1}, v'_{2p} = v'_0)$  is a cycle of $G_{D}$,  which forms a boundary walk of a face $\phi$ of  $X_{D}$. 

\item   The neighbors of each vertex $v_i$  in $G \setminus {D}$, $0 \leq i < p$,     are shared  between $v'_i$ and $v'_{i+p}$ in $G_{D}$, as follows.  The left and right sides of ${D}$ can be defined locally, i.e.,  in the neighborhood of each (embedded) edge  $\{v_i,  v_{i+1} \}$ of ${D}$. The edges incident to $v'_i$ and $v'_{i+1}$  in $G_{D}$ (and, by symmetry, the edges incident to $v'_{i+p}$ and $v'_{i+p+1}$) correspond to the edges incident to  $v_i$ and $v_{i+1}$ on the same side of  ${D}$ in $G$ according to the local definition of left and right sides in the neighborhood of  $\{v_i,  v_{i+1} \}$.  
\item The vertices $v'_i$ and  $v'_{i+p}$ have no other neighbors. 
\end{itemize}

\begin{figure}[tb] 
   \centering
   \input{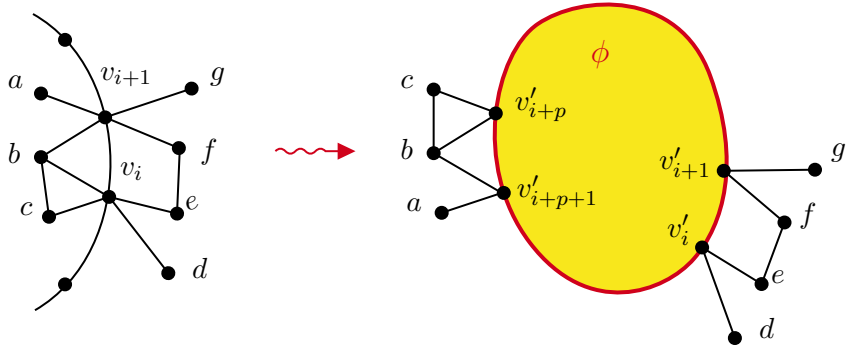}
   \caption{Doubling a non-orientable cycle.}
   \label{fig:doubling}
\end{figure}

\noindent We now show how to unfold $\NT_k$, as we did for unfolding $\T_k$ in the oriented case. 

\subsection{Unfolding $\P_k$ for $k \geq 1$}
\label{subsect:unfolding_Pk}

Let  $G$ be a graph with a 2-cell embedding $f$ on $\P_k$. The unfolding of $G$ has three phases, and only the first one, called \emph{doubling phase}   is new. The second phase is a face-duplication phase, and the third phase  is a face-reduction phase, identical to those described in the case of orientable surfaces.  The  doubling phase is as follows. 
 Let $X^{(0)}=\P_k$, and  let $D_1$ be a non-orientable cycle of $G^{(0)}=G$. Let us consider the embedding of $G^{(1)}=G^{(0)}_{D_1}$  induced by~$f$, on  the  surface $X^{(1)}=X^{(0)}_{D_1}$. There are two cases, both using Lemma~\ref{euler}: 
  \begin{itemize}
\item If $X^{(1)}$ is non-orientable,  then $X^{(1)}$ is homeomorphic to $\P_{k-1}$;
\item Otherwise, $X^{(1)}$ is homeomorphic to $\T_{\frac{k-1}{2}}$.  
\end{itemize}
In the first case, a doubling operation is repeated on  $G^{(1)}$, using a non-orientable cycle $D_2$ of $G^{(1)}$. Doubling operations are performed iteratively until an embedding on an orientable surface  is reached. Formally, there exists a 
a sequence of $m +1$ graphs $G^{(0)}, \dots, G^{(m)}$, $ m \leq k$, 
respectively embedded on closed surfaces $X^{(0)}, \dots, X^{m)}$, such that, for $0 \leq i <m$,   there  exists a non-orientable cycle $D_{i+1}$ of $G^{(i)}$ such that $G^{(i+1)} = (G^{(i)})_{D_{i+1}}$,  and $X^{(i+1)} = (X^{(i)})_{D_{i+1}}$ (up to homeomorphism). 
Necessarily,  for $0 \leq i <m$,  $X^{(i)}  = \P_{k-i}$ (up to homeomorphism),   and   $X^{(m)}  = \T_{(k-m)/2}$, thanks to Lemma~\ref{euler}. 
When $X^{(m)}$ is reached, $G^{(m)}$ contains  $m$ special faces,  whose boundary walks are resulting from the successive doubling of $D_1, \dots, D_m$, respectively. The doubling phase is then completed. 

The face duplication phase starts, initialized with the embedding of $G^{(m)}$ on  $X^{(m)}$.  
Let $k' = \frac{k-m}{2}$. The duplication phase is performed, as in Section  \ref{subsec:faceduplication}. Specifically,  
there exists a sequence of $k' +1$  graphs $G^{(m)}, \dots, G^{(m)+k'}$, respectively embedded  on  closed surfaces $X^{(m)}, \dots, X'^{(m+ k')}$, such that, for $0 \leq i <k'$,   there exists 
a non-separating  cycle $C_{i+1}$ of $G^{(m + i)}$ such that 
$G^{(m + i+1)} = G^{(m + i)}_{C_{i+1}}$, and $X^{(m + i+1)} = X^{(m +i)}_{C_{i+1}}$. 
Necessarily,  for $0 \leq i \leq k'$, $X^{(m + i)}  = \T_{k'-i}$ up to homeomorphism, thanks to Lemma~\ref{euler}. 
In particular, $X^{(m + k')}  = \T_{0}$.  
When $X^{(m+ k')} $  is reached, $G^{(m +k')} $ contains  $2k' + m $ special faces,  whose boundary walks are resulting  from the successive doubling of the cycles  $D_1,\dots,D_m$, and from the duplications of the cycles $C_1, \dots,C_{k'}$. 
At this point,  the face-duplication  phase is completed. 

The face-reduction phase starts,   as in Section  \ref{subsec:facereduction}, in order  to merge the 
$2k' + m  = k $ special faces of $G^{(m +k')}$ into a single face. 
Let us denote the $2k' + m  = k $ special faces of $G^{(m +k')}$ by $\psi_1,\dots,\psi_k$.  Let  $\psi_1 = \chi_1$.  There exists a sequence of paths  $P_1, \dots, P_{k-1}$ such that,  for $1 \leq i \leq k-1$, the duplication of   $P_i$ merges $\chi_{i }$ and  $\psi_{i+1}$ in a single face $\chi_{i+1}$.  A sequence  of planar graphs $G^{(m+ k')},\dots, G^{(m + k'+k-1)}$ results from these merges, where,  for $ 0 \leq i <k-1$,  $P_{i+1}$  is a path of $G^{(m + k'+i)}$, and $ G^{(m + k'+i+1)}  = G^{(m+ k'+i)}_{P_{i+1}}$. For  $1 \leq i \leq k-1$,  $G^{(m + k'+i)}$ has $k  -i$ special faces  $\chi_{i+1} ,  \psi_{i+2} , \dots,\psi_k$. In particular, $G^{(m + k'+k-1)}$ has a unique special face $\chi_{k-1}$. 

To summarize,  as in Section \ref{sec:unfolding-any-Tk}, the embedding $f$ of $G$ in $\P_k$ induces a planar embedding of $H^*= G^{(m + k'+k-1)}$ whose external face is  $\phi^*= \chi_{k-1}$. The boundary of face $\phi^*$ contains all the vertices obtained by splittings resulting from doublings or duplications.

\subsection{Certifying Demigenus at Most~$k$}

Conversely, for a graph $G$ of demigenus~$k$, an embedding of  $G$ in $\P_k$ can be  induced from the embedding $f^*$ of $H^*$ on $\T_0$, and from the boundary walk~$B^*$  of~$\phi^*$. The latter is indeed entirely determined  by the successive cycle-duplications, path-duplications, and cycle doublings performed during the whole process. It contains all duplicated vertices resulting from the cycles $D'_1,\dots ,D'_m$, the cycles $C'_1,\dots ,C'_{k'}$ and $C''_1,\dots ,C''_{k'}$, and from the paths $P'_1, \dots, P'_{k-1}$ and $P''_1, \dots, P''_{k-1}$.  

Now, let $H$ be a splitting of a graph~$G$, let $f$ be a planar embedding of~$H$, and let~$\phi$ be a face of~$H$ embedded on~$\T_0$. Let 
$
B=(u_0,\dots,u_N)
$
be a boundary walk of $\phi$, and let $\vec{B}$ be an arbitrary direction given to~$B$, say clockwise. Let $\mathcal{U}=(U_0,\dots,U_{m + k'+k-1})$, with $m +2  k' = k$ and $m \geq 1$, be a sequence of graphs such that $U_0=G$, $U_{m + k'+k-1}=H$, and, for every $i\in\{0,\dots,m + k'+k-1\}$, $U_{i+1}$ is a 2-splitting of $U_i$. The splitting of $U_i$ into $U_{i+1}$ is denoted by $\sigma_i=(\alpha_i,\beta_i)$. 
The definition of global consistency of $(G,H,\vec{B},\mathcal{U})$, in the case of orientable surfaces,
 can trivially be adapted to the case of non-orientable surfaces by revisiting  conditions~1 and~2,  of  Section~\ref{subsec:global_certificate},  in such a way that the indices correspond to the unfolding of $\P_k$. 
We thus say that  $(G,H,\vec{B},\mathcal{U})$ is \emph{globally consistent} for $\P_k$ if the (revisited) conditions~1 and~2 in  Section \ref{subsec:global_certificate} hold,  plus the following additional condition corresponding to the  doubling phase: 

\begin{itemize}
\item \textbf{Cycle doubling checking. } For every $i=1,\dots,\ell$, there exist  faces $\phi'^{(i)}_1,\phi'^{(i)}_2,\dots, \phi'^{(i)}_i$ of $U_i$ with respective directed boundary walks $\vec{B}(\phi'^{(i)}_1),\vec{B}(\phi'^{(i)}_2),\dots, \vec{B}(\phi'^{(i)}_i)$  such that
\begin{itemize}
\item $\vec{B}(\phi'^{(i)}_{i})=(v'_0,v'_1,\dots,v'_{2p-1}, v'_{2p}= v'_0)$ with, for  $0 \leq j < p$,  $\sigma_{i-1}^{-1}( \{ v'_j, v'_{j+p}\} ) \in V(U_{i-1})$;
\item for $j=1,\dots,i-1$, $\sigma_{i-1}(\vec{B}(\phi'^{(i-1)}_j))=\vec{B}(\phi'^{(i)}_j)$.
\end{itemize}
\end{itemize}

By the construction of Section \ref{subsect:unfolding_Pk}, for every graph $G$ of demigenus~$k$, $(G,H^*,\vec{B}^*,\mathcal{U}^*)$ is globally consistent for $\P_k$, where $\mathcal{U}^*=(G^{(0)},\dots,G^{(m + k'+k-1)})$. The following lemma is the analog to Lemma~\ref{lem:certif-Tk-centralized} for non-orientable surface. Its proof is essentially the same as the proof of Lemma~\ref{lem:certif-Tk-centralized}, in which an argument should be added, for handling cycle doublings, that is, for  identifying opposite vertices of the cycle $D'_i$ in order to creat a cross-cap. The details are omitted. 

\begin{lemma}\label{lem:certif-Pk-centralized}
Let $H$ be a splitting of a graph $G$, and assume that there exists a planar embedding~$f$ of $H$ with a face $\phi$  and a boundary walk $B$ of $\phi$.  Let  $m, k'$ be integers such that $1 \leq m \leq k$ and  $m +2k' = k$, and let $\mathcal{U}=(U_0,\dots,U_{m + k'+k-1})$ be a series of graphs such that $U_0=G$, $U_{m + k'+k-1}=H$, and, for every $i\in\{0,\dots,m + k'+k-2\}$, $U_{i+1}$ is a 2-splitting of $U_i$. If  $(G,H,\vec{B},\mathcal{U})$ is globally consistent for $\P_k$, then $G$ can be embedded on $\P_k$.
\end{lemma}

Thanks to Lemma~\ref{lem:certif-Pk-centralized}, the overall outcome of this section is that the tuple
$
c=(H^*,f^*,\phi^*,B^*,\mathcal{U}^*)
$
constructed in Section~\ref{subsect:unfolding_Pk}  is indeed a certificate that $G$ can be embedded on~$\P_k$. 
\subsection{From Centralized  Certificate to Local Certificate }

The method to distribute  the centralized certificates  uses the same approach and the same tools as those used in Section \ref{sec:tools} in the orientable case.  Only the differences are pointed out in this section. 
In the non-orientable case, the  set of types is
$$S_k=\{D'_1,\dots,D'_\ell,C'_1,\dots,C'_{k'}, C''_1,\dots,C''_{k'},P'_1,\dots,P'_{k-1},P''_1,\dots,P''_{k-1}\}. $$
The footprints and their construction are identical to the orientable case, except that a 
 cross-cap  rule is introduced (see Figure~\ref{fig:crosscap}). 
 
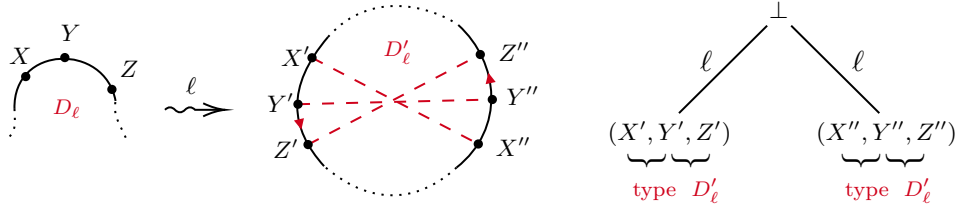
\begin{figure}[tb]
   \centering
   \tikzset{every picture/.style={line width=0.75pt}} 

\begin{tikzpicture}[x=0.75pt,y=0.75pt,yscale=-1,xscale=1]

\draw [color={rgb, 255:red, 208; green, 2; blue, 27 }  ,draw opacity=1 ] [dash pattern={on 4.5pt off 4.5pt}]  (176.85,133.48) -- (262.82,88.18) ;
\draw [color={rgb, 255:red, 208; green, 2; blue, 27 }  ,draw opacity=1 ] [dash pattern={on 4.5pt off 4.5pt}]  (178.82,89.82) -- (261.82,133.18) ;
\draw [color={rgb, 255:red, 208; green, 2; blue, 27 }  ,draw opacity=1 ] [dash pattern={on 4.5pt off 4.5pt}]  (171.82,113.18) -- (268.82,110.82) ;
\draw    (251.83,73.33) .. controls (272.5,92.33) and (275.4,125.6) .. (251,145) ;
\draw  [fill={rgb, 255:red, 0; green, 0; blue, 0 }  ,fill opacity=1 ] (78.93,107.54) .. controls (77.93,107.55) and (77.11,106.74) .. (77.11,105.73) .. controls (77.11,104.73) and (77.92,103.91) .. (78.92,103.91) .. controls (79.92,103.91) and (80.74,104.72) .. (80.74,105.72) .. controls (80.75,106.73) and (79.93,107.54) .. (78.93,107.54) -- cycle ;
\draw  [fill={rgb, 255:red, 0; green, 0; blue, 0 }  ,fill opacity=1 ] (55.43,91.53) .. controls (54.43,91.53) and (53.61,90.72) .. (53.61,89.72) .. controls (53.61,88.72) and (54.42,87.9) .. (55.42,87.9) .. controls (56.43,87.9) and (57.24,88.71) .. (57.24,89.71) .. controls (57.25,90.71) and (56.44,91.53) .. (55.43,91.53) -- cycle ;
\draw  [fill={rgb, 255:red, 0; green, 0; blue, 0 }  ,fill opacity=1 ] (35.91,101.09) .. controls (34.91,101.09) and (34.09,100.28) .. (34.09,99.28) .. controls (34.09,98.27) and (34.9,97.46) .. (35.9,97.46) .. controls (36.9,97.45) and (37.72,98.26) .. (37.72,99.27) .. controls (37.73,100.27) and (36.92,101.09) .. (35.91,101.09) -- cycle ;
\draw    (30.3,115.53) .. controls (29.43,83.19) and (75.42,82.05) .. (80.76,111.54) ;
\draw [color={rgb, 255:red, 0; green, 0; blue, 0 }  ,draw opacity=1 ] [dash pattern={on 0.84pt off 2.51pt}]  (30.3,115.53) .. controls (30.5,123) and (31.5,126) .. (26.5,131) ;
\draw    (105.42,116) .. controls (107.09,114.33) and (108.75,114.33) .. (110.42,116) .. controls (112.09,117.67) and (113.75,117.67) .. (115.42,116) .. controls (117.09,114.33) and (118.75,114.33) .. (120.42,116) -- (125.42,116) -- (133.42,116) ;
\draw [shift={(135.42,116)}, rotate = 180] [color={rgb, 255:red, 0; green, 0; blue, 0 }  ][line width=0.75]    (10.93,-3.29) .. controls (6.95,-1.4) and (3.31,-0.3) .. (0,0) .. controls (3.31,0.3) and (6.95,1.4) .. (10.93,3.29)   ;
\draw  [fill={rgb, 255:red, 0; green, 0; blue, 0 }  ,fill opacity=1 ] (177,89.82) .. controls (177,88.81) and (177.81,88) .. (178.82,88) .. controls (179.82,88) and (180.63,88.81) .. (180.63,89.82) .. controls (180.63,90.82) and (179.82,91.63) .. (178.82,91.63) .. controls (177.81,91.63) and (177,90.82) .. (177,89.82) -- cycle ;
\draw  [fill={rgb, 255:red, 0; green, 0; blue, 0 }  ,fill opacity=1 ] (267,110.82) .. controls (267,109.81) and (267.81,109) .. (268.82,109) .. controls (269.82,109) and (270.63,109.81) .. (270.63,110.82) .. controls (270.63,111.82) and (269.82,112.63) .. (268.82,112.63) .. controls (267.81,112.63) and (267,111.82) .. (267,110.82) -- cycle ;
\draw    (186.83,79.33) .. controls (165.5,104.67) and (167.83,125.33) .. (184.83,144.33) ;
\draw  [dash pattern={on 0.84pt off 2.51pt}]  (251,145) .. controls (234.17,160.67) and (210.5,166.67) .. (184.83,144.33) ;
\draw  [fill={rgb, 255:red, 0; green, 0; blue, 0 }  ,fill opacity=1 ] (175.03,133.48) .. controls (175.03,132.48) and (175.85,131.67) .. (176.85,131.67) .. controls (177.85,131.67) and (178.67,132.48) .. (178.67,133.48) .. controls (178.67,134.49) and (177.85,135.3) .. (176.85,135.3) .. controls (175.85,135.3) and (175.03,134.49) .. (175.03,133.48) -- cycle ;
\draw  [dash pattern={on 0.84pt off 2.51pt}]  (186.83,79.33) .. controls (209.17,54) and (240.5,65.33) .. (251.83,73.33) ;
\draw  [fill={rgb, 255:red, 0; green, 0; blue, 0 }  ,fill opacity=1 ] (260,133.18) .. controls (260,132.18) and (260.81,131.37) .. (261.82,131.37) .. controls (262.82,131.37) and (263.63,132.18) .. (263.63,133.18) .. controls (263.63,134.19) and (262.82,135) .. (261.82,135) .. controls (260.81,135) and (260,134.19) .. (260,133.18) -- cycle ;
\draw  [fill={rgb, 255:red, 0; green, 0; blue, 0 }  ,fill opacity=1 ] (261,88.18) .. controls (261,87.18) and (261.81,86.37) .. (262.82,86.37) .. controls (263.82,86.37) and (264.63,87.18) .. (264.63,88.18) .. controls (264.63,89.19) and (263.82,90) .. (262.82,90) .. controls (261.81,90) and (261,89.19) .. (261,88.18) -- cycle ;
\draw  [fill={rgb, 255:red, 0; green, 0; blue, 0 }  ,fill opacity=1 ] (170,113.18) .. controls (170,112.18) and (170.81,111.37) .. (171.82,111.37) .. controls (172.82,111.37) and (173.63,112.18) .. (173.63,113.18) .. controls (173.63,114.19) and (172.82,115) .. (171.82,115) .. controls (170.81,115) and (170,114.19) .. (170,113.18) -- cycle ;
\draw [color={rgb, 255:red, 0; green, 0; blue, 0 }  ,draw opacity=1 ] [dash pattern={on 0.84pt off 2.51pt}]  (85.83,130) .. controls (82.5,122.33) and (79.5,119.33) .. (80.76,111.54) ;
\draw    (418,73.67) -- (462,117.67) ;
\draw    (406,73.67) -- (362,117.67) ;
\draw [color={rgb, 255:red, 208; green, 2; blue, 27 }  ,draw opacity=1 ]   (172.5,121.67) -- (172.83,123.08) ;
\draw [shift={(173.5,126)}, rotate = 257.01] [fill={rgb, 255:red, 208; green, 2; blue, 27 }  ,fill opacity=1 ][line width=0.08]  [draw opacity=0] (5.36,-2.57) -- (0,0) -- (5.36,2.57) -- cycle    ;
\draw [color={rgb, 255:red, 208; green, 2; blue, 27 }  ,draw opacity=1 ]   (267.5,100.33) -- (267.48,100.26) ;
\draw [shift={(266.83,97.33)}, rotate = 437.47] [fill={rgb, 255:red, 208; green, 2; blue, 27 }  ,fill opacity=1 ][line width=0.08]  [draw opacity=0] (5.36,-2.57) -- (0,0) -- (5.36,2.57) -- cycle    ;

\draw (26.08,83.4) node [anchor=north west][inner sep=0.75pt]  [font=\footnotesize]  {$X$};
\draw (52,71.73) node [anchor=north west][inner sep=0.75pt]  [font=\footnotesize]  {$Y$};
\draw (81.33,89.73) node [anchor=north west][inner sep=0.75pt]  [font=\footnotesize]  {$Z$};
\draw (114.42,98.4) node [anchor=north west][inner sep=0.75pt]  [font=\footnotesize]  {$\ell $};
\draw (213,79.4) node [anchor=north west][inner sep=0.75pt]  [font=\scriptsize,color={rgb, 255:red, 208; green, 2; blue, 27 }  ,opacity=1 ]  {$D'_{\ell }$};
\draw (48,109.4) node [anchor=north west][inner sep=0.75pt]  [font=\scriptsize,color={rgb, 255:red, 208; green, 2; blue, 27 }  ,opacity=1 ]  {$D_{\ell }$};
\draw (162,82.73) node [anchor=north west][inner sep=0.75pt]  [font=\footnotesize]  {$X'$};
\draw (155,106.07) node [anchor=north west][inner sep=0.75pt]  [font=\footnotesize]  {$Y'$};
\draw (158,129.07) node [anchor=north west][inner sep=0.75pt]  [font=\footnotesize]  {$Z'$};
\draw (269.08,128.4) node [anchor=north west][inner sep=0.75pt]  [font=\footnotesize]  {$X''$};
\draw (274.33,102.4) node [anchor=north west][inner sep=0.75pt]  [font=\footnotesize]  {$Y''$};
\draw (270.33,80.4) node [anchor=north west][inner sep=0.75pt]  [font=\footnotesize]  {$Z''$};
\draw (325,120.07) node [anchor=north west][inner sep=0.75pt]  [font=\footnotesize,color={rgb, 255:red, 0; green, 0; blue, 0 }  ,opacity=1 ]  {$( X',Y',Z')$};
\draw (405.5,60.07) node [anchor=north west][inner sep=0.75pt]  [font=\footnotesize,color={rgb, 255:red, 0; green, 0; blue, 0 }  ,opacity=1 ]  {$\perp $};
\draw (429,120.35) node [anchor=north west][inner sep=0.75pt]  [font=\footnotesize,color={rgb, 255:red, 0; green, 0; blue, 0 }  ,opacity=1 ]  {$( X'',Y'',Z'')$};
\draw (379,135) node [anchor=north west][inner sep=0.75pt]  [font=\Large,rotate=-90]  {$\}$};
\draw (357,135) node [anchor=north west][inner sep=0.75pt]  [font=\Large,rotate=-90]  {$\}$};
\draw (464,135) node [anchor=north west][inner sep=0.75pt]  [font=\Large,rotate=-90]  {$\}$};
\draw (486,135) node [anchor=north west][inner sep=0.75pt]  [font=\Large,rotate=-90]  {$\}$};
\draw (337.43,149) node [anchor=north west][inner sep=0.75pt]  [font=\scriptsize,color={rgb, 255:red, 208; green, 2; blue, 27 }  ,opacity=1 ] [align=left] {type $\displaystyle \ D'_{\ell }$};
\draw (443.52,149) node [anchor=north west][inner sep=0.75pt]  [font=\scriptsize,color={rgb, 255:red, 208; green, 2; blue, 27 }  ,opacity=1 ] [align=left] {type $\displaystyle \ D'_{\ell }$};
\draw (371,85) node [anchor=north west][inner sep=0.75pt]    {$\ell $};
\draw (447,85) node [anchor=north west][inner sep=0.75pt]    {$\ell $};

\end{tikzpicture}
   \caption{The cross-cap rule.}
   \label{fig:crosscap}
\end{figure}

 \begin{description}
\item[Cross-cap rule.]   Assuming $X \xrightarrow{\ell} X' ,X'' , \; Y \xrightarrow{\ell} Y', Y'', \; \mbox{and} \; Z \xrightarrow{\ell} Z',Z''$, the cross-cap rule matches two footprints of two children $Y'$ and $Y''$, and produces none at the parent~$Y$: 
\[
(X',Y',Z'), \; (X'',Y'',Z'') \xrightarrow{\ell}  \bot.
\]
The cross-cap rule applies to the case of identifying opposite vertices of the boundary of a face, in the reverse operation of doubling.   The corresponding face disappears, and their boundaries can be discarded. 
\end{description}

The assignments  of types to footprints is performed in the same as in Section~\ref{sec:tools}, and the same distributed algorithm is used for checking the planarity of $H$. An important  difference with the orientable case appears in  the definition  of the local consistency of distributed certificates (previously defined in Section~\ref{subsec:consistent}). Again,  an additional condition is introduced, for reflecting the creation of cross-caps. 
\begin{itemize}
\item  For every $\ell=1,\dots,m$, the collection of footprints at the nodes at level $\ell$ whose both edges have type $D'_{\ell}$  in the histories in~$\hist(G)$ can be ordered as $(X'_0,Y'_0,Z'_0),\dots,$ $(X'_{2r_\ell-1},Y'_{2r_\ell-1},Z'_{2r_\ell-1})$, for some $r_\ell\geq 1$, such that:
\begin{enumerate}
\item for every $i=0,\dots,r_\ell-1$, $Y_i \xrightarrow{\ell} \{Y'_i,Y''_{i + r_\ell}\}$; 
\item for every $i=0,\dots,2r_\ell-1$, $Y_i=Z_{i-1}=X_{i+1}$ (where  indices are taken  modulo $2r_\ell$); 
\end{enumerate}
\end{itemize}
The following lemma is the analog of Lemma \ref{lem:verification}, but for non-orientable surfaces. 
Its proof is identical to the proof of Lemma \ref{lem:verification}, with an additional argument, stating that the conditions added for handling non-orientable surfaces enable opposite vertices of  the face surrounded by $D'_\ell$  in $U_\ell$,  $1\leq \ell \leq 2k-1$,  to be  identified for creating a cross-cap in  $U_{\ell-1}$. 

\begin{lemma}\label{lem:verif_Pk}
Let  $H$ be a splitting of a graph $G$, let $f$ be a planar embedding of $H$, let $\phi$ be a face of $H$ with boundary walk $\vec{B}$ directed clockwise. Let $\hist(G)$ be an history of all the vertices in $G$. If $(G,H,\vec{B},\hist(G))$ is locally consistent, then $(G,H,\vec{B},\mathcal{U})$ is globally consistent, where $\mathcal{U}=U_0,\dots,U_{m+ k'+ k-1}$ is  a sequence of graphs enabling the global consistency of $(G,H,\vec{B},\hist(G))$ to hold. 
\end{lemma}

\subsection{Verification Procedure}

The verification procedure is similar to the one described in Section~\ref{subsec:verif-procedure}, and is therefore omitted.


\section{Conclusion}
\label{sec:conclusion}

In this paper, we have designed proof-labeling schemes for the class of graphs of bounded genus, as well as for the class of graphs with bounded demigenus. All our schemes use certificates on $O(\log n)$ bits, which is optimal, as it is known that even certifying the class of planar graphs requires proof-labeling schemes with certificates on $\Omega(\log n)$ bits~\cite{FeuilloleyFRRMT}. The existence of ``compact'' proof-labeling schemes (i.e., schemes using certificates of polylogarithmic size) for other classes of sparse graphs is still not known. In particular, proving or disproving the existence of such a scheme for $H$-minor-free graphs appears to be a challenging problem. Indeed, Robertson and Seymour's decomposition theorem states that every $H$-minor-free graph can be expressed as a tree structure of ``pieces'', where each piece is a graph that can be embedded in a surface on which $H$ cannot be embedded, plus a bounded number of so-called \emph{apex} vertices, and a bounded number of so-called \emph{vortex} subgraphs. The  decomposition theorem provides a powerful tool for the design of (centralized or distributed) algorithms. However, this theorem is not a characterization, that is, there are graphs that are not $H$-minor-free, and yet can be expressed as a tree structure satisfying the required properties (surfaces of bounded genus, bounded number of apices, bounded number of vortices, etc.).  It follows that, although Robertson and Seymour's decomposition theorem should most probably play a crucial role for designing a compact proof-labeling scheme for $H$-minor-free graphs (if such a scheme exists), this development may require identifying additional properties satisfied by these graphs.


%

\bigbreak 

\paragraph{Acknowledgements.} The first and fifth authors are thankful  to Gelasio Salazar for his very detailed answers to their questions about closed surfaces.

\bibliographystyle{plain}
\bibliography{pls-genus}



\end{document}